\pdfoutput=1
\documentclass[11pt,a4paper]{article}

\usepackage[utf8]{inputenc}
\usepackage[T1]{fontenc}
\usepackage[english]{babel}
\usepackage{lmodern}
\usepackage{amsmath,amssymb,amsthm,mathtools}
\usepackage{bm}
\usepackage{graphicx}
\usepackage[margin=2.6cm]{geometry}
\usepackage{booktabs}
\usepackage{array}
\usepackage{enumitem}
\usepackage{caption}
\usepackage{subcaption}
\usepackage{xcolor}
\usepackage[colorlinks=true,linkcolor=blue!55!black,citecolor=green!45!black,%
            urlcolor=blue!60!black]{hyperref}
\usepackage{cleveref}

\graphicspath{{figures/}}

\theoremstyle{plain}
\newtheorem{teorema}{Theorem}[section]
\newtheorem{proposicion}[teorema]{Proposition}
\newtheorem{lema}[teorema]{Lemma}
\newtheorem{corolario}[teorema]{Corollary}

\theoremstyle{definition}
\newtheorem{definicion}[teorema]{Definition}

\newtheorem{hipotesis}[teorema]{Assumption}

\theoremstyle{remark}
\newtheorem{observacion}[teorema]{Remark}

\crefname{teorema}{theorem}{theorems}
\Crefname{teorema}{Theorem}{Theorems}
\crefname{proposicion}{proposition}{propositions}
\Crefname{proposicion}{Proposition}{Propositions}
\crefname{lema}{lemma}{lemmas}
\Crefname{lema}{Lemma}{Lemmas}
\crefname{corolario}{corollary}{corollaries}
\Crefname{corolario}{Corollary}{Corollaries}
\crefname{definicion}{definition}{definitions}
\Crefname{definicion}{Definition}{Definitions}
\crefname{ejemplo}{example}{examples}
\Crefname{ejemplo}{Example}{Examples}
\crefname{hipotesis}{assumption}{assumptions}
\Crefname{hipotesis}{Assumption}{Assumptions}
\crefname{observacion}{remark}{remarks}
\Crefname{observacion}{Remark}{Remarks}

\crefname{table}{table}{tables}
\Crefname{table}{Table}{Tables}
\crefname{section}{section}{sections}
\Crefname{section}{Section}{Sections}
\crefname{subsection}{subsection}{subsections}
\Crefname{subsection}{Subsection}{Subsections}
\crefname{figure}{figure}{figures}
\Crefname{figure}{Figure}{Figures}
\crefname{equation}{equation}{equations}
\Crefname{equation}{Equation}{Equations}

\newcommand{\Bt}{\mathbf{B}_t}
\newcommand{\Bx}{\mathbf{B}_x}
\newcommand{\Btx}{\mathbf{B}_{tx}}

\newcommand{\st}{s_t}
\newcommand{\sx}{s_x}

\newcommand{\LG}{\mathcal{L}_{\mathrm{G}}}
\newcommand{\transf}[1]{\LG\!\left\{#1\right\}}

\newcommand{\Deter}{\Delta}
\newcommand{\Deterc}{\Delta_c}
\newcommand{\Deterd}{\Delta_d}

\newcommand{\Zimp}{Z}
\newcommand{\Yadm}{Y}

\newcommand{\Real}{\operatorname{Re}}
\newcommand{\Imag}{\operatorname{Im}}
\newcommand{\ArgBt}{\operatorname{Arg}_{\Bt}}
\newcommand{\ArgBx}{\operatorname{Arg}_{\Bx}}
\newcommand{\RR}{\mathbb{R}}
\newcommand{\CC}{\mathbb{C}}
\newcommand{\Cl}{\mathrm{Cl}}

\newcommand{\Res}{\mathcal{R}}

\newcommand{\pt}{\partial_t}
\newcommand{\px}{\partial_x}
\newcommand{\dd}{\,\mathrm{d}}

\newcommand{\eidem}{\mathsf{e}}

\newcommand{\unit}[1]{\,\mathrm{#1}}

\title{\textbf{Multivariable Geometric Laplace Transform\\
and Fault Detection in Distributed-Converter Lines}}
\author{Francisco Manuel Arrabal-Campos$^{\,1,*}$,
  Francisco G. Montoya$^{\,1}$,
  Santiago Sánchez-Acevedo$^{\,2}$,\\
  Raymundo E. Torres-Olguin$^{\,2}$
  and Alfredo Alcayde$^{\,1}$\\[8pt]
  \small $^{1}$\,Department of Engineering, University of Almeria,
  La Cañada de San Urbano, 04120 Almeria, Spain\\
  \small \texttt{fmarrabal@ual.es}, \texttt{pagilm@ual.es},
  \texttt{aalcayde@ual.es}\\[3pt]
  \small $^{2}$\,SINTEF, Norway\\
  \small \texttt{santiago.sanchez@sintef.no},
  \texttt{raymundo.torres-olguin@sintef.no}\\[3pt]
  \small $^{*}$\,Corresponding author: \texttt{fmarrabal@ual.es}}
\date{\today}

\begin{document}

\maketitle

\begin{abstract}
Monitoring a DC line with many distributed power converters is a genuinely
spatio-temporal problem: the information about a localized fault travels along
the whole conductor and reaches a few measurement points mixed with the
dynamics of the line itself. This paper develops a two-dimensional
\emph{geometric Laplace transform} $(t,x)\mapsto(\st,\sx)$ over a commutative
subalgebra of $\Cl(4,0)$, isomorphic to Segre's bicomplex numbers, in which
two bivectors $\Bt$ and $\Bx$ act as independent imaginary units for the
temporal and the spatial phase. Because the two phases live in algebraically
distinguishable planes, a fault at position $x_f$ leaves a transformed
residual that factorizes as $F_f(\st)\,e^{-\sx x_f}$: its temporal nature
stays in the first factor and its location can be read as a geometric
argument of the second. On this representation we build a transmission-line
model of the converter line and its space--time dispersion relation, a
distributed control by admittance shaping ---with an exact treatment of
discrete converter sites: spatial sampling, aliasing, and a per-converter
droop realization that is exact on the sub-Nyquist band---, and a fault
diagnosis chain that detects, localizes and classifies injection-loss, shunt,
sensor and local-controller faults, extends to multiple simultaneous faults
with automatic order selection, and distinguishes the outage of a plant from
a cable defect. As an integral object the transform is known in bicomplex
analysis, and with a single independent variable it reduces to the complex
Laplace transform; the contribution lies in its geometric embedding and in
its operational use for fault diagnosis in distributed-converter networks.
All results are reproduced by an accompanying open implementation.
\end{abstract}

\tableofcontents
\bigskip

\section{Introduction and motivation}
\label{sec:intro}

\subsection{Applied context: DC lines with distributed converters}

DC microgrids and DC distribution buses have taken on a central role in the
integration of renewable generation, storage, and electronic loads. In these
architectures, numerous power converters are connected along a distribution
line, occupying distinct physical positions, and they interact through the
electromagnetic dynamics of the conductor itself. The ensemble therefore
behaves as a \emph{spatio-temporal} system: the relevant electrical quantities
(line voltage and current) depend simultaneously on time~$t$ and on the spatial
coordinate~$x$ along the line
\cite{dragicevic2016dc,paul2007transmission}.

In this setting, a fundamental operational objective is monitoring: detecting
the onset of a fault, \emph{locating} it (identifying the position~$x_f$ of the
affected converter or segment), and \emph{classifying} it (distinguishing, for
example, an injection loss from a shunt short circuit or a sensor fault). The
difficulty lies in the fact that the information about a localized fault
propagates along the entire line and is observed, intermingled, at a limited
number of measurement points.

Throughout this work we adopt a set of modeling assumptions that precisely
delimit the scope of the analysis.

\begin{hipotesis}[Incremental linearization]
\label{hyp:linealizacion}
Around an admissible steady-state operating point, the voltage and current
deviations from that point are small enough for the dynamics of the line and of
the converters to admit an incremental linear description. All quantities that
follow are understood as deviations from the steady-state regime.
\end{hipotesis}

\begin{hipotesis}[Spatial continuity]
\label{hyp:continuidad}
The distribution line is modeled as a spatially continuous medium in the
coordinate~$x$, described by distributed parameters (resistance, inductance,
conductance, and capacitance per unit length). The localized distributions
associated with the converters are interpreted in the sense of densities over
that continuum.
\end{hipotesis}

\begin{hipotesis}[Converters as distributed injection]
\label{hyp:inyeccion}
Each converter located at a position~$x_f$ is approximated, in the incremental
regime, by its contribution to a distributed injection $u(t,x)$ that acts on
the line. A converter concentrated at~$x_f$ corresponds to the limiting case in
which $u(t,x)$ is concentrated around~$x_f$, whereas a density of nearby
converters is represented by a smooth injection in~$x$. The exact discrete
case ---$P$ point injections at known sites--- and its continuum limit are
treated in \cref{subsec:inyeccion-discreta}.
\end{hipotesis}

\Cref{hyp:linealizacion,hyp:continuidad,hyp:inyeccion} set the regime of
validity of the model: small signal, continuous medium, and distributed
actuation. We do not aim to describe large-signal transients or the internal
switching of the converters, but rather the incremental line dynamics on which
the monitoring is built.

\subsection{Mathematical motivation: why a single complex variable does not suffice}

The classical Laplace transform associates with a time signal $f(t)$ the
function
\begin{equation}
\label{eq:laplace-clasica}
  F(s) = \int_0^\infty f(t)\, e^{-s t}\dd t,
  \qquad s = \sigma + j\,\omega,\quad j^2 = -1 ,
\end{equation}
where the imaginary unit~$j$ encodes the phase of the oscillatory components.
The choice of~$j$, however, is not essential: it can be replaced by any
element~$B$ of an algebra satisfying $B^2=-1$. If $B$ is a \emph{bivector} of a geometric algebra, the resulting
calculus is \emph{algebraically equivalent} to that of
\cref{eq:laplace-clasica}, since the entire structure rests on the single
relation $B^2=-1$ \cite{hestenes1984clifford}. With a single independent
variable, therefore, replacing~$j$ by a bivector provides a \emph{geometric
interpretation} of the phase, but no new modeling capability.

The situation changes qualitatively when there are \emph{two} independent
variables, time~$t$ and space~$x$. In a spatio-temporal system a temporal phase
(associated with the frequency~$\omega$) and a spatial phase (associated with
the wavenumber~$k$) coexist. To be precise about the scope: a
complex Laplace transform with \emph{two} independent variables
$(s_t,s_x)\in\CC^2$ already distinguishes the two phases, since each variable
retains its own coordinate. What we pursue is not that distinction,
but rather \emph{representing both phases in algebraically distinguishable
planes within a single structure}, which makes it easier to read arguments,
mixed products, and spatial signatures on a single object. To this end we
introduce \emph{two} independent geometric imaginary units, the bivectors $\Bt$
and $\Bx$, with
\begin{equation}
\label{eq:relaciones-base}
  \Bt^2 = -1, \qquad \Bx^2 = -1, \qquad
  \Bt\Bx = \Bx\Bt = \Btx, \qquad \Btx^2 = +1 ,
\end{equation}
so that the algebra generated is commutative and isomorphic to the algebra of
Segre's bicomplex numbers \cite{luna2015bicomplex}. A general element is written
$z = a + b\,\Bt + c\,\Bx + d\,\Btx$ with $a,b,c,d\in\RR$. Over this structure we
define the two-dimensional geometric Laplace transform
\begin{equation}
\label{eq:transformada-intro}
  \transf{f} = F(\st,\sx)
  = \int_0^\infty\!\!\int_0^\infty f(t,x)\,
      e^{-\st t}\, e^{-\sx x}\dd x\dd t,
\end{equation}
with geometric variables
\begin{equation}
\label{eq:variables-intro}
  \st = \rho_t + \Bt\,\omega, \qquad
  \sx = \rho_x + \Bx\,k .
\end{equation}
The factor $e^{-\st t}$ carries the temporal attenuation and phase in the plane
generated by~$\Bt$, while $e^{-\sx x}$ carries the spatial attenuation and phase
in the plane generated by~$\Bx$. Since $\Bt$ and $\Bx$ are \emph{distinct}
imaginary units, the two phases are not conflated: the temporal information is
labeled by~$\Bt$ and the spatial information by~$\Bx$.

It is this separation that gives the proposal its applied meaning. The position
of a fault along the line manifests itself as a well-defined spatial phase
shift, which can be read as an argument in the~$\Bx$ plane through the operator
$\ArgBx$, in a manner analogous to how the temporal phase is read with
$\ArgBt$. The spatial location of a fault thus becomes the reading of a
geometric argument, separable from the temporal dynamics of the fault itself.

\subsection{Related work and delimitation of the novelty}
\label{subsec:trabajo-relacionado}

Three strands of literature intersect this work. We review them briefly, to
acknowledge what is already known and to make clear where our contribution
lies.

\paragraph{Bicomplex and double Laplace transforms.}
The integral object \cref{eq:transformada-intro} is not, in itself, new. The
Laplace transform with a bicomplex variable is established in the bicomplex
analysis literature: \cite{kumar2011bicomplex} studies its existence and its
region of convergence via the idempotent projections, and
\cite{agarwal2014convolution} develops its convolution theorem and
applications; the foundations of the underlying algebra are collected in
\cite{price1991bicomplex,luna2015bicomplex}. The \emph{double} Laplace transform
in two complex variables, for its part, is classical and possesses a complete
body of operational properties \cite{debnath2016double}. Consequently, in this
work \emph{no novelty is claimed for the transform as a mathematical object}:
the properties in \cref{sec:transformada} should be understood as a
self-contained reworking, in the language of geometric algebra, of results that
are for the most part known within those frameworks.

\paragraph{Hypercomplex transforms in signal processing.}
In signal and image processing there is a well-established tradition of
transforms over hypercomplex algebras: the quaternion Fourier transform,
introduced precisely to analyze two-dimensional linear time-invariant systems
described by partial differential equations \cite{ell1993quaternion} and later
developed for color imaging \cite{ell2007hypercomplex}; the hypercomplex
analytic signal \cite{bulow2001hypercomplex}; and the Clifford Fourier
transforms \cite{hitzer2013quaternion}. These frameworks are, in general,
\emph{noncommutative} (quaternions, full Clifford algebras) and their dominant
motivation is image analysis. On the Laplace side, a quaternion transform of
two real variables appears already in Ell's thesis
\cite{ell1992hypercomplex} and, very recently, a \emph{geometric Laplace
transform} has been introduced for multivector-valued functions of a
\emph{single} real variable over geometric algebras $\Cl(q,r)$ with
$q+r\le5$ and a multivector Laplace variable \cite{velasco2026geometric};
there the noncommutativity forces left and right transforms, one-sided
operational rules, and an exponential transform that departs from the
classical one. Despite the shared name, that object and the one developed
here differ in both directions: ours is genuinely two-variable in $(t,x)$,
and its codomain is a commutative subalgebra chosen precisely so that the
classical operational calculus survives. The present work thus uses the
\emph{commutative} (bicomplex) variant embedded as a subalgebra of $\Cl(4,0)$,
in which the $\Bt$ and $\Bx$ planes are algebraically distinguishable and the
Laplace-type operational properties are preserved without the ordering
complications of noncommutative products, and directs it at a concrete
spatio-temporal physical system $(t,x)$.

\paragraph{Fault location in DC lines and networks.}
Protection and fault location in DC lines and microgrids rely on
well-established techniques, most of them conceived for the large-signal
transient of a short circuit: traveling-wave methods
\cite{nanayakkara2012traveling}, fast differential protection
\cite{fletcher2014highspeed}, location through test power injection in ring
buses \cite{park2013dcring}, and active impedance estimation
\cite{christopher2013fault}; see also the surveys of DC microgrid control
\cite{dragicevic2016dc,meng2017dynamics}. These techniques typically require
pronounced transients or dedicated injection infrastructure. The niche
addressed by this work is different and complementary: the \emph{incremental
small-signal monitoring} (\cref{hyp:linealizacion}) of a network of generators
connected through converters distributed along the line, with detection,
location, and classification integrated within a single model-based framework
\cite{ding2008model,isermann2006fault}, oriented toward incipient faults
(injection loss, admittance drifts, sensor faults, or local-control faults)
that do not go so far as to trip the large-signal protections.

\paragraph{Delimitation of the novelty.}
In light of the above, the contribution of this work is not the bicomplex
transform, but rather its \emph{operational exploitation for fault diagnosis in
a distributed network of generators with converters}: (i)~the geometric
embedding of the transform in $\Cl(4,0)$, which endows the two phases with
distinguishable bivector planes and allows the spatial argument operator
$\ArgBx$ to be defined; (ii)~the representation of a localized fault as a
spatial singularity of the distributed residual, factorizable as
$F_f(\st)\,e^{-\sx x_f}$, whose position is read as a geometric argument on the
same transformed object that encodes the temporal dynamics; and (iii)~the
resulting integrated framework of detection, location, and classification, with
reproducible numerical validation on a model of a DC line with distributed
converters. To the best of our knowledge, this combination had not been
formulated or exploited for fault diagnosis in networks of distributed
converters.

\subsection{Contributions}

In summary, and with the boundaries drawn in
\cref{subsec:trabajo-relacionado}, this work contributes the following.

\begin{enumerate}[leftmargin=1.6em,itemsep=2pt]
  \item An integrated framework for \emph{fault detection, location, and
        classification} in a network of generators with distributed converters
        on a DC line, in which localized faults are represented as
        \emph{spatial} singularities of the distributed residual, factorizable
        as $F_f(\st)\,e^{-\sx x_f}$ in the transformed domain, and the position
        of the fault is read as a geometric argument in the $\Bx$ plane
        (operator $\ArgBx$) or through modal weights. This is the main
        contribution of the work.
  \item A geometric formulation of the two-dimensional bicomplex Laplace
        transform, \cref{eq:transformada-intro}, over the commutative subalgebra
        of $\Cl(4,0)$ generated by $\Bt$ and $\Bx$, with a self-contained
        exposition of its operational properties and its idempotent
        decomposition. The transform as an integral object is known
        \cite{kumar2011bicomplex,agarwal2014convolution,debnath2016double}; the
        contribution is its geometric embedding and the operational reading of
        the spatial phase that this enables.
  \item A spatio-temporal model of a DC line with distributed converters and its
        dispersion relation $\Deter(\st,\sx)=0$ in the geometric domain, which
        separates the temporal and spatial dynamics, together with the exact
        spatial-sampling theory of discrete converter sites (orthogonality,
        aliasing, controllability).
  \item A distributed control design by admittance shaping formulated in the
        geometric variables $(\st,\sx)$.
  \item A reproducible implementation and a numerical validation of the
        preceding methods.
\end{enumerate}

\subsection{Document structure}

The remainder of the document is organized as follows. \Cref{sec:algebra}
introduces the commutative geometric algebra generated by $\Bt$ and $\Bx$ and
its isomorphism with the bicomplex numbers. \Cref{sec:transformada} defines the
geometric Laplace transform \cref{eq:transformada-intro} and establishes its
properties. \Cref{sec:linea} develops the model of a DC line with distributed
converters and its spatio-temporal dispersion relation. \Cref{sec:control}
addresses the distributed control design by admittance shaping. \Cref{sec:fallos}
formulates fault detection, location, and classification from the distributed
residual. \Cref{sec:resultados} presents the numerical validation and, finally,
\cref{sec:conclusiones} gathers the conclusions and the lines of future work,
among them the extension to the case of noncommutative bivectors.

\begin{observacion}[Scope and limitations]
\label{obs:alcance}
The geometric separation of phases facilitates the reading of the position of a
fault, but it does not eliminate the intrinsic limitations of any spatial
monitoring problem. The effective capability for detection, location, and
classification depends on the \emph{spatial observability} of the system, on the
number and placement of the sensors, on the measurement noise level, on the
boundary conditions imposed on the line, and on the quality of the incremental
model adopted in
\cref{hyp:linealizacion,hyp:continuidad,hyp:inyeccion}. Consequently, the
results that follow should not be understood as a universal guarantee of
diagnosis, but rather as a framework whose performance remains conditioned by
those factors, in line with the usual practice in model-based diagnosis
\cite{ding2008model}.
\end{observacion}

\section{Minimal geometric algebra}
\label{sec:algebra}

The development of the space--time transform in the sections that follow
rests on a four-real-dimensional commutative subalgebra, generated by two
commuting bivectors within the even part of a Clifford algebra (it is a
\emph{proper} subalgebra of the even subalgebra, which has higher
dimension). This section defines it precisely, establishes its calculation
rules (product, inverse, and exponential), and fixes the notation used
throughout the document. The exposition follows the style of geometric
calculus \cite{hestenes1984clifford,doran2003geometric}, but is
deliberately restricted to what is strictly necessary: two commuting
bivectors and their product. As will be made explicit, the resulting
structure is isomorphic to the algebra of Segre's bicomplex numbers
\cite{price1991bicomplex,luna2015bicomplex}, which allows the classical
results on diagonalization and bicomplex analysis to be carried over to our
geometric framework.

\subsection{Definition of the subalgebra}

\begin{definicion}[Commutative bivector subalgebra]
\label{def:subalgebra}
Let $\Cl(4,0)$ be the geometric algebra of the Euclidean space $\RR^4$ with
orthonormal basis $\{e_1,e_2,e_3,e_4\}$ and geometric product determined by
$e_ie_j+e_je_i=2\delta_{ij}$. Denote the bivectors
\begin{equation}
\Bt \coloneqq \sigma_{12}=e_1e_2,\qquad
\Bx \coloneqq \sigma_{34}=e_3e_4,\qquad
\Btx \coloneqq \Bt\,\Bx = \sigma_{1234}=e_1e_2e_3e_4 .
\label{eq:gen-def}
\end{equation}
We define $\mathcal{A}$ as the real subalgebra of $\Cl(4,0)$ generated by
$\{1,\Bt,\Bx,\Btx\}$. Its elements are written uniquely as
\begin{equation}
z = a + b\,\Bt + c\,\Bx + d\,\Btx,
\qquad (a,b,c,d)\in\RR^4 ,
\label{eq:elemento}
\end{equation}
and we shall identify $z$ with the real coordinate vector $[a,b,c,d]$. The
generators satisfy
\begin{equation}
\Bt^2=\Bx^2=-1,\qquad
\Bt\Bx=\Bx\Bt=\Btx,\qquad
\Btx^2=+1 .
\label{eq:reglas}
\end{equation}
\end{definicion}

The relations \eqref{eq:reglas} are not postulated but follow directly from
\eqref{eq:gen-def} and from the anticommutativity of orthogonal vectors.
Indeed, for distinct orthonormal vectors $e_ie_j=-e_je_i$ and $e_i^2=1$, so
that $\Bt^2=e_1e_2e_1e_2=-e_1e_1e_2e_2=-1$, and identically $\Bx^2=-1$. Since
$\{e_1,e_2\}$ and $\{e_3,e_4\}$ are disjoint blocks, each factor of $\Bt$
commutes with each factor of $\Bx$ after an even number of transpositions,
hence $\Bt\Bx=\Bx\Bt$. Finally,
$\Btx^2=(\Bt\Bx)(\Bt\Bx)=\Bt^2\Bx^2=(-1)(-1)=+1$. The choice of disjoint
orthogonal blocks $\{1,2\}$ and $\{3,4\}$ is precisely what guarantees
commutativity and, with it, all of the symbolic calculus that follows.

\begin{observacion}
\label{obs:cierre}
The set $\{1,\Bt,\Bx,\Btx\}$ is closed under the geometric product up to
sign (see \cref{prop:producto}); therefore $\mathcal{A}$ is indeed a real
$4$-dimensional subalgebra, associative and unital. Moreover, since all of
its generators commute pairwise, $\mathcal{A}$ is \emph{commutative}, in
contrast with the ambient geometric algebra $\Cl(4,0)$, which is not.
\end{observacion}

\subsection{Product and commutativity}

\begin{proposicion}[Product formula]
\label{prop:producto}
Let $z=a+b\,\Bt+c\,\Bx+d\,\Btx$ and $w=p+q\,\Bt+r\,\Bx+s\,\Btx$ be elements
of $\mathcal{A}$. Then their product $zw\in\mathcal{A}$ has components
\begin{align}
(zw)_{1}   &= ap - bq - cr + ds, \label{eq:prod-1}\\
(zw)_{\Bt} &= aq + bp - cs - dr, \label{eq:prod-Bt}\\
(zw)_{\Bx} &= ar - bs + cp - dq, \label{eq:prod-Bx}\\
(zw)_{\Btx}&= as + br + cq + dp. \label{eq:prod-Btx}
\end{align}
In particular $zw=wz$: the algebra $\mathcal{A}$ is commutative.
\end{proposicion}

\begin{proof}
By \cref{obs:cierre} it suffices to know the multiplication table of the
generators. From \eqref{eq:reglas} one obtains
\[
\begin{array}{c|cccc}
\cdot & 1 & \Bt & \Bx & \Btx\\\hline
1 & 1 & \Bt & \Bx & \Btx\\
\Bt & \Bt & -1 & \Btx & -\Bx\\
\Bx & \Bx & \Btx & -1 & -\Bt\\
\Btx & \Btx & -\Bx & -\Bt & 1
\end{array}
\]
where, for example, $\Bt\,\Btx=\Bt(\Bt\Bx)=\Bt^2\Bx=-\Bx$ and
$\Btx\,\Btx=+1$. The table is symmetric, which already shows that the
generators commute. Expanding the product $zw$ by bilinearity,
\[
zw=\sum_{u,v\in\{1,\Bt,\Bx,\Btx\}} (\text{coef. of }u\text{ in }z)\,
(\text{coef. of }v\text{ in }w)\,(uv),
\]
and grouping the sixteen terms according to the resulting generator by means
of the table above yields exactly the identities
\eqref{eq:prod-1}--\eqref{eq:prod-Btx}. The symmetry of each expression under
the interchange $(a,b,c,d)\leftrightarrow(p,q,r,s)$ is manifest:
\eqref{eq:prod-1} and \eqref{eq:prod-Btx} are symmetric term by term, and the
pair \eqref{eq:prod-Bt}--\eqref{eq:prod-Bx} is mapped into itself. Hence
$zw=wz$.
\end{proof}

\begin{observacion}[Bicomplex isomorphism]
\label{obs:bicomplejo}
The map sending $\Bt\mapsto i$, $\Bx\mapsto j$, $\Btx\mapsto ij$, with
$i^2=j^2=-1$, $ij=ji$, is an $\RR$-algebra isomorphism between $\mathcal{A}$
and the algebra of Segre's \emph{bicomplex numbers}
\cite{price1991bicomplex,luna2015bicomplex}. This correspondence is useful
because it allows the bicomplex theory of diagonalization and of holomorphic
functions to be imported into our geometric framework, without giving up the
geometric reading of $\Bt$, $\Bx$, and $\Btx$ as bivectors of orthogonal
planes.
\end{observacion}

\subsection{Idempotent decomposition and diagonalization}

Commutativity and the presence of an element $\Btx$ whose square is $+1$
make it possible to diagonalize $\mathcal{A}$ by means of two orthogonal
idempotents.

\begin{proposicion}[Idempotent diagonalization]
\label{prop:idempotente}
We define
\begin{equation}
\eidem_1=\tfrac12\bigl(1+\Btx\bigr),\qquad
\eidem_2=\tfrac12\bigl(1-\Btx\bigr).
\label{eq:idempotentes}
\end{equation}
Then
\begin{equation}
\eidem_1^2=\eidem_1,\qquad
\eidem_2^2=\eidem_2,\qquad
\eidem_1\eidem_2=\eidem_2\eidem_1=0,\qquad
\eidem_1+\eidem_2=1 .
\label{eq:idem-props}
\end{equation}
Moreover, every $z=a+b\,\Bt+c\,\Bx+d\,\Btx\in\mathcal{A}$ decomposes uniquely
as
\begin{equation}
z=\zeta_1\,\eidem_1+\zeta_2\,\eidem_2,
\qquad
\zeta_1=(a+d)+(b-c)\,\Bt,\quad
\zeta_2=(a-d)+(b+c)\,\Bt,
\label{eq:desc-zeta}
\end{equation}
where $\zeta_1,\zeta_2$ are ordinary complex numbers in the field
$\CC_{\Bt}\coloneqq\RR\oplus\RR\,\Bt\cong\CC$ generated by the imaginary unit
$\Bt$. In this representation the product, the inverse, and the exponential
act componentwise:
\begin{equation}
zw=(\zeta_1\eta_1)\eidem_1+(\zeta_2\eta_2)\eidem_2,
\qquad
w=\eta_1\eidem_1+\eta_2\eidem_2 .
\label{eq:prod-diag}
\end{equation}
\end{proposicion}

\begin{proof}
For \eqref{eq:idem-props}, using $\Btx^2=+1$ from \eqref{eq:reglas},
\[
\eidem_1^2=\tfrac14(1+\Btx)^2=\tfrac14(1+2\Btx+\Btx^2)
=\tfrac14(2+2\Btx)=\tfrac12(1+\Btx)=\eidem_1,
\]
and analogously $\eidem_2^2=\eidem_2$. The cross product is
\[
\eidem_1\eidem_2=\tfrac14(1+\Btx)(1-\Btx)=\tfrac14(1-\Btx^2)
=\tfrac14(1-1)=0,
\]
with $\eidem_2\eidem_1=\eidem_1\eidem_2$ by commutativity
(\cref{prop:producto}); and $\eidem_1+\eidem_2=1$ is immediate.

For the decomposition, we write $z=\zeta_1\eidem_1+\zeta_2\eidem_2$ with
$\zeta_1,\zeta_2\in\CC_{\Bt}$ undetermined, $\zeta_k=\alpha_k+\beta_k\Bt$.
Substituting \eqref{eq:idempotentes} and grouping by the basis
$\{1,\Bt,\Bx,\Btx\}$ (recall $\Bt\Btx=-\Bx$, from the table in the proof of
\cref{prop:producto}):
\[
\zeta_1\eidem_1+\zeta_2\eidem_2
=\tfrac{\alpha_1+\alpha_2}{2}
+\tfrac{\beta_1+\beta_2}{2}\Bt
-\tfrac{\beta_1-\beta_2}{2}\Bx
+\tfrac{\alpha_1-\alpha_2}{2}\Btx .
\]
Equating coefficients with \eqref{eq:elemento} yields the linear system
$\tfrac12(\alpha_1+\alpha_2)=a$, $\tfrac12(\beta_1+\beta_2)=b$,
$-\tfrac12(\beta_1-\beta_2)=c$, $\tfrac12(\alpha_1-\alpha_2)=d$, whose unique
solution is $\alpha_1=a+d$, $\beta_1=b-c$, $\alpha_2=a-d$, $\beta_2=b+c$,
that is, \eqref{eq:desc-zeta}. Uniqueness follows from the fact that the
system has nonzero determinant.

Finally, rule \eqref{eq:prod-diag} is a consequence of
\eqref{eq:idem-props}: for $z=\zeta_1\eidem_1+\zeta_2\eidem_2$ and
$w=\eta_1\eidem_1+\eta_2\eidem_2$,
\[
zw=\zeta_1\eta_1\,\eidem_1^2+\zeta_2\eta_2\,\eidem_2^2
+(\zeta_1\eta_2+\zeta_2\eta_1)\,\eidem_1\eidem_2
=\zeta_1\eta_1\,\eidem_1+\zeta_2\eta_2\,\eidem_2,
\]
where we have used that $\zeta_k,\eta_k\in\CC_{\Bt}$ commute with the
idempotents (both live in the commutative subalgebra $\mathcal{A}$). This
identifies $\mathcal{A}\cong\CC_{\Bt}\times\CC_{\Bt}$ as a product algebra,
whence every rational or analytic operation is carried out componentwise.
\end{proof}

\subsection{Inverse and zero divisors}

\begin{proposicion}[Characterization of the inverse]
\label{prop:inverso}
Let $z=\zeta_1\eidem_1+\zeta_2\eidem_2\in\mathcal{A}$ with $\zeta_1,\zeta_2$
given by \eqref{eq:desc-zeta}. Then $z$ is invertible in $\mathcal{A}$ if and
only if $\zeta_1\neq 0$ and $\zeta_2\neq 0$, and in that case
\begin{equation}
z^{-1}=\zeta_1^{-1}\,\eidem_1+\zeta_2^{-1}\,\eidem_2 ,
\label{eq:inverso-diag}
\end{equation}
where $\zeta_k^{-1}$ is the ordinary complex inverse in $\CC_{\Bt}$.
\end{proposicion}

\begin{proof}
If $\zeta_1\neq 0$ and $\zeta_2\neq 0$, the element
$y=\zeta_1^{-1}\eidem_1+\zeta_2^{-1}\eidem_2$ is well defined and, by
\eqref{eq:prod-diag},
$zy=(\zeta_1\zeta_1^{-1})\eidem_1+(\zeta_2\zeta_2^{-1})\eidem_2
=\eidem_1+\eidem_2=1$, hence $y=z^{-1}$. Conversely, if either of the
components vanishes, say $\zeta_1=0$, then $z=\zeta_2\eidem_2$ and for any
$w=\eta_1\eidem_1+\eta_2\eidem_2$ one has $zw=\zeta_2\eta_2\,\eidem_2$, whose
component along $\eidem_1$ is always zero; therefore $zw\neq 1$ for every $w$
and $z$ is not invertible.
\end{proof}

\begin{observacion}[Zero divisors]
\label{obs:divisores}
Unlike $\CC$, the algebra $\mathcal{A}$ is \emph{not} a field: it contains
zero divisors. The idempotents themselves are examples, since
$\eidem_1\eidem_2=0$ with $\eidem_1,\eidem_2\neq 0$. Geometrically, the
noninvertible elements are exactly those lying on the two ``isotropic
cones'' $\zeta_1=0$ or $\zeta_2=0$. This circumstance conditions the later
analytic discussion (poles and residues of the transform) and forces a
distinction between genuine zeros and degenerate singularities.
\end{observacion}

In practice it is convenient to have a way of computing the inverse that
does not go through diagonalization, especially for a direct numerical
implementation in real coordinates. Left multiplication by $z$ is an
$\RR$-linear map $w\mapsto zw$ from $\RR^4$ to $\RR^4$.

\begin{proposicion}[Multiplication matrix]
\label{prop:matriz}
For $z=a+b\,\Bt+c\,\Bx+d\,\Btx$, the matrix of left multiplication in the
basis $[\,1,\Bt,\Bx,\Btx\,]$ is
\begin{equation}
M(z)=
\begin{pmatrix}
a & -b & -c & d\\
b & a & -d & -c\\
c & -d & a & -b\\
d & c & b & a
\end{pmatrix}.
\label{eq:matriz-M}
\end{equation}
Consequently, $z$ is invertible if and only if $\det M(z)\neq 0$, and the
coordinates of $z^{-1}$ are obtained by solving the linear system
$M(z)\,w=[\,1,0,0,0\,]^{\top}$ by standard matrix
algebra~\cite{harville1997matrix}. Moreover $M$ is an algebra homomorphism:
$M(zw)=M(z)M(w)$ and $M(z)^{\top}=M(\bar z)$ with
$\bar z=a-b\,\Bt-c\,\Bx+d\,\Btx$.
\end{proposicion}

\begin{proof}
The $k$-th column of $M(z)$ contains the coordinates of $z\cdot u_k$, where
$u_k$ ranges over the basis $\{1,\Bt,\Bx,\Btx\}$. Taking $w=u_1=1$ in
\eqref{eq:prod-1}--\eqref{eq:prod-Btx} (that is, $p=1$, $q=r=s=0$) gives the
first column $[a,b,c,d]^{\top}$; with $w=u_2=\Bt$ ($q=1$, the rest zero) the
second column $[-b,a,-d,c]^{\top}$; with $w=u_3=\Bx$ the third
$[-c,-d,a,b]^{\top}$; and with $w=u_4=\Btx$ the fourth $[d,-c,-b,a]^{\top}$.
Arranging them as columns yields exactly \eqref{eq:matriz-M}. That $M$ is a
homomorphism follows from the associativity of the product:
$M(zw)v=(zw)v=z(wv)=M(z)M(w)v$ for every $v$. The identity
$M(z)^{\top}=M(\bar z)$ is checked by direct inspection of the entries. Since
$M$ is injective (the regular representation of a unital algebra is), $\det
M(z)\neq 0$ is equivalent to the invertibility of $z$, and
$M(z^{-1})=M(z)^{-1}$.
\end{proof}

\begin{observacion}
\label{obs:det-M}
The two routes are consistent: a direct computation gives
$\det M(z)=\lvert\zeta_1\rvert^2\,\lvert\zeta_2\rvert^2$, where
$\lvert\cdot\rvert$ is the complex modulus in $\CC_{\Bt}$. Therefore
$\det M(z)=0$ is equivalent to $\zeta_1=0$ or $\zeta_2=0$, in full agreement
with \cref{prop:inverso} and \cref{obs:divisores}.
\end{observacion}

\subsection{Exponential and elementary functions}

\begin{proposicion}[Exponentials of the generators]
\label{prop:exp}
For every $\theta\in\RR$ and $\delta\in\RR$, in terms of the exponential
series $\exp(z)=\sum_{n\ge0}z^n/n!$, the identities
\begin{align}
\exp(\Bt\,\theta) &= \cos\theta + \Bt\sin\theta, \label{eq:exp-Bt}\\
\exp(\Bx\,\theta) &= \cos\theta + \Bx\sin\theta, \label{eq:exp-Bx}\\
\exp(\Btx\,\delta) &= \cosh\delta + \Btx\sinh\delta. \label{eq:exp-Btx}
\end{align}
hold. More generally, for every
$z=\zeta_1\eidem_1+\zeta_2\eidem_2\in\mathcal{A}$,
\begin{equation}
\exp(z)=e^{\zeta_1}\,\eidem_1+e^{\zeta_2}\,\eidem_2 ,
\label{eq:exp-general}
\end{equation}
where $e^{\zeta_k}$ is the ordinary complex exponential in $\CC_{\Bt}$.
\end{proposicion}

\begin{proof}
Since $\Bt^2=-1$ (rule \eqref{eq:reglas}), the powers repeat with period
four, exactly like the imaginary unit: $\Bt^{2m}=(-1)^m$,
$\Bt^{2m+1}=(-1)^m\Bt$. Splitting the series into even and odd terms,
\[
\exp(\Bt\theta)=\sum_{m\ge0}\frac{(-1)^m\theta^{2m}}{(2m)!}
+\Bt\sum_{m\ge0}\frac{(-1)^m\theta^{2m+1}}{(2m+1)!}
=\cos\theta+\Bt\sin\theta,
\]
which proves \eqref{eq:exp-Bt}; identity \eqref{eq:exp-Bx} is identical with
$\Bx^2=-1$. For \eqref{eq:exp-Btx}, now $\Btx^2=+1$, so that
$\Btx^{2m}=1$ and $\Btx^{2m+1}=\Btx$, and the same splitting gives the series
for $\cosh$ and $\sinh$.

The general formula \eqref{eq:exp-general} follows from
\cref{prop:idempotente}: by \eqref{eq:prod-diag}, the powers are computed
componentwise, $z^n=\zeta_1^n\eidem_1+\zeta_2^n\eidem_2$ (using
$\eidem_k^n=\eidem_k$ and $\eidem_1\eidem_2=0$), and the exponential series is
absolutely convergent in the norm of \cref{def:norma}, which allows the sum
and the idempotents to be interchanged:
\[
\exp(z)=\sum_{n\ge0}\frac{z^n}{n!}
=\Bigl(\sum_{n\ge0}\frac{\zeta_1^n}{n!}\Bigr)\eidem_1
+\Bigl(\sum_{n\ge0}\frac{\zeta_2^n}{n!}\Bigr)\eidem_2
=e^{\zeta_1}\eidem_1+e^{\zeta_2}\eidem_2 .
\]
The identities \eqref{eq:exp-Bt}--\eqref{eq:exp-Btx} are special cases of
\eqref{eq:exp-general} after decomposing $\Bt\theta$, $\Bx\theta$, and
$\Btx\delta$ according to \eqref{eq:desc-zeta}.
\end{proof}

These exponentials are the basis of the transform of
\cref{sec:transformada}: the kernel $e^{-\st t}e^{-\sx x}$ appearing in the
definition of $\transf{\cdot}$ factorizes, thanks to the commutativity of
$\Bt$ and $\Bx$, into a product of a temporal factor and a spatial factor
whose phases rotate in independent bivector planes.

\subsection{Geometric arguments and auxiliary norm}

To locate phases (in particular, the spatial phase associated with the
position of a fault in later sections) it is convenient to define arguments
referred to each bivector plane.

\begin{definicion}[Arguments and norm]
\label{def:norma}
For $z=a+b\,\Bt+c\,\Bx+d\,\Btx\in\mathcal{A}$ the temporal and spatial
arguments are defined as
\begin{equation}
\ArgBt(z)=\operatorname{atan2}(b,a),\qquad
\ArgBx(z)=\operatorname{atan2}(c,a),
\label{eq:args}
\end{equation}
together with the \emph{auxiliary Euclidean norm}
\begin{equation}
\lVert z\rVert=\sqrt{a^2+b^2+c^2+d^2}.
\label{eq:norma}
\end{equation}
\end{definicion}

\begin{observacion}[Reading a rotor]
\label{obs:rotor}
For a pure \emph{spatial rotor} of the form
$z=r\bigl(\cos\varphi+\Bx\sin\varphi\bigr)$ with $r>0$, the coordinates are
$a=r\cos\varphi$, $c=r\sin\varphi$, and $b=d=0$, so that
$\ArgBx(z)=\operatorname{atan2}(r\sin\varphi,\,r\cos\varphi)=\varphi$ recovers
exactly the angle $\varphi$ (modulo $2\pi$), and $\lVert z\rVert=r$. This
property is what will allow a spatial position to be read as a phase angle.
\end{observacion}

\begin{observacion}[The norm is not multiplicative]
\label{obs:no-mult}
It should be emphasized that \eqref{eq:norma} is a vector-space norm but is
\emph{not} multiplicative: in general
$\lVert zw\rVert\neq\lVert z\rVert\,\lVert w\rVert$. This is unavoidable,
since $\mathcal{A}$ has zero divisors (\cref{obs:divisores}): taking
$z=\eidem_1\neq0$ and $w=\eidem_2\neq0$ one has $zw=0$ and therefore
$\lVert zw\rVert=0\neq\lVert\eidem_1\rVert\,\lVert\eidem_2\rVert$. The norm
\eqref{eq:norma} is thus used solely as an auxiliary instrument for
controlling convergence and magnitude, not as an algebraic absolute value.
\end{observacion}

\subsection{Relation to the implementation}

\begin{observacion}[Correspondence with the code]
\label{obs:codigo-alg}
The operations above correspond directly to the module
\texttt{geolaplace.algebra} of the reproducible implementation. The class
\texttt{Geo} stores each element $z$ as the real vector $[a,b,c,d]$ of
\eqref{eq:elemento}; the product implements
\eqref{eq:prod-1}--\eqref{eq:prod-Btx}; the inverse resorts to the
diagonalization of \cref{prop:inverso} or, equivalently, to solving the
system $M(z)\,w=[\,1,0,0,0\,]^{\top}$ with the matrix \eqref{eq:matriz-M}
\cite{harris2020numpy}; the exponential uses \eqref{eq:exp-general}; and the
argument and norm methods carry out \eqref{eq:args}--\eqref{eq:norma}. The
inversion routine flags as degenerate the cases $\zeta_1=0$ or $\zeta_2=0$
(\cref{obs:divisores,obs:det-M}), which in the application sections
identifies singular configurations of the distributed system.
\end{observacion}

\section{Two-dimensional geometric Laplace transform and operational properties}
\label{sec:transformada}

In this section we rigorously define the two-dimensional geometric Laplace
transform on the commutative subalgebra generated by the bivectors $\Bt$
and $\Bx$, we discuss its region of convergence through the idempotent
decomposition, and we establish the operational properties that underpin the
distributed model of the following sections. The treatment follows the spirit
of the classical theory of the Laplace transform \cite{oppenheim1997signals}, and
it exploits the structure of Segre's bicomplex numbers
\cite{luna2015bicomplex} to keep the temporal and spatial phases separate.

\subsection{Definition and factorized kernel}

We recall the rules of the algebra in use: $\Bt^2=-1$, $\Bx^2=-1$,
$\Bt\Bx=\Bx\Bt=\Btx$ and $\Btx^2=+1$. The algebra is \emph{commutative}
(bicomplex), so that a generic element is written
$z=a+b\Bt+c\Bx+d\Btx$ with $a,b,c,d\in\RR$.

\begin{definicion}[Two-dimensional geometric Laplace transform]
\label{def:transformada}
Let $f:\RR_{\ge0}\times\RR_{\ge0}\to\mathcal{A}$ be a function taking values in the
\emph{commutative subalgebra} $\mathcal{A}=\operatorname{span}\{1,\Bt,\Bx,\Btx\}$
of \cref{sec:algebra} (the restriction to $\mathcal{A}$, rather than to the whole
$\Cl(4,0)$, is what guarantees that the spectral factors commute with $f$ in the
operational proofs). Its \emph{geometric Laplace transform} is
\begin{equation}
\label{eq:def-transformada}
\transf{f}=F(\st,\sx)=\int_0^\infty\!\!\int_0^\infty f(t,x)\,
   e^{-\st t}\,e^{-\sx x}\dd x\dd t,
\end{equation}
with the \emph{geometric Laplace variables}
\begin{equation}
\label{eq:variables}
\st=\rho_t+\Bt\,\omega,\qquad \sx=\rho_x+\Bx\,k,
\end{equation}
where $\rho_t,\omega,\rho_x,k\in\RR$. Since $\Bt$ and $\Bx$ commute, the kernels
factor into a single exponential,
\begin{equation}
\label{eq:nucleo-factorizado}
e^{-\st t}\,e^{-\sx x}=e^{-(\st t+\sx x)},
\end{equation}
and the integral is understood componentwise in the basis
$\{1,\Bt,\Bx,\Btx\}$ whenever it exists.
\end{definicion}

The factorization \eqref{eq:nucleo-factorizado} is a direct consequence of
commutativity: $\st t$ and $\sx x$ belong to distinct planes ($\{1,\Bt\}$ and
$\{1,\Bx\}$), but because they commute one has
$e^{-\st t}e^{-\sx x}=e^{-\st t-\sx x}$ with no Baker--Campbell--Hausdorff
correction term. Each kernel is a rotor of decreasing magnitude,
\begin{equation}
\label{eq:nucleos-rotor}
e^{-\st t}=e^{-\rho_t t}\bigl(\cos\omega t-\Bt\sin\omega t\bigr),\qquad
e^{-\sx x}=e^{-\rho_x x}\bigl(\cos k x-\Bx\sin k x\bigr),
\end{equation}
so that $\rho_t$ and $\rho_x$ control the attenuation, while $\omega$ and $k$
introduce temporal and spatial oscillation in independent geometric planes.

\subsection{Region of convergence}

The existence of the integral \eqref{eq:def-transformada} is analyzed most
conveniently in the \emph{idempotent representation} of the algebra. With the
idempotents
\begin{equation}
\label{eq:idem-tr}
\eidem_1=\tfrac12(1+\Btx),\qquad \eidem_2=\tfrac12(1-\Btx),
\qquad \eidem_1^2=\eidem_1,\ \ \eidem_2^2=\eidem_2,\ \ \eidem_1\eidem_2=0,
\end{equation}
every element diagonalizes as $z=\zeta_1\eidem_1+\zeta_2\eidem_2$, where
$\zeta_1,\zeta_2$ are ordinary complex numbers in the plane $\{1,\Bt\}$ (with
$\Bt$ playing the role of the imaginary unit). In particular, the Laplace
variables decompose as
\begin{equation}
\label{eq:variables-idem}
\st t+\sx x=\bigl[(\rho_t t+\rho_x x)+\Bt(\omega t-k x)\bigr]\eidem_1
          +\bigl[(\rho_t t+\rho_x x)+\Bt(\omega t+k x)\bigr]\eidem_2,
\end{equation}
so that the kernel separates into two ordinary complex exponentials, one
for each idempotent component. This reduces the study of convergence to that of
two classical complex Laplace transforms \cite{oppenheim1997signals}.

A word of caution about norms is in order. The Euclidean norm of the
components, $\lVert a+b\Bt+c\Bx+d\Btx\rVert^2=a^2+b^2+c^2+d^2$, is \emph{not}
submultiplicative on $\mathcal{A}$: the idempotent $\eidem_1$ satisfies
$\eidem_1^2=\eidem_1$ and yet
$\lVert\eidem_1^2\rVert=1/\sqrt{2}>\lVert\eidem_1\rVert^2=1/2$, so the naive
bound $\lVert zw\rVert\le\lVert z\rVert\,\lVert w\rVert$ is not available. The
following elementary lemma provides the exact substitute used in the
convergence arguments below.

\begin{lema}[Norm of a product in the idempotent basis]
\label{lem:norma-producto}
Let $z=\zeta_1\eidem_1+\zeta_2\eidem_2$ and $w=\eta_1\eidem_1+\eta_2\eidem_2$
be elements of $\mathcal{A}$ in the idempotent representation
\eqref{eq:idem-tr}, with $\zeta_1,\zeta_2,\eta_1,\eta_2$ complex numbers in
the plane $\{1,\Bt\}$. Then
\begin{equation}
\label{eq:norma-idem}
\lVert z\rVert^2=\tfrac12\bigl(\lvert\zeta_1\rvert^2+\lvert\zeta_2\rvert^2\bigr)
\qquad\text{and}\qquad
\lVert zw\rVert^2=\tfrac12\bigl(\lvert\zeta_1\rvert^2\lvert\eta_1\rvert^2
   +\lvert\zeta_2\rvert^2\lvert\eta_2\rvert^2\bigr).
\end{equation}
In particular
$\lVert zw\rVert\le\max\bigl(\lvert\eta_1\rvert,\lvert\eta_2\rvert\bigr)\,
\lVert z\rVert$, and if both idempotent components of $w$ share the same
modulus $\lvert\eta_1\rvert=\lvert\eta_2\rvert=\mu$ then the product scales
the norm \emph{exactly}: $\lVert zw\rVert=\mu\,\lVert z\rVert$.
\end{lema}

\begin{proof}
Write $\zeta_i=p_i+\Bt q_i$ with $p_i,q_i\in\RR$. From
$\eidem_{1,2}=\tfrac12(1\pm\Btx)$ and $\Bt\Btx=-\Bx$ one gets
$\Bt\eidem_1=\tfrac12(\Bt-\Bx)$ and $\Bt\eidem_2=\tfrac12(\Bt+\Bx)$, so the
components of $z$ in the basis $\{1,\Bt,\Bx,\Btx\}$ are
$a=\tfrac12(p_1+p_2)$, $b=\tfrac12(q_1+q_2)$, $c=\tfrac12(q_2-q_1)$ and
$d=\tfrac12(p_1-p_2)$. Expanding,
$a^2+b^2+c^2+d^2=\tfrac12\bigl(p_1^2+q_1^2+p_2^2+q_2^2\bigr)$, which is the
first identity of \eqref{eq:norma-idem}. Since $\eidem_1\eidem_2=0$ and
$\eidem_i^2=\eidem_i$, the product diagonalizes componentwise,
$zw=\zeta_1\eta_1\eidem_1+\zeta_2\eta_2\eidem_2$, and the second identity
follows from the first applied to $zw$ together with the multiplicativity of
the complex modulus, $\lvert\zeta_i\eta_i\rvert=\lvert\zeta_i\rvert\,
\lvert\eta_i\rvert$. The bound and the equality case are immediate
consequences of \eqref{eq:norma-idem}.
\end{proof}

\begin{proposicion}[Region of convergence]
\label{prop:roc}
Suppose that $f$ is of exponential order, that is, that there exist constants
$M>0$, $\sigma_t,\sigma_x\in\RR$ such that
$\lVert f(t,x)\rVert\le M\,e^{\sigma_t t}\,e^{\sigma_x x}$ for all
$t,x\ge0$, where $\lVert\cdot\rVert$ is the Euclidean norm of the components.
Then the integral \eqref{eq:def-transformada} converges absolutely if
\begin{equation}
\label{eq:roc}
\Real(\st)=\rho_t>\sigma_t\qquad\text{and}\qquad \Real(\sx)=\rho_x>\sigma_x .
\end{equation}
Moreover, the transform is invariant per idempotent component: on each
half-plane \eqref{eq:roc} the two components $\zeta_1,\zeta_2$ of the integrand
have identical attenuation factors $e^{-\rho_t t}e^{-\rho_x x}$, so that
convergence is equivalent to that of \emph{two} ordinary complex Laplace
transforms.
\end{proposicion}

\begin{proof}
By \eqref{eq:variables-idem}, the two idempotent components of the kernel
$e^{-\st t}e^{-\sx x}$ are the ordinary complex exponentials
$\eta_1=e^{-(\rho_t t+\rho_x x)-\Bt(\omega t-kx)}$ and
$\eta_2=e^{-(\rho_t t+\rho_x x)-\Bt(\omega t+kx)}$, whose moduli coincide:
$\lvert\eta_1\rvert=\lvert\eta_2\rvert=e^{-\rho_t t}e^{-\rho_x x}$ (the
oscillatory parts are unit rotors). The equality case of
\cref{lem:norma-producto} then gives, for every $t,x\ge0$,
\begin{equation*}
\bigl\lVert f(t,x)\,e^{-\st t}e^{-\sx x}\bigr\rVert
=e^{-\rho_t t}e^{-\rho_x x}\,\lVert f(t,x)\rVert
\le M\,e^{-(\rho_t-\sigma_t)t}\,e^{-(\rho_x-\sigma_x)x},
\end{equation*}
where the last step uses the exponential-order hypothesis. Note that
\cref{lem:norma-producto} is genuinely needed here: the norm is not
submultiplicative, so one cannot simply bound
$\lVert f\,e^{-\st t}e^{-\sx x}\rVert$ by the product of the norms. Hence
\begin{equation*}
\int_0^\infty\!\!\int_0^\infty \bigl\lVert f(t,x)\,
   e^{-\st t}e^{-\sx x}\bigr\rVert\dd x\dd t
\le M\int_0^\infty e^{-(\rho_t-\sigma_t)t}\dd t
     \int_0^\infty e^{-(\rho_x-\sigma_x)x}\dd x,
\end{equation*}
a product of two integrals that are finite precisely when \eqref{eq:roc} holds;
in that case they equal $1/(\rho_t-\sigma_t)$ and $1/(\rho_x-\sigma_x)$.
The absolute convergence of the vector-valued integral follows componentwise.
Finally, writing the integrand in the idempotent basis
\eqref{eq:variables-idem}, the common attenuation factor $e^{-\rho_t t-\rho_x x}$
is that of each of the two complex exponentials $\zeta_1,\zeta_2$, which
proves the stated equivalence.
\end{proof}

\subsection{Linearity}

\begin{proposicion}[Linearity]
\label{prop:linealidad}
Let $f,g$ be functions whose geometric transforms exist in a common region,
and let $\lambda,\mu\in\mathcal{A}$ be constants. Then
\begin{equation}
\label{eq:linealidad}
\transf{\lambda f+\mu g}=\lambda\,\transf{f}+\mu\,\transf{g}
\end{equation}
on the intersection of the regions of convergence.
\end{proposicion}

\begin{proof}
The integral \eqref{eq:def-transformada} is linear by construction, and since the
algebra is commutative the constants $\lambda,\mu$ can be pulled out of the
integral on the left without altering the order of the factors. The identity then
follows from the additivity of the componentwise integral.
\end{proof}

\subsection{Transform of the partial derivatives}

The operational usefulness of the transform \eqref{eq:def-transformada} lies in
turning the partial derivatives into algebraic multiplications, as in the
classical case \cite{oppenheim1997signals}. We denote by
$F_{t=0}(\sx)=\int_0^\infty f(0,x)\,e^{-\sx x}\dd x$ the one-dimensional spatial
transform of the initial condition $f(0,x)$, and by
$F_{x=0}(\st)=\int_0^\infty f(t,0)\,e^{-\st t}\dd t$ the one-dimensional temporal
transform of the boundary condition $f(t,0)$.

\begin{proposicion}[Transform of the partial derivatives]
\label{prop:derivadas}
Let $f(t,x)$ be continuously differentiable on $\RR_{\ge0}\times\RR_{\ge0}$,
and suppose that $f$, $\pt f$ and $\px f$ are all of exponential order with
common constants $M>0$ and $\sigma_t,\sigma_x\in\RR$ in the sense of
\cref{prop:roc}. Then, for every $(\st,\sx)$ in the region \eqref{eq:roc}, the
transforms $F(\st,\sx)$, $\transf{\pt f}$ and $\transf{\px f}$, as well as the
one-dimensional transforms $F_{t=0}(\sx)$ and $F_{x=0}(\st)$, converge
absolutely, and
\begin{align}
\transf{\pt f}&=F\st-F_{t=0}(\sx),\label{eq:deriv-t}\\
\transf{\px f}&=F\sx-F_{x=0}(\st).\label{eq:deriv-x}
\end{align}
In particular, with \emph{zero} initial and boundary conditions
($f(0,x)\equiv0$ and $f(t,0)\equiv0$) one obtains
\begin{equation}
\label{eq:deriv-nulas}
\transf{\pt f}=F\st,\qquad \transf{\px f}=F\sx .
\end{equation}
\end{proposicion}

\begin{proof}
We prove \eqref{eq:deriv-t}; \eqref{eq:deriv-x} is analogous by symmetry in the
variables. The absolute-convergence claims follow from \cref{prop:roc} applied
to $f$, $\pt f$ and $\px f$; for the one-dimensional transforms, the
exponential-order bound at $t=0$ gives
$\lVert f(0,x)\rVert\le M\,e^{\sigma_x x}$, so $F_{t=0}(\sx)$ converges
absolutely whenever $\rho_x>\sigma_x$ (and symmetrically for $F_{x=0}$).
In particular, the integrand $\pt f(t,x)\,e^{-\st t}e^{-\sx x}$ is absolutely
integrable on $[0,\infty)^2$ throughout the region \eqref{eq:roc}, so by the
Fubini--Tonelli theorem \cite{folland1999real} the double integral can be
evaluated as an iterated integral in either order; we integrate first in $t$,
with the spatial integral as a parameter. Since $f$ takes values in the
commutative subalgebra $\mathcal{A}$, the spectral factor $\st$ commutes with
$f(t,x)$ and $e^{-\st t}$ is differentiated as a scalar:
$\pt\bigl(e^{-\st t}\bigr)=-\st\,e^{-\st t}$. Integration by parts then yields
\begin{align*}
\transf{\pt f}
&=\int_0^\infty\!\!\int_0^\infty \pt f(t,x)\,e^{-\st t}e^{-\sx x}\dd x\dd t\\
&=\int_0^\infty\!\!\left[\,\Bigl(f(t,x)\,e^{-\st t}\Bigr)_{t=0}^{t\to\infty}
   +\st\int_0^\infty f(t,x)\,e^{-\st t}\dd t\right]e^{-\sx x}\dd x .
\end{align*}
The boundary term vanishes as $t\to\infty$ because, by the equality case of
\cref{lem:norma-producto} (the two idempotent components of $e^{-\st t}$ share
the modulus $e^{-\rho_t t}$),
$\lVert f(t,x)\,e^{-\st t}\rVert=e^{-\rho_t t}\lVert f(t,x)\rVert
\le M\,e^{-(\rho_t-\sigma_t)t}e^{\sigma_x x}\to0$ for each fixed $x$; at $t=0$
it contributes $-f(0,x)$. Substituting,
\begin{equation*}
\transf{\pt f}
=\st\int_0^\infty\!\!\int_0^\infty f(t,x)\,e^{-\st t}e^{-\sx x}\dd x\dd t
 -\int_0^\infty f(0,x)\,e^{-\sx x}\dd x
=F\st-F_{t=0}(\sx),
\end{equation*}
where the first integral is $F\st$ (the constant $\st$ commutes and factors out) and
the second is, by definition, $F_{t=0}(\sx)$. The case \eqref{eq:deriv-nulas} is
obtained by cancelling the condition terms. The identity \eqref{eq:deriv-x} is
proved by integrating by parts in $x$ with $\px(e^{-\sx x})=-\sx\,e^{-\sx x}$.
\end{proof}

\subsection{Transform of a separable mode}

\begin{proposicion}[Transform of a separable mode]
\label{prop:separable}
Let the separable space--time mode be
\begin{equation}
\label{eq:modo-separable}
f(t,x)=A\,e^{(-\alpha_t+\Bt\omega_0)t}\,e^{(-\alpha_x+\Bx k_0)x},
\end{equation}
with $A\in\mathcal{A}$ constant and $\alpha_t,\omega_0,\alpha_x,k_0\in\RR$. If
$\rho_t>-\alpha_t$ and $\rho_x>-\alpha_x$, its geometric transform is
\begin{equation}
\label{eq:transf-separable}
F(\st,\sx)=A\,(\st+\alpha_t-\Bt\omega_0)^{-1}\,(\sx+\alpha_x-\Bx k_0)^{-1}.
\end{equation}
\end{proposicion}

\begin{proof}
Since the integrand \eqref{eq:modo-separable} factors into a product of a
function of $t$ and a function of $x$ and the algebra is commutative, the double
integral splits into a product of two single integrals,
\begin{equation*}
F(\st,\sx)=A\left(\int_0^\infty e^{-(\st+\alpha_t-\Bt\omega_0)t}\dd t\right)
            \left(\int_0^\infty e^{-(\sx+\alpha_x-\Bx k_0)x}\dd x\right).
\end{equation*}
For the temporal integral, let $z_t=\st+\alpha_t-\Bt\omega_0$. In the
idempotent representation \eqref{eq:idem-tr}, $z_t=\zeta_1\eidem_1+
\zeta_2\eidem_2$ with $\Real(\zeta_1)=\Real(\zeta_2)=\rho_t+\alpha_t>0$, so that
$z_t$ is invertible (both idempotent components have positive real part,
in particular nonzero) and
\begin{equation*}
\int_0^\infty e^{-z_t t}\dd t
=\Bigl[-z_t^{-1}e^{-z_t t}\Bigr]_0^\infty=z_t^{-1}
=(\st+\alpha_t-\Bt\omega_0)^{-1},
\end{equation*}
where the term at $t\to\infty$ vanishes because the common attenuation factor of
both idempotent components is $e^{-(\rho_t+\alpha_t)t}\to0$. In exactly the same
way, with $z_x=\sx+\alpha_x-\Bx k_0$ and $\rho_x+\alpha_x>0$,
\begin{equation*}
\int_0^\infty e^{-z_x x}\dd x=(\sx+\alpha_x-\Bx k_0)^{-1}.
\end{equation*}
Substituting yields \eqref{eq:transf-separable}; the order of the factors is
irrelevant by commutativity.
\end{proof}

\begin{observacion}[Geometric poles]
\label{obs:polos}
The transform \eqref{eq:transf-separable} is singular when either of the two
factors ceases to be invertible. The value
\begin{equation}
\label{eq:polo-t}
\st=-\alpha_t+\Bt\omega_0
\end{equation}
is a \emph{geometric temporal pole}: its scalar part $-\alpha_t$ sets the decay
rate and its $\Bt$ component, $\omega_0$, the frequency of temporal
oscillation. Analogously,
\begin{equation}
\label{eq:polo-x}
\sx=-\alpha_x+\Bx k_0
\end{equation}
is a \emph{geometric spatial pole}, with $-\alpha_x$ the spatial attenuation
rate and $k_0$ the wavenumber. The separation of the two poles into distinct
planes ($\Bt$ and $\Bx$) is the property that is exploited in the sections on
modeling and fault detection: the temporal dynamics and the spatial propagation
are encoded in independent algebraic components.
\end{observacion}

\subsection{Comparison with the ordinary complex Laplace transform}

The scope of the construction should be clear from the outset: with a
\emph{single} independent variable, the geometric transform contributes no new
mathematical content beyond the classical complex Laplace transform.

\begin{observacion}[One-dimensional equivalence]
\label{obs:equivalencia-1d}
Let $g(t)$ be a signal taking values in the temporal plane $\{1,\Bt\}$ and consider
only the variable $\st=\rho_t+\Bt\omega$. The plane $\{1,\Bt\}$ is isomorphic
as a field to $\CC$ via $\Bt\mapsto j$, so that
\begin{equation}
\label{eq:equivalencia-1d}
\int_0^\infty g(t)\,e^{-\st t}\dd t
\;\;\xrightarrow{\ \Bt\mapsto j\ }\;\;
G(s)=\int_0^\infty g(t)\,e^{-st}\dd t,\qquad s=\rho_t+j\omega,
\end{equation}
which is exactly the ordinary Laplace transform
\cite{oppenheim1997signals}. Replacing the imaginary unit $j$ by a bivector
$\Bt$ is a mere renaming of the imaginary symbol and adds no structure.
\end{observacion}

The contribution of the transform \eqref{eq:def-transformada} is
\emph{multivariable} in nature, and it should be stated precisely. A complex
Laplace transform with \emph{two} independent variables
$(s_t,s_x)\in\CC^2$ already keeps the temporal and spatial phases separate,
without ambiguity: each variable retains its own coordinate. The multivariable
complex theory can certainly distinguish them. The contribution
of the geometric formulation is \emph{to represent both phases in algebraically
distinguishable planes} ---the temporal $\{1,\Bt\}$ and the spatial $\{1,\Bx\}$---
within \emph{one and the same} bicomplex structure \cite{luna2015bicomplex}: a single
element $z=a+b\Bt+c\Bx+d\Btx$ packages the magnitude ($a$), the temporal phase
($\Bt$ component), the spatial phase ($\Bx$ component) and their coupling ($\Btx$),
each legible separately through $\ArgBt$ and $\ArgBx$. This unified
representation makes it easier to operate with mixed products, to read arguments
and to manipulate spatial signatures (such as the rotor $e^{-\Bx k x_f}$ of a
localized fault) within a single algebraic object, and it is the basis for the
interpretation of the poles
\eqref{eq:polo-t}--\eqref{eq:polo-x} and for the geometric spatial-phase
localization methods of the later sections.

\begin{observacion}[Commutative case and future extension]
\label{obs:conmutativo}
The preceding development assumes $\Bt\Bx=\Bx\Bt$, that is, \emph{commutative}
generators (bicomplex algebra). This is the case actually implemented and
the only one that guarantees the factorization \eqref{eq:nucleo-factorizado}
without corrections. A \emph{non-commutative} version, in which $\Bt\Bx=-\Bx\Bt$,
would require handling Baker--Campbell--Hausdorff terms in
\eqref{eq:nucleo-factorizado} and is left as a line of future work, beyond the
scope of this article.
\end{observacion}

\subsection{Correspondence with the implementation}

\begin{observacion}[Relation to the code]
\label{obs:codigo-tr}
The constructions of this section correspond directly to the module
\texttt{geolaplace.transform} of the reproducible implementation. The geometric
variables \eqref{eq:variables} are obtained with \texttt{st(rho, omega)} and
\texttt{sx(rho, k)}, which return $\rho_t+\Bt\omega$ and $\rho_x+\Bx k$,
respectively. The kernels \eqref{eq:nucleos-rotor} are computed via
\texttt{kernel\_t(t, rho, omega)}$=e^{-\st t}$ and
\texttt{kernel\_x(x, rho, k)}$=e^{-\sx x}$, evaluated exactly through
the idempotent decomposition \eqref{eq:idem-tr}. Finally,
\texttt{separable\_transform(alpha\_t, omega0, alpha\_x, k0, amplitude)} returns
the closed-form function $F(\st,\sx)$ of \eqref{eq:transf-separable}, building the
invertible factors from the geometric poles
\eqref{eq:polo-t}--\eqref{eq:polo-x}; the invertibility required in the
proof of \cref{prop:separable} translates into the region-of-convergence
condition $\rho_t>-\alpha_t$ and $\rho_x>-\alpha_x$.
\end{observacion}

\section{Distributed converter-line model and dispersion relation}
\label{sec:linea}

In this section we apply the two-dimensional Geometric Laplace Transform
to a continuous model of a DC line of
distributed converters. The goal is twofold: on the one hand, to obtain a
space--time transfer function between the converter injection and the
line voltage; on the other, to derive the associated dispersion relation, which
governs the propagation of disturbances along the line. It is
precisely in this setting, with \emph{two} independent variables
(time $t$ and position $x$), that the transform contributes a
geometric structure that does not reduce to the scalar complex Laplace transform.

\subsection{Modeling assumptions}

We work with an incremental model obtained by linearization about a
nominal operating point of the line (steady-state currents and
voltages). Under this approach, the relevant signals are deviations
with respect to that point.

\begin{hipotesis}[Incremental linearization and regularity]
\label{hip:linea}
We assume the following:
\begin{enumerate}[label=(H\arabic*),leftmargin=*]
  \item \textbf{Linearization.} The line dynamics are described by its
    incremental model about an operating point; $v(t,x)$ denotes the
    incremental voltage, $i(t,x)$ the incremental longitudinal current, and
    $u(t,x)$ the incremental current injected by the converters per
    unit length. All quantities are small deviations, so
    that the nonlinear terms are negligible.
  \item \textbf{Spatial continuity.} Although the converters are
    discrete devices located at specific positions, their
    aggregate effect is approximated by a \emph{distributed} injection $u(t,x)$
    continuous in $x$. The discrete case is treated exactly in
    \cref{subsec:inyeccion-discreta}, which quantifies when the continuous
    description is a limit of the discrete one (the injection, with rate
    $O(P^{-2})$ in the number of converters $P$) and when it is exact
    (the uniform droop feedback on the sub-Nyquist band, \cref{sec:control}).
  \item \textbf{Constant per-unit-length parameters.} The parameters
    $r$, $\ell$, $c$, and $g$ (defined below) are uniform along
    the line.
  \item \textbf{Zero initial and boundary conditions.} For the analysis
    of the operational properties by means of the transform we assume
    $v(0,x)=i(0,x)=0$ and, where applicable, zero boundary conditions. This
    isolates the response forced by the injection $u$ and removes boundary
    terms.
\end{enumerate}
\end{hipotesis}

The regularity conditions (sufficient decay in $t$ and $x$, and
exponential-type integrability) that guarantee the existence of the
transform are assumed to hold in the corresponding region of convergence
and are not repeated here.

\subsection{Line equations}

We adopt a transmission-line model with distributed parameters
per unit length \cite{paul2007transmission}: resistance $r$,
inductance $\ell$, capacitance $c$, and incremental conductance $g$. The
latter is interpreted as the incremental slope of the converter
characteristic about the operating point and, unlike the classical
passive case, \emph{may be negative} (converters with constant-power
control exhibit, incrementally, a negative conductance that
affects stability). The line equations are
\begin{subequations}
\label{eq:linea-tx}
\begin{align}
  \px v + \ell\,\pt i + r\,i &= 0, \label{eq:linea-tx-v}\\
  \px i + c\,\pt v + g\,v &= u. \label{eq:linea-tx-i}
\end{align}
\end{subequations}
Equation \eqref{eq:linea-tx-v} expresses the longitudinal voltage drop across
the series impedance of the line; equation \eqref{eq:linea-tx-i} expresses the
current balance: the spatial variation of the longitudinal current is
distributed among the capacitive current, the incremental conductance, and the
injection $u$ of the converters.

\subsection{Transformation and distributed transfer function}

We apply the Geometric Laplace Transform in the two variables, with
$\st=\rho_t+\Bt\omega$ and $\sx=\rho_x+\Bx k$. We write
$V(\st,\sx)=\transf{v}$, $I(\st,\sx)=\transf{i}$, and $U(\st,\sx)=\transf{u}$.
Under the zero conditions of \cref{hip:linea}, the transform of a
partial derivative substitutes $\pt\mapsto\st$ and $\px\mapsto\sx$ with no
boundary terms. Introducing the distributed impedance and admittance
\begin{equation}
\label{eq:ZY}
  \Zimp(\st) = r + \ell\,\st,
  \qquad
  \Yadm(\st) = g + c\,\st,
\end{equation}
equations \eqref{eq:linea-tx} transform into the algebraic system
\begin{subequations}
\label{eq:linea-transf}
\begin{align}
  V\,\sx + \Zimp\,I &= 0, \label{eq:linea-transf-v}\\
  I\,\sx + \Yadm\,V &= U. \label{eq:linea-transf-i}
\end{align}
\end{subequations}
Note that, since the bicomplex algebra is \emph{commutative}, all the
products in \eqref{eq:ZY}--\eqref{eq:linea-transf} are manipulated with the
usual rules of algebra, taking into account only the relations
$\Bt^2=\Bx^2=-1$, $\Bt\Bx=\Bx\Bt=\Btx$, $\Btx^2=+1$.

\begin{proposicion}[Determinant and distributed transfer function]
\label{prop:transferencia}
Under \cref{hip:linea}, let
\begin{equation}
\label{eq:determinante}
  \Deter(\st,\sx) = \sx^{2} - \Zimp(\st)\,\Yadm(\st).
\end{equation}
At every point $(\st,\sx)$ of the bicomplex domain where $\Deter(\st,\sx)$ is
invertible, the transformed voltage is determined by the injection through
the distributed transfer function
\begin{equation}
\label{eq:G-voltage}
  V(\st,\sx) = -\,\frac{\Zimp(\st)}{\Deter(\st,\sx)}\,U(\st,\sx).
\end{equation}
\end{proposicion}

\begin{proof}
We write the system \eqref{eq:linea-transf} in matrix form over the
commutative subalgebra, with unknowns $(V, I)$:
\begin{equation}
\label{eq:demo-matriz}
  \begin{pmatrix} \sx & \Zimp \\ \Yadm & \sx \end{pmatrix}
  \begin{pmatrix} V \\ I \end{pmatrix}
  = \begin{pmatrix} 0 \\ U \end{pmatrix},
\end{equation}
whose determinant is, by definition \eqref{eq:determinante},
$\sx\,\sx - \Zimp\,\Yadm = \sx^{2} - \Zimp\Yadm = \Deter(\st,\sx)$. Since the
algebra is commutative, Cramer's rule (via the adjugate matrix) is valid
whenever $\Deter$ is invertible. Replacing the first column by the
right-hand side,
\begin{equation}
  V = \Deter^{-1}\,\det\!\begin{pmatrix} 0 & \Zimp \\ U & \sx \end{pmatrix}
    = \Deter^{-1}\,\bigl(0\cdot\sx - \Zimp\,U\bigr)
    = -\,\frac{\Zimp(\st)}{\Deter(\st,\sx)}\,U(\st,\sx),
\end{equation}
which is \eqref{eq:G-voltage}. The derivation requires \emph{only} the
invertibility of $\Deter$ (not that of $\Zimp$): the adjugate of the matrix is
well defined throughout the subalgebra, and commutativity guarantees that the
order of the factors is irrelevant.
\end{proof}

The quantity $\Deter(\st,\sx)$ plays the role of a characteristic
determinant of the distributed model: its zeros define the space--time
eigenmodes of the line and, therefore, its dispersive behavior.

\subsection{Dispersion relation}

The condition $\Deter(\st,\sx)=0$ links the temporal variable $\st$ with the
spatial variable $\sx$ and constitutes the \emph{dispersion relation} of the line:
\begin{equation}
\label{eq:dispersion}
  \Deter(\st,\sx) = \sx^{2} - \Zimp(\st)\,\Yadm(\st) = 0
  \;\Longleftrightarrow\;
  \sx^{2} = \bigl(r+\ell\,\st\bigr)\bigl(g+c\,\st\bigr).
\end{equation}
We consider the ideal lossless case, which admits a transparent physical
interpretation.

\begin{proposicion}[Ideal lossless dispersion]
\label{prop:dispersion-ideal}
Suppose $r=g=0$. Then the dispersion relation \eqref{eq:dispersion}
reduces to
\begin{equation}
\label{eq:disp-sinperdidas}
  \sx^{2} - \ell c\,\st^{2} = 0.
\end{equation}
Evaluating on the geometric axes $\st=\Bt\,\omega$ and $\sx=\Bx\,k$, with
$\omega,k\in\RR$, one obtains the real relation
\begin{equation}
\label{eq:disp-k2}
  k^{2} = \ell c\,\omega^{2},
\end{equation}
from which a constant phase velocity, independent of
frequency, follows:
\begin{equation}
\label{eq:vp}
  v_p = \frac{\omega}{k} = \frac{1}{\sqrt{\ell c}}.
\end{equation}
\end{proposicion}

\begin{proof}
With $r=g=0$, definitions \eqref{eq:ZY} give $\Zimp(\st)=\ell\,\st$ and
$\Yadm(\st)=c\,\st$, so that
$\Zimp(\st)\,\Yadm(\st)=\ell c\,\st^{2}$ and \eqref{eq:dispersion} becomes
\eqref{eq:disp-sinperdidas}. We now evaluate on the geometric axes.
Using $\Bt^{2}=-1$ and $\Bx^{2}=-1$,
\begin{equation}
  \st^{2} = (\Bt\,\omega)^{2} = \Bt^{2}\,\omega^{2} = -\,\omega^{2},
  \qquad
  \sx^{2} = (\Bx\,k)^{2} = \Bx^{2}\,k^{2} = -\,k^{2}.
\end{equation}
Substituting into \eqref{eq:disp-sinperdidas},
\begin{equation}
  -\,k^{2} - \ell c\,(-\,\omega^{2}) = 0
  \;\Longleftrightarrow\;
  -\,k^{2} + \ell c\,\omega^{2} = 0
  \;\Longleftrightarrow\;
  k^{2} = \ell c\,\omega^{2},
\end{equation}
which is \eqref{eq:disp-k2}. Taking the positive root $k=\sqrt{\ell c}\,\omega$
for $\omega>0$, the phase velocity is
\begin{equation}
  v_p = \frac{\omega}{k}
      = \frac{\omega}{\sqrt{\ell c}\,\omega}
      = \frac{1}{\sqrt{\ell c}},
\end{equation}
constant, which proves \eqref{eq:vp}. The relation $k=\sqrt{\ell c}\,\omega$
is linear: the ideal line is \emph{nondispersive}, all components
propagate at the same velocity, and there is no phase distortion.
\end{proof}

\begin{observacion}[Lossy case: exact locus $rc+\ell g=0$ and classical projection]
\label{obs:disp-perdidas}
The exact scope of the evaluation $\st=\Bt\omega$, $\sx=\Bx k$ needs to be
stated with care. On the one hand,
$\Zimp(\Bt\omega)\Yadm(\Bt\omega)=(r+\ell\Bt\omega)(g+c\Bt\omega)$ has, in
general, a component in the $\{1,\Bt\}$ plane,
\begin{equation}
\label{eq:ZY-perdidas}
  \Zimp(\Bt\omega)\Yadm(\Bt\omega)
  =(rg-\ell c\,\omega^2)+\Bt\,\omega\,(rc+\ell g),
\end{equation}
whereas $\sx^2=(\rho_x+\Bx k)^2=(\rho_x^2-k^2)+2\rho_x k\,\Bx$ only contributes
components in $\{1,\Bx\}$; on the pure spatial plane $\sx=\Bx k$, in
particular, $\sx^2=-k^2$ is real. Hence the restricted equation
$\Deter(\Bt\omega,\Bx k)=0$ admits a real branch $k(\omega)$ for all $\omega$
\emph{exactly} if and only if the $\Bt$ term of \eqref{eq:ZY-perdidas}
vanishes identically, that is, on the Heaviside-type locus
\begin{equation}
\label{eq:heaviside}
  rc+\ell g=0,
\end{equation}
in which case $k^2=\ell c\,\omega^2-rg$. The locus \eqref{eq:heaviside}
contains the ideal case $r=g=0$, but \emph{not only} the ideal case: for
$r>0$ it is met at $g=-rc/\ell<0$, a lossy configuration with \emph{active}
(negative) shunt conductance. This sign is not exotic here ---the base case
of \cref{sec:resultados} already uses $g<0$, and the virtual-conductance
feedback of \cref{sec:control} shifts the effective conductance to $g+k_v$---
although the locus itself lies outside the dissipative regime $g+k_v>0$
required there. Off the locus \eqref{eq:heaviside}, the equality
$\sx^2=\Zimp\Yadm$ cannot be satisfied by restricting $\sx$ to the pure
spatial plane, and the lossy curve must be understood as a \emph{classical
projection} of the propagation constant
$\gamma(\omega)=\sqrt{\Zimp(\Bt\omega)\Yadm(\Bt\omega)}$
(its oscillatory part $k(\omega)=\lvert\Imag\gamma(\omega)\rvert$), not as an
exact root of $\Deter=0$ restricted to $\sx=\Bx k$. The exact treatment with
generic losses requires allowing a general bicomplex spatial variable
$\sx=\rho_x+\Bx k$ (with $\rho_x\neq0$, spatial attenuation) or working in the
idempotent decomposition, where each component is an ordinary complex root.
\end{observacion}

The linear relation between $k$ and $\omega$ in the ideal case, and its
curvature under generic losses ($rc+\ell g\neq 0$) in the sense of
\cref{obs:disp-perdidas}, are illustrated in \cref{fig:dispersion}. The fact
that the evaluation on distinct geometric axes ($\Bt$ for time, $\Bx$
for space) preserves both phases separately is precisely what
distinguishes the geometric transform from a scalar complex Laplace transform.

\begin{figure}[htbp]
  \centering
  \includegraphics[width=0.72\textwidth]{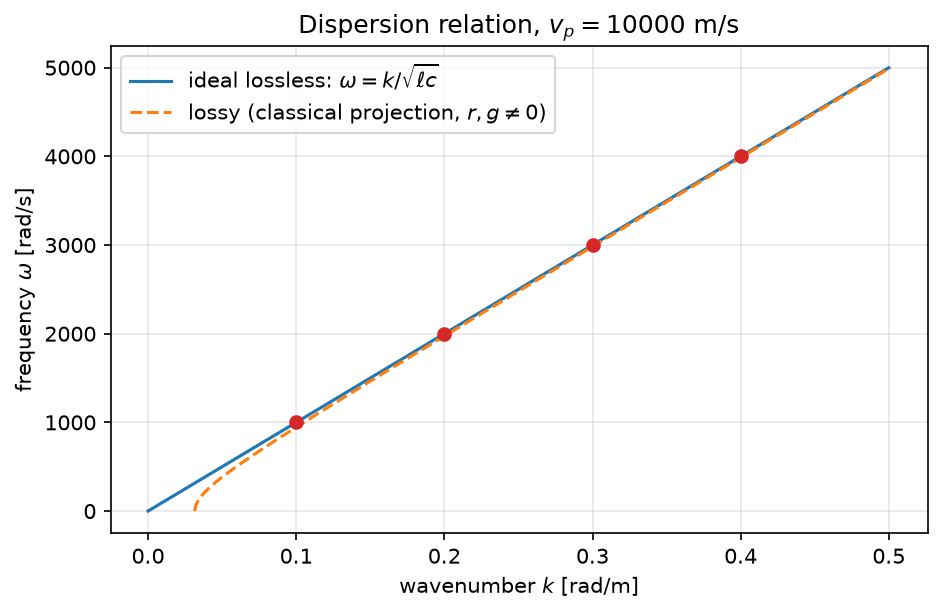}
  \caption{Dispersion relation of the distributed line. The ideal lossless
    straight line ($r=g=0$) is \emph{exact}: $\Deter(\Bt\omega,\Bx k)=0$ gives
    $k=\sqrt{\ell c}\,\omega$ (nondispersive line, constant phase velocity
    $v_p=1/\sqrt{\ell c}$). The lossy curve (dashed) is the
    \emph{classical projection of the lossy propagation constant},
    $k(\omega)=\lvert\Imag\sqrt{\Zimp(\Bt\omega)\Yadm(\Bt\omega)}\rvert$; it is
    shown for reference and does \emph{not} represent an exact root of
    $\Deter(\Bt\omega,\Bx k)=0$ restricted to $\sx=\Bx k$, except on the locus
    $rc+\ell g=0$ of \cref{obs:disp-perdidas} (the values used here give
    $rc+\ell g=1.5\cdot10^{-6}\,\mathrm{s/m^2}\neq0$).}
  \label{fig:dispersion}
\end{figure}

\subsection{Finite line: modes and modal denominator}

We now consider a finite-length line $0\le x\le L$ with zero voltage
at both ends, $v(t,0)=v(t,L)=0$. These boundary conditions
select a basis of sinusoidal spatial modes
\begin{equation}
\label{eq:modos}
  \phi_n(x) = \sin(k_n x),
  \qquad
  k_n = \frac{n\pi}{L},
  \qquad n=1,2,3,\dots
\end{equation}
Rather than formally substituting $\sx^{2}\mapsto -k_n^{2}$, we derive the
modal denominator directly from the transformed PDE. Applying the temporal
transform to the system \eqref{eq:linea-tx} (with zero initial conditions) one
obtains a system in $x$,
\begin{equation}
\label{eq:sistema-x}
  \px V + \Zimp(\st)\,I = 0,
  \qquad
  \px I + \Yadm(\st)\,V = U,
\end{equation}
with $\Zimp=r+\ell\st$, $\Yadm=g+c\st$. Differentiating the first equation with
respect to $x$ and substituting $\px I = U-\Yadm V$ from the second,
\begin{equation}
\label{eq:helmholtz}
  \px^2 V + \Zimp\,(U-\Yadm V)=0
  \;\Longleftrightarrow\;
  \px^2 V - \Zimp\Yadm\,V = -\,\Zimp\,U,
\end{equation}
a space--time Helmholtz equation. We project onto the modes
$\phi_n(x)=\sin(k_n x)$, which satisfy the boundary conditions
$v(t,0)=v(t,L)=0$ and $\px^2\phi_n=-k_n^2\phi_n$. With
$V(\st,x)=\sum_n V_n(\st)\phi_n(x)$ and $U(\st,x)=\sum_n U_n(\st)\phi_n(x)$, the
$n$-th component of \eqref{eq:helmholtz} is
$-k_n^2 V_n-\Zimp\Yadm V_n=-\Zimp U_n$, that is,
\begin{equation}
\label{eq:modo-n}
  \bigl[\,k_n^{2} + \Zimp(\st)\Yadm(\st)\,\bigr]\,V_n(\st)
  = \Zimp(\st)\,U_n(\st).
\end{equation}
The bracketed factor is the \emph{open-loop modal denominator}
\begin{align}
  D_n(\st)
  &= k_n^{2} + \Zimp(\st)\,\Yadm(\st) \notag\\
  &= k_n^{2} + (r+\ell\,\st)(g+c\,\st) \notag\\
  &= \ell c\,\st^{2} + (rc+\ell g)\,\st + (rg + k_n^{2}),
  \label{eq:Dn}
\end{align}
and from \eqref{eq:modo-n} the \emph{modal transfer function} from source to
voltage is
\begin{equation}
\label{eq:Gn}
  G_n(\st)=\frac{V_n(\st)}{U_n(\st)}=\frac{\Zimp(\st)}{D_n(\st)} .
\end{equation}
Note the \emph{numerator} $\Zimp(\st)$: the modal transfer function is not
$1/D_n$, but $\Zimp/D_n$. This is the transfer function that connects with the
modal amplitude of the residual of a fault (\cref{sec:fallos},
$A_n=V_n/G_n=F_f\,\psi_n(x_f)$). The modal denominator \eqref{eq:Dn} is a
second-degree polynomial in $\st$ for each mode $n$; its roots are the
\emph{modal poles} of the open-loop line and determine the stability and
damping of each spatial mode. The shape of the modes $\phi_n(x)$ is
shown in \cref{fig:modeshapes}.

\begin{figure}[htbp]
  \centering
  \includegraphics[width=0.72\textwidth]{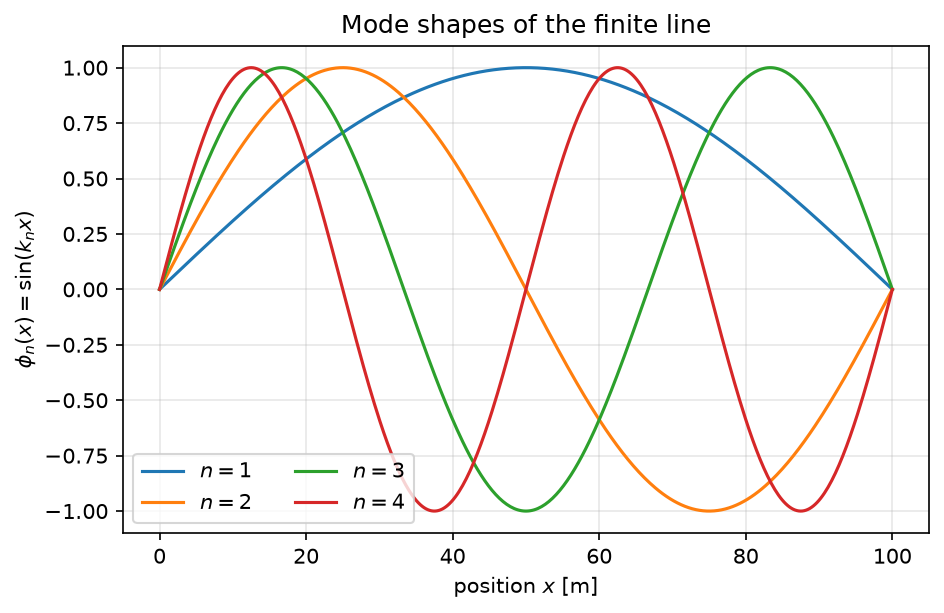}
  \caption{Mode shapes $\phi_n(x)=\sin(k_n x)$, with $k_n=n\pi/L$, of the
    finite line $0\le x\le L$ with zero voltage at the ends. Each spatial
    mode $n$ possesses its own modal denominator $D_n(\st)$ given by
    \eqref{eq:Dn}, whose roots are the modal poles that govern its
    temporal dynamics.}
  \label{fig:modeshapes}
\end{figure}

\subsection{Discrete converters: point injection and spatial sampling}
\label{subsec:inyeccion-discreta}

Hypothesis (H2) of \cref{hip:linea} replaces the physical converters,
located at $P$ discrete sites $0<x_1<\dots<x_P<L$, by a continuous
injection density. We now treat the discrete case exactly. If plant $p$
injects the incremental current $I_p(t)$ (in amperes, not per unit
length) at its site, the source term of \eqref{eq:linea-tx-i} becomes
\begin{equation}
\label{eq:inyeccion-discreta}
  u(t,x)=\sum_{p=1}^{P} I_p(t)\,\delta(x-x_p),
\end{equation}
where $\delta(\cdot)$ is the spatial Dirac delta. This model has the same
mathematical form as the localized fault of \cref{sec:fallos}
(\cref{def:fallo}), so the transform machinery is shared between the two
problems. Indeed, transforming \eqref{eq:inyeccion-discreta} in both
variables, each delta contributes a spatial rotor, by the same computation
as the fault signature \eqref{eq:firma}:
\begin{equation}
\label{eq:U-discreta}
  U(\st,\sx)=\sum_{p=1}^{P} I_p(\st)\,e^{-\sx x_p},
\end{equation}
a finite sum of spatial rotors, entire in $\sx$.

On the finite line, the sine-series coefficients of
\eqref{eq:inyeccion-discreta} in the convention of \eqref{eq:modo-n}
---that is, $U(\st,x)=\sum_n U_n(\st)\phi_n(x)$ with
$U_n=(2/L)\int_0^L U\phi_n\,\dd x$--- are
\begin{equation}
\label{eq:Un-discreta}
  U_n(\st)=\frac{2}{L}\sum_{p=1}^{P} I_p(\st)\,\sin(k_n x_p)
          =\frac{2}{L}\bigl[S\,I(\st)\bigr]_n,
  \qquad
  S_{np}=\sin(k_n x_p),
\end{equation}
where $S\in\RR^{N\times P}$ is the \emph{sampling matrix} of the converter
sites on the first $N$ modes. How faithfully the discrete injection can
reproduce, or actuate, the modal content of the line is entirely encoded
in $S$. For the uniform grid of practical interest the sampling is exact
in the following sense.

\begin{proposicion}[Exact sampling on the uniform midpoint grid]
\label{prop:muestreo}
Let the $P$ converters be placed on the midpoint grid
$x_p=(p-\tfrac12)L/P$, $p=1,\dots,P$, and write
$\theta_p=k_1x_p=(p-\tfrac12)\pi/P$, so that $k_nx_p=n\theta_p$. Then:
\begin{enumerate}[label=(\roman*),leftmargin=*]
  \item \textbf{Discrete orthogonality.} For $1\le n,m\le P$,
    \begin{equation}
    \label{eq:ortogonalidad-dst}
      \sum_{p=1}^{P}\sin(k_n x_p)\sin(k_m x_p)
      =\begin{cases}
        \dfrac{P}{2}\,\delta_{nm}, & n,m\le P-1,\\[1.2ex]
        P, & n=m=P,\\[0.4ex]
        0, & n\neq m,
      \end{cases}
    \end{equation}
    that is, the sampled mode shapes of orders $1,\dots,P$ remain
    orthogonal, with mode $P$ carrying twice the weight of the others.
  \item \textbf{Spatial aliasing.} For every $p$, every $n\ge1$, and every
    integer $j\ge1$,
    \begin{equation}
    \label{eq:alias}
      \sin(k_{2P-n}x_p)=\sin(k_n x_p),
      \quad
      \sin(k_{2P+n}x_p)=-\sin(k_n x_p),
      \quad
      \sin(k_{2Pj}x_p)=0.
    \end{equation}
    The sampled shapes are therefore $4P$-periodic in $n$, the modes
    $n=2Pj$ are invisible to the converter set (\emph{dead} modes), and
    every mode $n>P$ \emph{not} of the form $2Pj$ coincides at the sites,
    up to sign, with a canonical alias in $\{1,\dots,P\}$.
\end{enumerate}
\end{proposicion}

\begin{proof}
The key computation is the Dirichlet-kernel sum
$\Sigma_q:=\sum_{p=1}^{P}\cos(q\theta_p)$ for integer $q\ge0$. Let
$w=e^{iq\pi/(2P)}$ (an ordinary complex exponential: the geometric algebra
plays no role in this scalar lemma). Since $q\theta_p=q(2p-1)\pi/(2P)$, we
have $\cos(q\theta_p)=\Real w^{2p-1}$ and the sum is geometric in $w^2$.
If $q=2Pj$, then $w^{2p-1}=e^{ij(2p-1)\pi}=(-1)^{j}$ for every $p$, so
$\Sigma_{2Pj}=P(-1)^{j}$ and, taking imaginary parts,
$\sum_p\sin(2Pj\,\theta_p)=0$. If $q\not\equiv0\pmod{2P}$, then
$w^2\neq1$ and
\begin{equation}
\label{eq:demo-dirichlet}
  \sum_{p=1}^{P}w^{2p-1}
  =w\,\frac{w^{2P}-1}{w^{2}-1},
  \qquad w^{2P}=e^{iq\pi}=(-1)^{q}.
\end{equation}
For $q$ even, $w^{2P}=1$ and the sum vanishes identically. For $q$ odd,
$w^{2P}=-1$ and \eqref{eq:demo-dirichlet} equals
$-2w/(w^{2}-1)=-2/(w-w^{-1})=i/\sin\bigl(q\pi/(2P)\bigr)$, which is purely
imaginary; hence $\Sigma_q=0$ and, as a byproduct,
$\sum_p\sin(q\theta_p)=1/\sin\bigl(q\pi/(2P)\bigr)$ for $q$ odd. In
summary, $\Sigma_q=0$ unless $q=2Pj$, in which case $\Sigma_q=P(-1)^{j}$.

(i) By the product-to-sum identity,
$\sin(n\theta_p)\sin(m\theta_p)
=\tfrac12\bigl[\cos((n-m)\theta_p)-\cos((n+m)\theta_p)\bigr]$, so the
left-hand side of \eqref{eq:ortogonalidad-dst} equals
$\tfrac12(\Sigma_{n-m}-\Sigma_{n+m})$ (with $\Sigma_{-q}=\Sigma_q$). For
$1\le n,m\le P$ one has $|n-m|\le P-1<2P$, so $\Sigma_{n-m}$ contributes
$P\,\delta_{nm}$ (the multiple $q=0$, $j=0$), and $2\le n+m\le2P$, so the
only multiple of $2P$ is $n+m=2P$, attained within the range only at
$n=m=P$, where $\Sigma_{2P}=-P$. Combining: $\tfrac12(P-0)=P/2$ for
$n=m\le P-1$; $\tfrac12(P-(-P))=P$ for $n=m=P$; and $0$ otherwise.

(ii) On the midpoint grid, $2P\theta_p=(2p-1)\pi$, so
$\sin(2P\theta_p)=0$ and $\cos(2P\theta_p)=-1$. The angle-addition
formulas then give
$\sin\bigl((2P\mp n)\theta_p\bigr)
=\mp\cos(2P\theta_p)\sin(n\theta_p)=\pm\sin(n\theta_p)$,
which are the first two identities of \eqref{eq:alias}, and
$\sin(2Pj\,\theta_p)=\sin\bigl(j(2p-1)\pi\bigr)=0$ is the third. Applying
the reflection $n\mapsto2P-n$ and the shift $n\mapsto n+2P$ repeatedly
shows that the sampled shapes are $4P$-periodic in $n$ and that every mode
$n>P$ folds onto a canonical alias in $\{1,\dots,P\}$ up to sign.
\end{proof}

\begin{corolario}[Continuum limit]
\label{cor:limite-continuo}
Let $\bar u(t,\cdot)\in C^2[0,L]$ be a target injection density and let
each converter on the midpoint grid carry its cell share,
$I_p(t)=\bar u(t,x_p)\,L/P$. Then, for every fixed $n$,
\begin{equation}
  U_n^{\mathrm{disc}}(t)
  =\frac{2}{P}\sum_{p=1}^{P}\bar u(t,x_p)\sin(k_n x_p)
  \;\xrightarrow[P\to\infty]{}\;
  U_n^{\mathrm{cont}}(t)
  =\frac{2}{L}\int_0^L \bar u(t,x)\sin(k_n x)\,\dd x,
\end{equation}
with error $O(P^{-2})$. For a constant density $\bar u\equiv u_0$ the
comparison is exact in closed form: for $n$ odd,
\begin{equation}
\label{eq:Un-cerrada}
  U_n^{\mathrm{disc}}
  =\frac{2u_0/P}{\sin\bigl(n\pi/(2P)\bigr)},
  \qquad
  U_n^{\mathrm{cont}}=\frac{4u_0}{n\pi},
  \qquad
  \frac{U_n^{\mathrm{disc}}}{U_n^{\mathrm{cont}}}
  =\frac{n\pi/(2P)}{\sin\bigl(n\pi/(2P)\bigr)}\ge1,
\end{equation}
and both coefficients vanish for $n$ even. Hypothesis (H2) of
\cref{hip:linea} thereby acquires a quantitative convergence rate and an
exact special case.
\end{corolario}

\begin{proof}
The sum $U_n^{\mathrm{disc}}$ is the composite midpoint rule applied to
the integral $U_n^{\mathrm{cont}}$ on $P$ equal cells of width $L/P$; the
classical midpoint-rule error bound for $C^2$ integrands gives the
$O(P^{-2})$ rate. For $\bar u\equiv u_0$,
$U_n^{\mathrm{disc}}=(2u_0/P)\sum_p\sin(n\theta_p)$, and the row sums
computed in the proof of \cref{prop:muestreo} give
$\sum_p\sin(n\theta_p)=1/\sin\bigl(n\pi/(2P)\bigr)$ for $n$ odd and $0$
for $n$ even. On the continuous side,
$U_n^{\mathrm{cont}}=(2u_0/L)\int_0^L\sin(k_nx)\,\dd x
=2u_0\bigl(1-(-1)^n\bigr)/(n\pi)$, which is $4u_0/(n\pi)$ for $n$ odd and
$0$ for $n$ even. The ratio in \eqref{eq:Un-cerrada} follows, and the
bound $\sin\xi\le\xi$ for $\xi\in[0,\pi/2]$ shows it is $\ge1$, with limit
$1$ as $n/P\to0$.
\end{proof}

Whether a given mode can be actuated at all from the converter sites is
quantified by the \emph{modal controllability score}
\begin{equation}
\label{eq:score-controlabilidad}
  \gamma_n
  =\Biggl(\frac{1}{P}\sum_{p=1}^{P}\sin^2(k_n x_p)\Biggr)^{1/2}
  \in[0,1],
\end{equation}
which vanishes if and only if every converter sits on a node of mode $n$;
in that case mode $n$ is \emph{uncontrollable} from the converter set,
whatever the control law. The situation is not pathological: on the
uniform \emph{node} grid $x_p=pL/(P+1)$ one has
$\sin(k_{P+1}x_p)=\sin(p\pi)=0$ for all $p$, so mode $n=P+1$ (and its
multiples) is exactly uncontrollable. Together with
\cref{prop:muestreo}(ii), this expresses a spatial Nyquist limit of
actuation: from $P$ converter sites at most the $P$ lowest modes can be
actuated independently, and any injected content in modes $n>P$ folds back
onto the sub-Nyquist band $n\le P$. The score
\eqref{eq:score-controlabilidad} is the actuation dual of the modal
observability metric of \cref{obs:observabilidad} in \cref{sec:fallos},
which plays the same role on the sensor side.

\begin{observacion}[Correspondence with the implementation]
\label{obs:codigo-linea}
The expressions in this section correspond directly to the module
\texttt{geolaplace.line\_model} of the reference implementation. In
particular, $\Zimp(\st)$ and $\Yadm(\st)$ from \eqref{eq:ZY} are computed in the
routines \texttt{Z} and \texttt{Y}; the determinant $\Deter(\st,\sx)$ from
\eqref{eq:determinante} in \texttt{Delta}; the distributed transfer function
\eqref{eq:G-voltage} in \texttt{G\_voltage}; the modal denominator
$D_n(\st)$ from \eqref{eq:Dn} in \texttt{modal\_denominator}; the modal
transfer function $G_n(\st)=\Zimp(\st)/D_n(\st)$ from \eqref{eq:Gn} in
\texttt{modal\_transfer} (with numerator $\Zimp$, not $1/D_n$); its roots in
\texttt{modal\_poles}; and the ideal dispersion relation \eqref{eq:disp-k2} in
\texttt{dispersion\_ideal}. The discrete-converter results of
\cref{subsec:inyeccion-discreta} (and \cref{prop:control-droop-discreto} of
\cref{sec:control}) are implemented by the routines
\path{converter_positions}, \path{sampling_matrix}, \path{modal_injection},
\path{modal_coupling_matrix}, \path{modal_controllability_score},
\path{alias_mode}, and \path{discrete_droop_poles} of the module
\texttt{geolaplace.converters}. This correspondence makes it possible to
reproduce and verify numerically each analytical result.
\end{observacion}

\section{Distributed control via admittance shaping}
\label{sec:control}

Starting from the space-time dispersion relation of the line of distributed
converters obtained in \cref{sec:linea}, we address the design of a control law
that acts on the shunt injection of each node. The idea, inspired by the virtual
impedance/admittance control common in DC microgrids
\cite{dragicevic2016dc,meng2017dynamics}, is to \emph{shape} the effective shunt
admittance of the line in order to modify its dispersion relation and, with it,
its closed-loop dynamic behavior. The geometric Laplace transform is
particularly convenient here because it keeps the temporal and spatial phases
separate, so that a single controller can be expressed compactly as a rational
function $K(\st,\sx)$ of the two geometric variables.

Let us recall the incremental line model of \cref{sec:linea},
\begin{equation}
  \px v + \ell\,\pt i + r\,i = 0,
  \qquad
  \px i + c\,\pt v + g\,v = u,
  \label{eq:control-line}
\end{equation}
whose geometric transform, with the longitudinal impedance
$\Zimp(\st)=r+\ell\st$ and the shunt admittance $\Yadm(\st)=g+c\st$, leads to the
algebraic system $V\sx+\Zimp I=0$, $I\sx+\Yadm V=U$ and to the dispersion
relation
\begin{equation}
  \Deter(\st,\sx)=\sx^2-\Zimp(\st)\,\Yadm(\st).
  \label{eq:control-delta-open}
\end{equation}
The injection $u$ is now the \emph{actuation variable} of the converters.

\subsection*{Control law and closed-loop admittance}

We consider a feedback law on the local voltage of the form
\begin{equation}
  u = r_u - K[v],
  \label{eq:control-law-time}
\end{equation}
where $r_u$ is a reference (setpoint or scheduled injection) and $K[\cdot]$ is a
linear space-time operator that is causal in time. In the transformed domain,
\cref{eq:control-law-time} reads
\begin{equation}
  U = R_u - K(\st,\sx)\,V,
  \label{eq:control-law-transf}
\end{equation}
with $R_u=\transf{r_u}$ (the transform of the reference; the subscript $u$
distinguishes it from the fault residual $R$ of \cref{sec:fallos}). Substituting
\cref{eq:control-law-transf} into the second algebraic equation of the model,
$I\sx+\Yadm V=U$, we obtain $I\sx+[\Yadm(\st)+K(\st,\sx)]V=R_u$, so that the
controller acts \emph{additively} on the shunt admittance. We define the
\emph{closed-loop admittance}
\begin{equation}
  \Yadm_c(\st,\sx)=g+c\st+K(\st,\sx)=\Yadm(\st)+K(\st,\sx).
  \label{eq:control-Yc}
\end{equation}
Eliminating the current with $V\sx+\Zimp I=0$, the closed-loop determinant
inherits the same structure as \cref{eq:control-delta-open}, with $\Yadm$ simply
replaced by $\Yadm_c$:
\begin{equation}
  \boxed{\;\Deterc(\st,\sx)=\sx^2-\Zimp(\st)\,\bigl[\Yadm(\st)+K(\st,\sx)\bigr].\;}
  \label{eq:control-delta-closed}
\end{equation}
The entire closed-loop space-time dynamics is encoded in the zeros of $\Deterc$;
the design therefore amounts to choosing $K$ so as to place those zeros
conveniently.

\subsection{Distributed virtual conductance}
\label{subsec:control-virtual}

The simplest case is a pure proportional voltage feedback, $K=k_v$ constant,
which in the physical domain corresponds to the local law $u=r_u-k_v\,v$.
Physically, each converter injects a current proportional to its local voltage,
behaving like an additional conductance distributed along the line.

\begin{proposicion}[Virtual conductance]
\label{prop:control-virtual}
Let $K(\st,\sx)=k_v$ with $k_v\in\RR$. Then the closed-loop shunt admittance is
\begin{equation}
  \Yadm_c(\st,\sx)=(g+k_v)+c\st,
  \label{eq:control-virtual-Yc}
\end{equation}
that is, it is equivalent to the admittance of an \emph{identical} line in which
the distributed conductance changes from $g$ to $g+k_v$, and the closed-loop
determinant is
\begin{equation}
  \Deterc(\st,\sx)=\sx^2-\Zimp(\st)\,\bigl[(g+k_v)+c\st\bigr].
  \label{eq:control-virtual-delta}
\end{equation}
If $k_v>0$, the incremental shunt conductance increases by $k_v$ (from $g$ to
$g+k_v$) and, with it, the modal damping. If in addition $g+k_v>0$, the
effective shunt branch is strictly dissipative (incrementally passive).
\end{proposicion}

\begin{proof}
Substituting $K=k_v$ into \cref{eq:control-Yc} immediately yields
\cref{eq:control-virtual-Yc}, since $\Yadm(\st)+k_v=(g+c\st)+k_v=(g+k_v)+c\st$,
which is the admittance of a line with conductance $g+k_v$ and the same
capacitance $c$. Expression \cref{eq:control-virtual-delta} follows from carrying
\cref{eq:control-virtual-Yc} into \cref{eq:control-delta-closed}. As for the
effect on damping, the \emph{incremental conductive term} of the shunt branch in
\cref{eq:control-line} under the law $u=r_u-k_v v$ changes from $g\,v^2$ to
$(g+k_v)\,v^2$, whose coefficient is strictly larger when $k_v>0$. This increase
in the conductive coefficient raises the $\st$-term of the modal denominator (by
$\ell k_v$, see \cref{eq:control-modal-Dc-expanded}), which \emph{increases the
modal damping} and shifts the poles toward more negative real parts (see
\cref{subsec:control-modal}). That the conductive coefficient increases does not,
by itself, imply net passivity of the branch: this additionally requires
$g+k_v>0$.
\end{proof}

\begin{observacion}
The virtual conductance increases the conductance by $k_v$ ($g\mapsto g+k_v$) and
shifts the poles toward greater damping; it is implemented with a single local
voltage measurement per node, without additional dynamics or spatial couplings.
The ``dissipative'' character needs one qualification: the effective shunt branch
is strictly \emph{passive} only when the resulting conductance is positive,
$g+k_v>0$. In the base case ($g=-0.05\unit{S/m}$, $k_v=0.04\unit{S/m}$) one has
$g+k_v=-0.01\unit{S/m}<0$: the branch remains incrementally active, although the
control suffices to stabilize the modes considered (see \cref{sec:resultados}).
Its limitation is that it only shifts the constant term $g$, without acting on
the capacitive inertia $c$ or on the spatial structure of the line.
\end{observacion}

The virtual conductance admits an \emph{exact} discrete realization: when the
injection is produced by $P$ converters at discrete sites, as in
\cref{subsec:inyeccion-discreta} of \cref{sec:linea}, the natural per-converter
droop law reproduces the continuous closed loop without any homogenization
error on the sub-Nyquist band.

\begin{proposicion}[Exact discrete realization of the virtual conductance]
\label{prop:control-droop-discreto}
Let the injection be produced by $P$ converters at the sites
$x_p=(p-\tfrac12)L/P$ of the uniform midpoint grid, each implementing the
local droop law
\begin{equation}
  I_p(t)=-\,\kappa_p\,v(t,x_p),
  \qquad
  \kappa_p=k_v\,\frac{L}{P},
  \label{eq:droop-discreto}
\end{equation}
that is, the virtual conductance $k_v$ (in $\unit{S/m}$) aggregated over the
cell of length $L/P$ served by converter $p$. Expanding $v$ in the modes
$\phi_n(x)=\sin(k_nx)$ and retaining $N$ modes, the feedback contributes
$U_n^{\mathrm{fb}}=-\sum_{m=1}^{N}C_{nm}V_m$ to the modal series coefficients
\eqref{eq:Un-discreta}, with the coupling matrix
\begin{equation}
  C=\frac{2}{L}\,S\,\operatorname{diag}(\kappa_1,\dots,\kappa_P)\,
    S^{\mathsf T}\in\RR^{N\times N},
  \qquad
  S_{np}=\sin(k_nx_p),
  \label{eq:acoplo-alias}
\end{equation}
so that the closed-loop modal system is
$D_n(\st)V_n+\Zimp(\st)\sum_{m}C_{nm}V_m=\Zimp(\st)R_{u,n}$. On the midpoint
grid with the gains \eqref{eq:droop-discreto}:
\begin{enumerate}[label=(\roman*),leftmargin=*]
  \item if $N\le P$, the matrix $C$ is exactly diagonal,
        $C=k_v\operatorname{diag}(1,\dots,1)$ for $N\le P-1$ and
        $C=k_v\operatorname{diag}(1,\dots,1,2)$ for $N=P$; in particular,
        for every mode $n\le P-1$ the closed-loop modal denominator is
        \emph{exactly} $D_n(\st)+k_v\,\Zimp(\st)$, identical to the
        continuous virtual conductance of \cref{prop:control-virtual},
        while mode $P$ sees the doubled gain $2k_v$;
  \item if $N>P$, the rows of $S$ repeat up to sign and the feedback
        couples each mode with its alias partner with full diagonal
        strength; for example, $C_{n,2P-n}=k_v$ for $1\le n\le P-1$
        whenever $2P-n\le N$.
\end{enumerate}
\end{proposicion}

\begin{proof}
Expand the voltage in the modes, so that at each site
$v(t,x_p)=\sum_{m}V_m(t)\sin(k_mx_p)$, and let the droop law
\eqref{eq:droop-discreto} set $I_p=-\kappa_p\,v(t,x_p)$. Substituting the
resulting currents into the series coefficients \eqref{eq:Un-discreta},
\begin{equation}
  U_n^{\mathrm{fb}}
  =\frac{2}{L}\sum_{p=1}^{P}I_p\sin(k_nx_p)
  =-\sum_{m=1}^{N}
    \underbrace{\frac{2}{L}\sum_{p=1}^{P}\kappa_p\,
    \sin(k_nx_p)\sin(k_mx_p)}_{C_{nm}}\,V_m,
\end{equation}
which is \eqref{eq:acoplo-alias}. Carrying $U_n=R_{u,n}+U_n^{\mathrm{fb}}$
into the modal equation \eqref{eq:modo-n},
$D_nV_n=\Zimp\,(R_{u,n}-\sum_mC_{nm}V_m)$, gives the stated closed-loop
system. With constant gains $\kappa_p=k_vL/P$,
$C_{nm}=(2k_v/P)\sum_p\sin(k_nx_p)\sin(k_mx_p)$, and the discrete
orthogonality \eqref{eq:ortogonalidad-dst} of \cref{prop:muestreo}(i)
evaluates the sum exactly: $C_{nm}=k_v\delta_{nm}$ for $n,m\le P-1$,
$C_{PP}=2k_v$, and zero cross terms, which proves (i); for diagonal $C$ the
$n$-th equation decouples as $[D_n(\st)+k_v\Zimp(\st)]V_n=\Zimp(\st)R_{u,n}$,
the modal denominator of \cref{prop:control-virtual} (equivalently,
\eqref{eq:control-modal-Dc} with $K_n=k_v$). For (ii), the alias identities
\eqref{eq:alias} of \cref{prop:muestreo}(ii) give
$\sin(k_{2P-n}x_p)=\sin(k_nx_p)$ for all $p$, hence
$C_{n,2P-n}=(2k_v/P)\sum_p\sin^2(k_nx_p)=k_v$ for $1\le n\le P-1$ by
\eqref{eq:ortogonalidad-dst}.
\end{proof}

\begin{observacion}
Three consequences follow. (a) \cref{prop:control-droop-discreto}
upgrades the homogenization hypothesis from an approximation to an exact
statement: for uniformly placed converters with the droop law
\eqref{eq:droop-discreto}, the continuous closed-loop analysis of this
section is \emph{exact} on the sub-Nyquist band $n\le P-1$, not merely a
limit. (b) The validity of that reading requires the modal content of
interest to be confined to $n<P$: modes beyond the spatial actuation
Nyquist limit are aliased and coupled by the feedback, the actuation
counterpart of the sensor-side observability restriction
(\cref{obs:observabilidad} in \cref{sec:fallos}). (c) For non-uniform
placements, the natural choice $\kappa_p=k_v w_p$, with $w_p$ the widths of
the Voronoi cells of the sites, preserves the convergence to the continuous
loop as $P$ grows, in the quadrature sense of \cref{cor:limite-continuo},
but the exactness is lost: $C$ in \eqref{eq:acoplo-alias} is no longer
diagonal. The numerical behavior of both situations is examined in
\cref{subsec:discreto-result}.
\end{observacion}

\subsection{Space-time dynamic controller}
\label{subsec:control-dynamic}

To shape the admittance more fully, it is convenient to allow a controller that
depends on both geometric variables affinely in each and in their product:
\begin{equation}
  K(\st,\sx)=k_0+k_t\st+k_x\sx+k_{tx}\st\sx.
  \label{eq:control-dynamic-K}
\end{equation}
Since the algebra is commutative (bicomplex), the product $\st\sx$ is well
defined and commutes, so that $K$ is an ordinary multivariate polynomial.
Recalling that $\st\leftrightarrow\pt$ and $\sx\leftrightarrow\px$ under the
geometric transform with zero initial and boundary conditions for the
operational properties (see \cref{sec:transformada}), the law
\cref{eq:control-dynamic-K} corresponds in the physical domain to the operator
\begin{equation}
  K[v]=k_0\,v+k_t\,\pt v+k_x\,\px v+k_{tx}\,\pt\px v.
  \label{eq:control-dynamic-Ktime}
\end{equation}
Each term admits a direct physical interpretation:
\begin{itemize}
  \item $k_0\,v$ ($k_0$ term): constant virtual \emph{conductance}, as in
        \cref{subsec:control-virtual}; it modifies $g$.
  \item $k_t\,\pt v$ ($k_t\st$ term): \emph{temporal damping}, an injection
        proportional to the temporal rate of change of the local voltage; it
        acts on the effective capacitance $c$, modifying the inertia of the node.
  \item $k_x\,\px v$ ($k_x\sx$ term): first-order \emph{spatial coupling},
        proportional to the spatial gradient of the voltage; it is equivalent to
        an injection that compares the voltage with that of the neighboring nodes
        and contributes to a controlled transport or diffusion along the line.
  \item $k_{tx}\,\pt\px v$ ($k_{tx}\st\sx$ term): \emph{mixed space-time} term,
        simultaneously sensitive to the temporal variation and to the spatial
        gradient; it allows selective damping of the propagating components
        (traveling waves) without penalizing the quasi-static modes.
\end{itemize}

\begin{observacion}
The presence of the $\sx$ and $\st\sx$ terms is precisely what distinguishes
this design from a purely temporal controller: they require spatial information
(gradients), so they are realizable only if the nodes can estimate
neighbor-related quantities. With a single independent variable
($k_x=k_{tx}=0$) the controller reduces to $K=k_0+k_t\st$, equivalent to a
classical Laplace compensator with a geometric interpretation; the space-time
contribution appears only when the $\sx$ terms are activated.
\end{observacion}

\subsection{Design by target denominator}
\label{subsec:control-design}

The full shaping is formulated inversely: a desired closed-loop determinant
$\Deterd(\st,\sx)$ is specified and the controller that realizes it is solved
for. Setting $\Deterc=\Deterd$ in \cref{eq:control-delta-closed} and using
\cref{eq:control-delta-open},
\begin{equation}
  \sx^2-\Zimp(\st)\bigl[\Yadm(\st)+K(\st,\sx)\bigr]=\Deterd(\st,\sx),
\end{equation}
and since $\Deter(\st,\sx)=\sx^2-\Zimp(\st)\Yadm(\st)$, subtracting yields
$\Zimp(\st)\,K(\st,\sx)=\Deter(\st,\sx)-\Deterd(\st,\sx)$, whence

\begin{proposicion}[Controller by target denominator]
\label{prop:control-design}
Let $\Deterd(\st,\sx)$ be a target determinant and assume $\Zimp(\st)=r+\ell\st$
invertible in the region of interest (which holds except at $\st=-r/\ell$). Then
the unique controller of the form \cref{eq:control-law-transf} that imposes
$\Deterc=\Deterd$ is
\begin{equation}
  \boxed{\;K(\st,\sx)=\frac{\Deter(\st,\sx)-\Deterd(\st,\sx)}{\Zimp(\st)}.\;}
  \label{eq:control-design-K}
\end{equation}
\end{proposicion}

\begin{proof}
The equality $\Deterc=\Deterd$ is equivalent, as shown above, to
$\Zimp(\st)K=\Deter-\Deterd$. Since $\Zimp(\st)$ is invertible in the bicomplex
subalgebra except at its zeros, $K$ is solved for by multiplying by
$\Zimp(\st)^{-1}$, which gives \cref{eq:control-design-K}. Uniqueness is
immediate: any $K'$ satisfying $\Deterc=\Deterd$ satisfies
$\Zimp(\st)(K-K')=0$, and since $\Zimp(\st)$ is invertible we conclude $K=K'$.
\end{proof}

\begin{observacion}[Implementability]
\label{obs:control-implement}
Formula \cref{eq:control-design-K} is exact, but its physical realization is
subject to constraints that must be checked case by case:
\begin{itemize}
  \item \textbf{Temporal causality.} The quotient
        $[\Deter-\Deterd]/\Zimp(\st)$ must correspond, as a function of $\st$, to
        a causal and proper operator; otherwise it contains time derivatives of
        order higher than what is implementable and must be approximated by a
        causal filter (for example, derivatives with a low-pass filter).
  \item \textbf{Local or neighbor-based spatial measurements.} The dependence on
        $\sx$ translates into spatial derivatives $\px$, which each converter can
        reconstruct only from its own measurements and those of its immediate
        neighbors. High spatial orders demand a wide neighborhood and degrade the
        locality of the control.
  \item \textbf{Realization by discrete converters.} The real line is a
        \emph{discrete} set of nodes, but the impact depends on the term. The
        static droop ($k_v$, no spatial derivative) is realized \emph{exactly}
        by per-converter gains on the uniform midpoint grid
        (\cref{prop:control-droop-discreto}); the discretization error
        concerns the dynamic space-time terms (the $\sx$-dependent part of
        $K$), whose operators $\px$ and $\pt\px$ are approximated by finite
        differences between adjacent nodes, introducing an error that limits
        the effective spatial bandwidth.
  \item \textbf{Saturation and communication delays.} The injections are bounded
        by the capability of the converters, and the exchange of neighbor-based
        measurements incurs communication delays; both effects limit the
        aggressiveness of the achievable $\Deterd$ and may compromise stability
        if they are not accounted for in the design.
\end{itemize}
For these reasons, in practice one prefers to restrict $\Deterd$ to low-order
forms (those of \cref{subsec:control-dynamic}) and to verify a posteriori the
robustness against discretization, saturation, and delay.
\end{observacion}

\subsection{Modal control on a finite line}
\label{subsec:control-modal}

For the finite line $0\le x\le L$ with zero voltage at the ends, the
decomposition into the modes $\phi_n(x)=\sin(k_n x)$, $k_n=n\pi/L$, of
\cref{sec:linea} reduces the space-time problem to a family of independent
temporal problems. Each mode $n$ is controlled with a modal feedback law
\begin{equation}
  U_n=R_{u,n}-K_n(\st)\,V_n,
  \label{eq:control-modal-law}
\end{equation}
where $K_n(\st)$ is the projection of the controller onto mode $n$ and depends
only on the temporal variable $\st$. Recalling that the open-loop modal
denominator is $D_n(\st)=k_n^2+\Zimp(\st)\Yadm(\st)$, the effect of
\cref{eq:control-modal-law} on the denominator is obtained, as in
\cref{eq:control-delta-closed}, by adding $K_n$ to the admittance and keeping the
impedance factor $\Zimp(\st)=r+\ell\st$:
\begin{equation}
  D_{c,n}(\st)=D_n(\st)+(r+\ell\st)\,K_n(\st).
  \label{eq:control-modal-Dc}
\end{equation}

We specialize to a modal controller that is affine in time,
$K_n(\st)=k_{0n}+k_{tn}\st$, which combines virtual conductance ($k_{0n}$) and
temporal damping ($k_{tn}$). We expand \cref{eq:control-modal-Dc}. With
$\Zimp(\st)\Yadm(\st)=(r+\ell\st)(g+c\st)=\ell c\,\st^2+(rc+\ell g)\st+rg$, the
open-loop denominator is $D_n(\st)=\ell c\,\st^2+(rc+\ell g)\st+(k_n^2+rg)$. The
feedback term equals
\begin{align}
  (r+\ell\st)\,K_n(\st)
  &=(r+\ell\st)(k_{0n}+k_{tn}\st)\notag\\
  &=\ell k_{tn}\,\st^2+(r k_{tn}+\ell k_{0n})\,\st+r k_{0n}.
  \label{eq:control-modal-fb}
\end{align}
Adding \cref{eq:control-modal-fb} to $D_n(\st)$ and grouping by powers of $\st$,
\begin{equation}
  \boxed{\;D_{c,n}(\st)=\ell(c+k_{tn})\,\st^2
        +\bigl[r(c+k_{tn})+\ell(g+k_{0n})\bigr]\,\st
        +\bigl[k_n^2+r(g+k_{0n})\bigr].\;}
  \label{eq:control-modal-Dc-expanded}
\end{equation}
One observes that $k_{tn}$ shifts the effective capacitance $c\mapsto c+k_{tn}$
(inertia/damping) and $k_{0n}$ shifts the effective conductance
$g\mapsto g+k_{0n}$, exactly as anticipated by the interpretations of
\cref{subsec:control-dynamic}, now at the level of each mode.

\subsection*{Assignment of temporal geometric poles}

The natural modal design objective is to place the two temporal poles of mode
$n$ at a prescribed geometric conjugate pair
\begin{equation}
  \st{}_{,n}^{\pm}=-\alpha_n\pm\Bt\beta_n,
  \qquad \alpha_n>0,\ \beta_n\ge0,
  \label{eq:control-target-poles}
\end{equation}
where $\alpha_n$ sets the damping and $\beta_n$ the temporal geometric frequency.
The monic second-degree polynomial with those roots, using $\Bt^2=-1$, is
\begin{equation}
  (\st-\st{}_{,n}^{+})(\st-\st{}_{,n}^{-})
  =\st^2+2\alpha_n\st+(\alpha_n^2+\beta_n^2).
  \label{eq:control-target-poly}
\end{equation}
We require that $D_{c,n}$ of \cref{eq:control-modal-Dc-expanded}, once normalized
by its leading coefficient $\ell(c+k_{tn})$, coincide with
\cref{eq:control-target-poly}.

\begin{proposicion}[Gains for modal pole assignment]
\label{prop:control-poles}
Assume $\ell>0$ and a target $(\alpha_n,\beta_n)$ such that
\begin{equation}
  \Lambda(\alpha_n,\beta_n)
  :=(\alpha_n^2+\beta_n^2)\ell^2-2\alpha_n r\ell+r^2\neq0
  \label{eq:control-poles-denom}
\end{equation}
(a scalar condition, not to be confused with the modal denominator $D_n$).
Let us define the effective closed-loop quantities $c_t:=c+k_{tn}$ and
$g_k:=g+k_{0n}$. Then the modal controller $K_n(\st)=k_{0n}+k_{tn}\st$ that places
the poles of mode $n$ at \cref{eq:control-target-poles} is given by
\begin{equation}
  c_t=\frac{k_n^2\,\ell}
           {(\alpha_n^2+\beta_n^2)\ell^2-2\alpha_n r\ell+r^2},
  \qquad
  g_k=\frac{c_t\,(2\alpha_n\ell-r)}{\ell},
  \label{eq:control-poles-gains}
\end{equation}
whence $k_{tn}=c_t-c$ and $k_{0n}=g_k-g$. The solution is physical, that is,
implementable as positive effective capacitance and conductance, only if
$c_t>0$ and $g_k>0$; otherwise \cref{eq:control-poles-gains} is still the
algebraic solution, but its realizability must be checked separately.
\end{proposicion}

\begin{proof}
Normalizing \cref{eq:control-modal-Dc-expanded} by $\ell c_t=\ell(c+k_{tn})$, the
closed-loop monic polynomial is
\begin{equation}
  \st^2+\frac{r c_t+\ell g_k}{\ell c_t}\,\st
        +\frac{k_n^2+r g_k}{\ell c_t}.
  \label{eq:control-poles-monic}
\end{equation}
Matching coefficient by coefficient with \cref{eq:control-target-poly} yields the
system
\begin{align}
  \frac{r c_t+\ell g_k}{\ell c_t}&=2\alpha_n,
  \label{eq:control-poles-sys1}\\
  \frac{k_n^2+r g_k}{\ell c_t}&=\alpha_n^2+\beta_n^2.
  \label{eq:control-poles-sys2}
\end{align}
From \cref{eq:control-poles-sys1}, $r c_t+\ell g_k=2\alpha_n\ell c_t$, that is
\begin{equation}
  \ell g_k=(2\alpha_n\ell-r)\,c_t,
  \qquad\text{equivalently}\qquad
  g_k=\frac{c_t(2\alpha_n\ell-r)}{\ell},
  \label{eq:control-poles-gk}
\end{equation}
which is the second equality of \cref{eq:control-poles-gains}. Substituting
$r g_k=r(2\alpha_n\ell-r)c_t/\ell$ into \cref{eq:control-poles-sys2} and
multiplying by $\ell c_t$,
\begin{equation}
  k_n^2+\frac{r(2\alpha_n\ell-r)}{\ell}\,c_t
  =(\alpha_n^2+\beta_n^2)\,\ell c_t.
\end{equation}
Grouping the terms in $c_t$,
\begin{equation}
  k_n^2=\left[(\alpha_n^2+\beta_n^2)\ell
        -\frac{r(2\alpha_n\ell-r)}{\ell}\right]c_t
       =\frac{(\alpha_n^2+\beta_n^2)\ell^2-2\alpha_n r\ell+r^2}{\ell}\,c_t,
\end{equation}
and under hypothesis \cref{eq:control-poles-denom} we solve for
$c_t=k_n^2\ell/[(\alpha_n^2+\beta_n^2)\ell^2-2\alpha_n r\ell+r^2]$, the first
equality of \cref{eq:control-poles-gains}. Finally $k_{tn}=c_t-c$ and
$k_{0n}=g_k-g$ by the definitions of $c_t$ and $g_k$. The physical realizability
condition $c_t>0,\ g_k>0$ follows from requiring that the effective closed-loop
capacitance and conductance be positive.
\end{proof}

\begin{observacion}
Formula \cref{eq:control-poles-gains} is remarkably compact thanks to the
geometric structure: the desired pole $-\alpha_n\pm\Bt\beta_n$ is interpreted
directly as a root in the temporal plane generated by $\Bt$, and the conjugate
pair appears automatically because $\Bt^2=-1$. Note further that the denominator
\cref{eq:control-poles-denom} is exactly
$|(r+\ell\st{}_{,n}^{+})|^2_{\Bt}=(r-\alpha_n\ell)^2+(\beta_n\ell)^2$ evaluated at
the target pole, which connects the non-degeneracy condition with $\Zimp(\st)$
not vanishing at $\st{}_{,n}^{\pm}$, consistently with the invertibility required
in \cref{prop:control-design}.
\end{observacion}

\begin{observacion}[Relation to the implementation]
\label{obs:control-code}
The results of this section correspond directly to the
\texttt{geolaplace.control} module of the reproducible implementation. The
virtual conductance of \cref{prop:control-virtual} is
\path{virtual_conductance}; the closed-loop determinant
\cref{eq:control-delta-closed} is \path{closed_delta}; the design by target
denominator \cref{eq:control-design-K} is \path{design_K_from_delta}; the
closed-loop modal denominator \cref{eq:control-modal-Dc-expanded} is
\path{modal_closed_denominator}; and the pole assignment
\cref{eq:control-poles-gains} is \path{modal_gain_for_target_poles}. The exact
discrete droop of \cref{prop:control-droop-discreto} is covered by the module
\texttt{geolaplace.converters}: \path{modal_coupling_matrix} builds the
coupling matrix \cref{eq:acoplo-alias} and \path{discrete_droop_poles} computes
the closed-loop poles of the coupled modal system. This
one-to-one correspondence makes it possible to validate each formula numerically
(see \cref{sec:resultados}).
\end{observacion}

\section{Localized faults: detection, localization, and classification}
\label{sec:fallos}

This section addresses the problem of detecting, localizing, and classifying
faults that appear at specific positions along the distributed line of
converters. The guiding idea is that a localized fault manifests itself as a
\emph{space-time singularity} of the model residual, and that its Geometric
Laplace Transform possesses a distinctive algebraic structure: the dependence on
the fault position is encoded in a spatial rotor. The approach falls within
model-based fault diagnosis and residual
generation~\cite{ding2008model,isermann2006fault}.

Two different objects must be distinguished from the outset:
\begin{itemize}
  \item the \emph{balance residual} $\Res(t,x)$, which represents the equivalent
        anomalous source introduced by the fault;
  \item the \emph{observed response}, for example $v_{\mathrm{res}}(t,x)$,
        which is the propagated effect of that residual through the line.
\end{itemize}
This distinction matters in fault detection: the geometric signature
$F_f(\st)e^{-\sx x_f}$ belongs to the balance residual. If only voltages are
measured, the line transfer function must first be removed (a \emph{prewhitening}
or modal deconvolution operation) in order to recover the residual signature. If
this separation is not carried out, a sensor anomaly can be mistaken for a
localized physical fault.

\subsection{Distributed residual and the signature of a localized fault}
\label{subsec:residuo}

We start from the nominal closed-loop distributed model. The plant satisfies
\begin{equation}
  \px v + \ell\,\pt i + r\,i = 0,
  \qquad
  \px i + c\,\pt v + g\,v = u,
\end{equation}
and the nominal controller applies
\begin{equation}
  u = r_u - K[v].
\end{equation}
In healthy operation, therefore,
\begin{equation}
  \px i + c\,\pt v + g\,v + K[v]-r_u = 0.
\end{equation}

\begin{definicion}[Distributed balance residual]
\label{def:residuo}
Let $K[\,\cdot\,]$ be the linear distributed control operator and $r_u$ the
known excitation. We define the balance residual as
\begin{equation}
\label{eq:residuo}
\Res(t,x) = \px i + c\,\pt v + g\,v + K[v] - r_u .
\end{equation}
Under ideal healthy operation, $\Res(t,x)=0$. In the presence of noise or model
errors, a statistical threshold on a norm of $\Res$ is used to decide on
detection.
\end{definicion}

If the physical fault introduces an anomalous injection $u_f(t,x)$, the actual
equation becomes
\begin{equation}
  \px i + c\,\pt v + g\,v = r_u-K[v]+u_f(t,x),
\end{equation}
and, by the previous definition,
\begin{equation}
  \Res(t,x)=u_f(t,x).
\end{equation}
Thus, the balance residual is the equivalent source of the fault. It is this
identity that justifies localizing faults from $\Res$.

\begin{definicion}[Additive localized fault]
\label{def:fallo}
A fault localized at $x_f\in(0,L)$ with temporal signature $f_f(t)$ is modeled as
\begin{equation}
\label{eq:fallo}
\Res(t,x)=u_f(t,x)=f_f(t)\,\delta(x-x_f),
\end{equation}
where $\delta(\cdot)$ is the spatial Dirac delta. The function $f_f(t)$ has units
of incremental current and is called the \emph{temporal signature} of the fault.
\end{definicion}

\begin{observacion}[Additive versus multiplicative faults]
\label{obs:aditivo-multiplicativo}
Form \eqref{eq:fallo} is exact for faults that behave as an anomalous current
injection. A local short circuit or shunt, in contrast, is not an independent
source: it introduces a relation of the type $f_f(t)=-Y_f[v(t,x_f)]$. It is
therefore a multiplicative, or local-admittance, disturbance. It can be handled
with the same notation if $f_f(t)$ is interpreted as the resulting fault current
and if, after localizing $x_f$, one estimates $Y_f$ by dividing by the local
voltage. For strong shunt faults, the local voltage $v(t,x_f)$ is itself modified
by the fault; in that case the classification must be understood as an estimation
within the already-faulted system, not as an independent external source.
\end{observacion}

\begin{proposicion}[Transformed signature of a localized fault]
\label{prop:firma}
Let $F_f(\st)=\LG\{f_f\}(\st)$ be the temporal transform of the signature. If the
line is treated as a spatial half-axis, or if the residual is extended by zero
outside $[0,L]$, the Geometric Laplace Transform of residual \eqref{eq:fallo} is
\begin{equation}
\label{eq:firma}
R(\st,\sx)=F_f(\st)\,e^{-\sx x_f}.
\end{equation}
In particular, evaluating on the purely oscillatory spatial axis $\sx=\Bx k$,
\begin{equation}
\label{eq:firma-rotor}
R(\st,\Bx k)=F_f(\st)e^{-\Bx kx_f}
=F_f(\st)\bigl[\cos(kx_f)-\Bx\sin(kx_f)\bigr].
\end{equation}
\end{proposicion}

\begin{proof}
Applying the definition of the transform to \eqref{eq:fallo},
\begin{align}
R(\st,\sx)
  &= \int_0^\infty\!\int_0^\infty
     f_f(t)\,\delta(x-x_f)e^{-\st t}e^{-\sx x}\,\dd x\dd t  \\
  &= \left(\int_0^\infty f_f(t)e^{-\st t}\,\dd t\right)
     \left(\int_0^\infty \delta(x-x_f)e^{-\sx x}\,\dd x\right) .
\end{align}
The first factor is $F_f(\st)$ and the second is $e^{-\sx x_f}$ for $x_f>0$. If
$\sx=\Bx k$, then $\Bx^2=-1$ and one obtains the geometric Euler identity
\begin{equation}
  e^{-\Bx kx_f}=\cos(kx_f)-\Bx\sin(kx_f).
\end{equation}
\end{proof}

\cref{prop:firma} is the basic localization result: the temporal part of the
fault appears in $F_f(\st)$, whereas the position appears in the spatial rotor
$e^{-\Bx kx_f}$. This separation should not be confused with the measured voltage
response, which includes line propagation.

\subsection{From the residual to the measured signals}
\label{subsec:residuo-vs-medida}

In many systems, $\Res(t,x)$ is not measured directly over the entire continuum.
The usual case is to measure voltage, and perhaps current, at a finite set of
converters/sensors. It is therefore convenient to write out explicitly how the
physical residual propagates to the voltage.

In the transformed domain, the closed-loop model with effective admittance
$\Yadm_c(\st,\sx)=\Yadm(\st)+K(\st,\sx)$ leads to
\begin{equation}
  \Deterc(\st,\sx)=\sx^2-\Zimp(\st)\Yadm_c(\st,\sx).
\end{equation}
If the balance residual is $R(\st,\sx)$, the residual voltage satisfies
\begin{equation}
\label{eq:Vres-from-R}
  V_{\mathrm{res}}(\st,\sx)
  = -\frac{\Zimp(\st)}{\Deterc(\st,\sx)}\,R(\st,\sx),
\end{equation}
provided the factors are invertible in the region of interest. For a point fault,
\begin{equation}
\label{eq:Vres-fallo}
  V_{\mathrm{res}}(\st,\sx)
  = -\frac{\Zimp(\st)}{\Deterc(\st,\sx)}
    F_f(\st)e^{-\sx x_f}.
\end{equation}
Hence, if one intends to localize using voltage data, the prewhitened signature
must be used,
\begin{equation}
\label{eq:prewhitening}
  \widetilde R(\st,\sx)
  := -\frac{\Deterc(\st,\sx)}{\Zimp(\st)} V_{\mathrm{res}}(\st,\sx),
\end{equation}
which, ideally, recovers $F_f(\st)e^{-\sx x_f}$. This operation is not harmless:
it is sensitive to zeros of $\Zimp$, to zeros or near-zeros of $\Deterc$, to
noise, and to parametric uncertainty.

\begin{observacion}[Finite line]
\label{obs:linea-finita}
The continuous formula \eqref{eq:Vres-fallo} corresponds to the $(\st,\sx)$-domain
treatment of a homogeneous line. In a finite line with boundary conditions, the
response must be expressed through spatial modes or through the Green's function
that satisfies those boundary conditions. In that case, free-phase localization in
$k$ must be replaced by the modal localization of \cref{subsec:modal}, unless one
works with a long-line approximation or with local windows where reflections are
negligible.
\end{observacion}

\subsubsection*{Regularized prewhitening}

The ideal inversion \eqref{eq:prewhitening} is an exact deconvolution and, for
that reason, sensitive to noise when the propagation operator is small or
ill-conditioned. Writing $V_{\mathrm{res}}=H\,R$ with
\begin{equation}
\label{eq:H-propagacion}
  H(\st,\sx)=-\frac{\Zimp(\st)}{\Deterc(\st,\sx)},
\end{equation}
the direct inversion $R=H^{-1}V_{\mathrm{res}}$ amplifies noise near the zeros of
$H$ (or, equivalently, near the zero divisors of the algebra). We propose instead
the \emph{regularized} (Tikhonov-type~\cite{tikhonov1963regularization})
deconvolution
\begin{equation}
\label{eq:prewhit-reg}
  \widehat R_\lambda
  = \frac{\overline{H}}{\lvert H\rvert^{2}+\lambda}\,V_{\mathrm{res}},
  \qquad \lambda>0,
\end{equation}
where the conjugation $\overline{H}$ and the modulus $\lvert H\rvert^{2}$ must be
understood in a representation that is \emph{compatible} with the algebra
structure. To avoid a problematic notion of modulus in the presence of zero
divisors, the real matrix representation is used: if $M=M(H)$ is the $4\times4$
matrix of multiplication by $H$ (\cref{prop:matriz}) and $v$ is the coordinate
vector of $V_{\mathrm{res}}$, then
\begin{equation}
\label{eq:prewhit-reg-mat}
  \widehat R_\lambda \;\leftrightarrow\;
  \bigl(M^{\!\top}M+\lambda I\bigr)^{-1}M^{\!\top}v .
\end{equation}
With $\lambda=0$ and $H$ invertible, the exact inversion is recovered; with
$\lambda>0$ the solution remains bounded even if $H$ is ill-conditioned. In a
finite line, the modal version is cleaner: with $G_n(\st)=\Zimp(\st)/D_n(\st)$
(\cref{eq:Gn}) and $V_n=G_n A_n$,
\begin{equation}
\label{eq:modal-reg}
  \widehat A_n
  = \frac{\overline{G_n}}{\lvert G_n\rvert^{2}+\lambda}\,V_n
  \;\leftrightarrow\;
  \bigl(M(G_n)^{\!\top}M(G_n)+\lambda I\bigr)^{-1}M(G_n)^{\!\top}v_n .
\end{equation}
The routines \texttt{regularized\_prewhiten\_voltage} and
\texttt{regularized\_modal\_amplitude} implement \eqref{eq:prewhit-reg-mat} and
\eqref{eq:modal-reg}.

\subsection{Localization by geometric spatial phase}
\label{subsec:fase}

The factorization \eqref{eq:firma-rotor} suggests recovering $x_f$ by reading the
geometric spatial phase. This strategy must be applied to the residual $R$ or to
the prewhitened signal $\widetilde R$ of \eqref{eq:prewhitening}, not directly to
an uncompensated voltage response.

\begin{proposicion}[Phase localization with two wavenumbers]
\label{prop:fase}
Let $k_1\neq k_2$ and
\begin{equation}
  R_m:=R(\st,\Bx k_m)=F_f(\st)e^{-\Bx k_m x_f},\qquad m=1,2.
\end{equation}
Assume that $F_f(\st)$ is invertible in the bicomplex subalgebra and that $R_1$ is
not close to a zero divisor. Then
\begin{equation}
\label{eq:fase-cociente}
  R_2R_1^{-1}=e^{-\Bx(k_2-k_1)x_f},
\end{equation}
and
\begin{equation}
\label{eq:fase-xf}
  x_f=-\frac{\ArgBx(R_2R_1^{-1})}{k_2-k_1}
  \pmod{\frac{2\pi}{k_2-k_1}} .
\end{equation}
\end{proposicion}

\begin{proof}
By commutativity,
\begin{equation}
R_2R_1^{-1}=F_f e^{-\Bx k_2x_f}F_f^{-1}e^{\Bx k_1x_f}
=e^{-\Bx(k_2-k_1)x_f}.
\end{equation}
The right-hand side is a pure spatial rotor. If $\varphi=(k_2-k_1)x_f$, then
$e^{-\Bx\varphi}=\cos\varphi-\Bx\sin\varphi$ and
$\ArgBx(e^{-\Bx\varphi})=-\varphi$ modulo $2\pi$.
\end{proof}

\begin{observacion}[Domain of validity of phase localization]
\label{obs:fase-validez}
\cref{prop:fase} and the robust estimator \eqref{eq:fase-ls} presuppose that the
observed signature is a \emph{pure} spatial rotor $F_f(\st)e^{-\Bx k x_f}$. This
is appropriate for: an infinite line or half-axis; the local long-line
approximation; space-time windows where reflections are negligible; and the
prewhitened balance residual \eqref{eq:prewhitening} or one measured directly. It
should \emph{not} be applied directly to data from a finite line dominated by
reflections and global boundary conditions: in that regime, the response
superposes standing modes $\sin(k_nx)$ and the free spatial phase in $k$ ceases
to be a pure rotor. For the global problem on $0\le x\le L$, the estimator
consistent with the boundary conditions is the \emph{modal} one of
\cref{subsec:modal}; the phase estimator is valid there only as a local or
long-line approximation.
\end{observacion}

\begin{observacion}[Use of $\ArgBx$ with noise]
\label{obs:arg-ruido}
In the exact case, $R_2R_1^{-1}$ lies in the plane $\{1,\Bx\}$. With noise,
imperfect prewhitening, or model errors, components along $\Bt$ and $\Btx$ may
appear. In that case one should not blindly apply
$\ArgBx(z)=\operatorname{atan2}(\langle z\rangle_{\Bx},\langle z\rangle_0)$ to the
full element. A robust option is to first project onto the spatial plane,
\begin{equation}
  \Pi_x z := \langle z\rangle_0 + \langle z\rangle_{\Bx}\Bx,
\end{equation}
and to use $\ArgBx(\Pi_x z)$ only if $\|\Pi_x z\|$ exceeds a threshold. The
preferable alternative, when several wavenumbers are available, is the
least-squares fit of \cref{prop:fase-robusto}.
\end{observacion}

\begin{proposicion}[Robust least-squares estimator]
\label{prop:fase-robusto}
Given, for $m=1,\dots,M$, the prewhitened signatures
$R_m\approx F e^{-\Bx k_mx_f}$, consider
\begin{equation}
\label{eq:fase-ls}
  \min_{x\in\mathcal I,\;F\in\mathcal A}
  \sum_{m=1}^M w_m\,\|R_m-F e^{-\Bx k_mx}\|^2,
\end{equation}
where $\mathcal A=\operatorname{span}\{1,\Bt,\Bx,\Btx\}$, $w_m\ge 0$ are weights,
and $\mathcal I$ is the physically admissible interval, for example $[0,L]$. For
each fixed $x$, the problem is linear in the four real components of $F$.
Substituting the minimizer $\hat F(x)$ yields a one-dimensional search in $x$; the
minimizer $\hat x$ and $\hat F=\hat F(\hat x)$ give the estimated position and
temporal signature.
\end{proposicion}

\begin{proof}
For fixed $x$, the factor $E_m(x):=e^{-\Bx k_mx}$ is known. Since the algebra is
commutative, the map $F\mapsto F E_m(x)$ is linear in the coordinates of $F$.
Denote by $M_m(x)\in\RR^{4\times4}$ the real matrix of multiplication by $E_m(x)$
and by $r_m$ the vector of components of $R_m$. Then
\begin{equation}
  \|R_m-FE_m(x)\|^2=\|r_m-M_m(x)\operatorname{vec}(F)\|^2.
\end{equation}
Stacking the weighted blocks yields
\begin{equation}
  \min_{\operatorname{vec}(F)}\|A(x)\operatorname{vec}(F)-b\|^2,
\end{equation}
whose least-squares solution is $\operatorname{vec}(\hat F(x))=A(x)^+b$. The
reduced cost is minimized over the interval $\mathcal I$ by sweeping, refinement,
or a bounded one-dimensional method.
\end{proof}

\begin{observacion}[Phase ambiguity]
\label{obs:ambiguedad-fase}
Phase estimation carries an ambiguity modulo $2\pi/(k_2-k_1)$. Therefore, for a
line of length $L$ it is advisable to choose the wavenumbers so that the ambiguity
period covers the physical interval, or, alternatively, to solve the global
problem \eqref{eq:fase-ls} over $[0,L]$ with several $k_m$.
\end{observacion}

\subsection{Modal localization in a finite line}
\label{subsec:modal}

When the line has finite length $0\le x\le L$, the boundary conditions select
spatial modes. To avoid ambiguous normalization factors, we introduce an
orthonormal basis
\begin{equation}
\label{eq:psi-n}
  \psi_n(x)=\sqrt{\frac{2}{L}}\sin(k_nx),\qquad k_n=\frac{n\pi}{L}.
\end{equation}
If the unnormalized basis $\phi_n(x)=\sin(k_nx)$ is used, all the following results
remain valid provided the factor $2/L$ is introduced in the modal projection; that
factor cancels in the localization ratios, but not in the severity estimation.

A point residual has modal projection
\begin{equation}
  R_n(\st)=F_f(\st)\psi_n(x_f).
\end{equation}
If $G_n(\st)$ is the modal residual-to-voltage transfer, then
\begin{equation}
\label{eq:modo-Vn}
  V_n(\st)=G_n(\st)F_f(\st)\psi_n(x_f).
\end{equation}
We define the normalized modal amplitude
\begin{equation}
\label{eq:An}
  A_n(\st):=\frac{V_n(\st)}{G_n(\st)}=F_f(\st)\psi_n(x_f),
\end{equation}
provided $G_n(\st)$ is neither zero nor ill-conditioned. The amplitudes $A_n$ share
the temporal signature and differ only in the spatial weight.

The general modal estimator is
\begin{equation}
\label{eq:modal-ls}
  \min_{x\in[0,L],\;F\in\mathcal A}
  \sum_{n\in\mathcal N} w_n\,\|A_n-F\psi_n(x)\|^2 .
\end{equation}

\begin{proposicion}[Modal localization with two modes]
\label{prop:modal}
Let $A_1=F_f\psi_1(x_f)$ and $A_2=F_f\psi_2(x_f)$, with $k_1=\pi/L$ and
$k_2=2\pi/L$. If $A_1\neq0$ and the common signature can be canceled, then
\begin{equation}
\label{eq:modal-cociente}
  \frac{A_2}{A_1}=2\cos\!\left(\frac{\pi x_f}{L}\right),
  \qquad
  x_f=\frac{L}{\pi}\arccos\!\left(\frac{A_2}{2A_1}\right),
\end{equation}
for $x_f\in(0,L)$.
\end{proposicion}

\begin{proof}
The normalization factor $\sqrt{2/L}$ is common to $\psi_1$ and $\psi_2$ and
cancels. Since
\begin{equation}
  \sin(2\pi x_f/L)=2\sin(\pi x_f/L)\cos(\pi x_f/L),
\end{equation}
one obtains
\begin{equation}
\frac{A_2}{A_1}
=\frac{F_f\sin(2\pi x_f/L)}{F_f\sin(\pi x_f/L)}
=2\cos(\pi x_f/L).
\end{equation}
The arccosine is bijective on $\pi x_f/L\in(0,\pi)$.
\end{proof}

\begin{observacion}[Limitations of the modal method]
\label{obs:modal}
The closed-form formula \eqref{eq:modal-cociente} is useful as a quick diagnostic,
but it should not be the main method in the presence of noise. It fails near modal
nodes, requires $A_1\neq0$, and can suffer from ambiguity if few modes are used.
The fit \eqref{eq:modal-ls} allows several modes to be used and a confidence metric
to be produced, for example
\begin{equation}
  \Gamma=1-\frac{\sum_n\|A_n-\hat F\psi_n(\hat x)\|^2}{\sum_n\|A_n\|^2+\epsilon}.
\end{equation}
If the signal energy $\sum_n\|A_n\|^2$ is negligible (a fault at a modal node or at
a line end), $\Gamma$ is not informative and should be degraded to $0$ rather than
reporting a spurious maximal confidence.
\end{observacion}

\begin{observacion}[Modal observability of a localized fault]
\label{obs:observabilidad}
Modal localization starts from $A_n=F_f\,\psi_n(x_f)$ for $n$ in an observed set
$\mathcal N$. If $\psi_n(x_f)=0$ for all $n\in\mathcal N$, the fault is
\emph{invisible} in those modes. More generally, if two positions $x_a$ and $x_b$
produce proportional modal vectors,
\begin{equation}
\label{eq:obs-proporcional}
  \bigl(\psi_n(x_a)\bigr)_{n\in\mathcal N}
  =\mu\,\bigl(\psi_n(x_b)\bigr)_{n\in\mathcal N},
  \qquad \mu\in\RR,
\end{equation}
the position is not identifiable without additional information, because the
amplitude $F_f$ can absorb the factor $\mu$. Localization therefore requires that
the map
\begin{equation}
\label{eq:obs-inyectiva}
  x\;\longmapsto\;\bigl[\psi_n(x)\bigr]_{n\in\mathcal N}
\end{equation}
be \emph{injective up to scale} on the admissible interval. A useful diagnostic
metric is the energy of the modal pattern,
$\mathrm{obs}(x)=\bigl(\sum_{n\in\mathcal N}\sin^2(k_n x)\bigr)^{1/2}$, which
vanishes at the ends $x=0,L$ and at nodes common to all observed modes (routine
\texttt{modal\_observability\_score}); the normalized pattern $\hat\psi(x)$
(\texttt{modal\_pattern}) makes it possible to detect the degeneracies
\eqref{eq:obs-proporcional}.
\end{observacion}

\subsection{Extension: multiple faults}
\label{subsec:fallos-multiples}

The modeling of a single point fault extends naturally to $J$ simultaneous faults
by linearity of the residual. In the continuum,
\begin{equation}
\label{eq:multi-continuo}
  R(\st,\sx)=\sum_{j=1}^{J}F_j(\st)\,e^{-\sx x_j},
\end{equation}
and in a finite line, projecting onto the modes,
\begin{equation}
\label{eq:multi-modal}
  A_n(\st)=\sum_{j=1}^{J}F_j(\st)\,\psi_n(x_j).
\end{equation}
The joint identification of $\{F_j,x_j\}$ is then a \emph{sparse recovery} problem:
\begin{equation}
\label{eq:multi-sparse}
  \min_{\{F_j,x_j\}}\;
  \Bigl\|A-\sum_{j}F_j\,\psi(x_j)\Bigr\|^2
  +\lambda\sum_j\|F_j\|,
\end{equation}
or, over a fixed spatial grid $\{x_g\}$ with modal dictionary
$\Phi=[\psi_n(x_g)]_{n,g}$ (routine \path{modal_dictionary}),
\begin{equation}
\label{eq:multi-grid}
  \min_{f}\;\|A-\Phi f\|^2+\lambda\,\Omega(f),
\end{equation}
where $f_g$ collects the amplitude (the four algebra components) of the candidate
fault at $x_g$ and $\Omega$ promotes sparsity.

\subsubsection*{Implemented recovery}

This inversion is implemented (not merely formulated) through two complementary
routes, both using \texttt{numpy}/\texttt{scipy}:

\begin{itemize}
  \item \textbf{Simultaneous OMP}~\cite{pati1993omp,tropp2007omp}
        (\texttt{recover\_faults\_omp}): a greedy
        algorithm that, at each iteration, selects the column of $\Phi$ most
        correlated with the residual \emph{jointly} across the four algebra
        components (common support), refits the block of amplitudes by least
        squares, and updates the residual, until $J$ faults or a threshold
        relative residual is reached. The positions are refined off-grid by
        separable least squares (VARPRO~\cite{golub1973varpro}). It requires no
        tuning of $\lambda$ and is accurate for well-separated faults.
  \item \textbf{Group-LASSO}~\cite{yuan2006grouplasso}
        (\texttt{recover\_faults\_grid\_lasso}): a convex
        relaxation of \eqref{eq:multi-grid} with group norm
        $\Omega(f)=\sum_g\|f_g\|_2$ ($\ell_{2,1}$, which imposes common support on
        the four components), solved by FISTA~\cite{beck2009fista} with group
        soft-thresholding; the
        peaks are extracted with \texttt{extract\_fault\_peaks}.
\end{itemize}

Representing each amplitude by its four real components, with a shared \emph{real}
dictionary $\Phi$, is consistent with the matrix treatment used in regularized
prewhitening and avoids a problematic notion of modulus in an algebra with zero
divisors. Identifiability requires sufficient modal observability
(\cref{obs:observabilidad}): faults that are too close relative to the spatial
resolution of the observed modes produce nearly collinear columns of $\Phi$ and
become difficult to separate. A numerical example is presented in
\cref{subsec:multifault-result}.

\subsubsection*{Close faults: super-resolution}

When two faults are closer than the spatial resolution of the grid, the previous
methods lose accuracy. Two mechanisms address this regime.

\paragraph{Off-grid refinement (VARPRO).} In \texttt{recover\_faults\_omp}, after
selection over the grid, the positions are treated as \emph{continuous} parameters
and one minimizes
\begin{equation}
\label{eq:varpro}
  \min_{x_1,\dots,x_J}\ \Bigl\|A-\Phi(x)\,F(x)\Bigr\|^2,
  \qquad F(x)=\Phi(x)^{+}A,
\end{equation}
that is, a \emph{separable} least-squares problem (VARPRO) in which the amplitudes
$F(x)$ are eliminated analytically. The fine grid only provides the
initialization; the solution is not limited by its step. This refinement is the
scheme's practical super-resolution mechanism.

\paragraph{Grid-free spectral estimation (matrix pencil / ESPRIT).} The modal
measurement also admits a \emph{grid-free} reading. Since $k_n=n\pi/L$, the
sequence $A_n$ in $n$ is a \emph{sum of sinusoids} of frequency $\theta_j=\pi
x_j/L$:
\begin{equation}
\label{eq:harmonic}
  A_n=\sum_{j=1}^{J}F_j\,\sin(n\,\theta_j)
     =\sum_{k}c_k\,z_k^{\,n},
  \qquad z_k=e^{\pm i\theta_j}.
\end{equation}
Localization is then a \emph{harmonic retrieval} problem: the poles $z_k$ are
estimated by the matrix-pencil method (matrix pencil~\cite{hua1990pencil} /
ESPRIT~\cite{roy1989esprit}) from the signal
subspace of a Hankel matrix formed with $A_n$ (stacking the four algebra
components as \emph{snapshots} of common poles, with forward-backward averaging
that exploits the conjugate symmetry of real sinusoids). From $\theta_j=\arg z_j$
one obtains $x_j=L\,\theta_j/\pi$ without any grid
(\texttt{recover\_faults\_esprit}). This approach can resolve faults closer than
the classical Rayleigh limit $\sim L/N$.

\paragraph{Estimating the number of faults (model order).} The previous methods
require the number of faults $J$. This can be estimated automatically from the same
harmonic structure: in \eqref{eq:harmonic}, $A_n$ is a sum of $J$ real sinusoids,
so the Hankel matrix has \emph{signal rank} $2J$ (one conjugate pair per sinusoid).
Its singular values separate into $2J$ ``signal'' ones (large) and the rest
``noise'' ones (small), and the number of components is selected by an information
criterion ---Wax--Kailath MDL~\cite{wax1985mdl} by default, or AIC---:
\begin{equation}
\label{eq:mdl}
  \widehat M=\arg\min_{k}\;
    (p-k)\,N_s\,\log\frac{a_k}{g_k}
    +\tfrac12 k(2p-k)\log N_s,
  \qquad \widehat J=\Bigl[\tfrac{\widehat M}{2}\Bigr],
\end{equation}
where $a_k$ and $g_k$ are the arithmetic and geometric means of the $p-k$ smallest
eigenvalues and $N_s$ is the number of \emph{snapshots}. MDL avoids setting a
threshold on the singular values manually, but it is \emph{not} an
assumption-free decision: it requires choosing the pencil size, the effective
number of \emph{snapshots}, and a reasonable upper bound on the order
($\texttt{max\_faults}$). Moreover, on nearly noise-free data or data confined to a
single algebra component, it is advisable to apply a \emph{numerical rank guard}
before evaluating MDL, so as to avoid overestimations due to round-off singular
values. The routine \texttt{estimate\_num\_faults} implements \eqref{eq:mdl} with
that guard; passing \texttt{n\_faults=None} to the recoverers activates the
estimation. Its accuracy is quantified in \cref{subsec:multifault-result}.

Concretely, before evaluating MDL the number of singular values above the
threshold $\tau_{\mathrm{rg}}=\max(\tau_{\mathrm{abs}},\,\tau_{\mathrm{rel}}\,s_{\max})$
---parameters \path{abs_rank_tol} and \path{rel_rank_tol} of the routine--- is
counted.
If the remaining tail is at the round-off level (strict rank deficiency), that count
is exactly $2J$ and $\widehat J$ is returned without invoking MDL; this prevents the
criterion from trying to ``explain'' numerically null singular values and
overestimating the order in the exact or nearly exact limit. This guard does
\emph{not} replace MDL in the noisy regime ---where all singular values exceed the
threshold and the count coincides with the subspace dimension---; it merely avoids
a numerical pathology in the limiting case. In addition, the effective number of
\emph{snapshots} used in \eqref{eq:mdl} is that of forward Hankel rows per
component, $N_s=N-L_p$, and not the inflated count that would result from treating
as independent the rows from forward-backward averaging and from stacking the four
components (which would overweight the likelihood relative to the penalty and
overestimate the order).

\begin{observacion}[Scope of automatic order selection]
\label{obs:orden-alcance}
Automatic order estimation does \emph{not} constitute a universal guarantee. Its
reliability depends on the number of modes $N$, the SNR, the spatial separation of
the faults, the noise structure, the pencil size, and the bound
$\texttt{max\_faults}$. In the deeply sub-Rayleigh regime, or with orders close to
the limit resolvable with $N$ modes, order selection can become ambiguous and must
be complemented with prior knowledge of the problem.
\end{observacion}

\paragraph{Resolution limit.} Both methods have a limit governed by the number of
observed modes $N$ and the signal-to-noise ratio. \cref{subsec:multifault-result}
characterizes this limit empirically: at a fixed noise level, recovery is reliable
down to separations well below $L/N$, and it degrades in a controlled manner when
the faults come even closer. VARPRO refinement proves somewhat more robust in the
most demanding regime, whereas ESPRIT offers a grid-free alternative with lower cost
and memory. The \emph{deeply} sub-Rayleigh regime, where even order selection
becomes ambiguous, remains an open question (\cref{subsec:trabajo-futuro}).

\subsection{Fault classification}
\label{subsec:clasificacion}

Once the fault has been localized, its nature is inferred from the estimated
temporal signature $F_f(\st)$ and, where appropriate, from the equivalent local
admittance $Y_f(\st)$. Classification should be carried out after checking the
physical consistency of the spatial fit; otherwise, a sensor anomaly may look like
a localized source.

\paragraph{(a) Open circuit or injection loss.}
If a converter stops injecting a constant current $I_0$ from an instant $t_f$
onward, then
\begin{equation}
\label{eq:fallo-apertura-tiempo}
  f_f(t)=-I_0H(t-t_f),
\end{equation}
and
\begin{equation}
\label{eq:fallo-apertura}
  F_f(\st)=-\frac{I_0}{\st}e^{-\st t_f}.
\end{equation}
If the time origin is relocated to the instant of detection, $t_f=0$ and one is
left with $F_f(\st)=-I_0/\st$. A direct test is to check that
$\st F_f(\st)e^{\st t_f}$ is approximately constant.

\paragraph{(b) Local shunt or short circuit.}
A local leakage path introduces a current proportional to the local voltage:
\begin{equation}
\label{eq:fallo-shunt-tiempo}
  f_f(t)=-Y_f[\,v(t,x_f)\,].
\end{equation}
In the transformed domain,
\begin{equation}
\label{eq:fallo-shunt}
  F_f(\st)=-Y_f(\st)V_f(\st),
\end{equation}
where $V_f(\st)$ is the transformed local voltage at the fault point. Therefore,
\begin{equation}
\label{eq:Yf-estimada}
  \widehat Y_f(\st)=-\frac{\widehat F_f(\st)}{\widehat V_f(\st)}.
\end{equation}
This estimation is reliable only if $\widehat V_f$ is not small and if the spatial
localization is consistent. Subclassification can be carried out according to
\begin{align}
  Y_f(\st) &\approx G_f>0 && \text{resistive},\label{eq:shunt-R}\\
  Y_f(\st) &\approx C_f\st && \text{capacitive},\label{eq:shunt-C}\\
  Y_f(\st) &\approx \frac{1}{L_f\st} && \text{inductive}.\label{eq:shunt-L}
\end{align}
Operationally, one evaluates which of the quantities $Y_f$, $Y_f/\st$, or
$Y_f\st$ exhibits the smallest relative dispersion over the set of $\st$ values
used for the diagnosis.

\paragraph{(c) Sensor fault.}
A sensor fault does not introduce a physical source into the line. Instead, it
alters the measured signals. If $v_m$ is the voltage measurement, a simple model
is
\begin{equation}
  v_m(t,x)=v(t,x)+e_s(t)\delta_s(x-x_s),
\end{equation}
where $\delta_s$ represents the location of the sensor or measurement node. A
genuine physical fault must satisfy simultaneously:
\begin{equation}
\label{eq:prop-fisica-residuo}
  R(\st,\sx)\approx F_f(\st)e^{-\sx x_f},
\end{equation}
and
\begin{equation}
\label{eq:prop-fisica-tension}
  V_{\mathrm{res}}(\st,\sx)
  \approx -\frac{\Zimp(\st)}{\Deterc(\st,\sx)}F_f(\st)e^{-\sx x_f}.
\end{equation}
A sensor fault usually breaks this compatibility: it may produce an apparent
residual, but not a response coherent with the line propagation operator.

We define a physical consistency test through the robust relative error
\begin{equation}
\label{eq:consistencia}
  \eta=\frac{\|V_{\mathrm{med}}-V_{\mathrm{pred}}\|}
  {\max\{\|V_{\mathrm{med}}\|,\|V_{\mathrm{pred}}\|,\epsilon\}},
\end{equation}
where $V_{\mathrm{pred}}$ is computed with \eqref{eq:prop-fisica-tension} using the
estimated position and signature. If $\eta$ exceeds a threshold, or if the modal
confidence is low, the diagnosis must reject the localized physical fault
hypothesis and label the anomaly as a sensor fault or model inconsistency.

\paragraph{(d) Local controller fault.}
If the nominal controller applies $u=r_u-K[v]$ and at node $x_f$ the gain changes
$K\mapsto K+\Delta K_f$, the actual injection contains the additional term
\begin{equation}
  u_f(t,x)=-\Delta K_f[\,v(t,x_f)\,]\delta(x-x_f).
\end{equation}
Therefore, in the shunt-fault notation,
\begin{equation}
\label{eq:controlador-Yf}
  F_f(\st)=-\Delta K_f(\st)V_f(\st),
  \qquad
  Y_f(\st)=\Delta K_f(\st).
\end{equation}
Formally, it is indistinguishable from an anomalous local admittance; telling a
``controller fault'' apart from a ``physical shunt'' requires additional
information about the converter's internal state, setpoint logs, switch currents,
or local telemetry.

\begin{observacion}[Discrete converters: residual and site identification]
\label{obs:fallos-sitios}
When the healthy injection is produced by $P$ discrete converters, as in
\cref{subsec:inyeccion-discreta}, the nominal term subtracted by the residual
\eqref{eq:residuo} is the known rotor sum \eqref{eq:U-discreta}: writing
$U^{\mathrm{tot}}=\sum_p I_p(\st)e^{-\sx x_p}+F_f(\st)e^{-\sx x_f}$ for the
faulty line, the subtraction of the healthy sum leaves
$R(\st,\sx)=F_f(\st)\,e^{-\sx x_f}$, exactly the signature \eqref{eq:firma}.
The entire pipeline of this section ---detection, phase and modal
localization, classification--- therefore carries over unchanged to the
discrete-converter line. The discrete setting adds, moreover, a diagnostic
that the continuous model cannot express: the estimated position
$\widehat x_f$ can be matched against the known sites $\{x_p\}$ (routine
\path{nearest_converter}). An anomaly with the \emph{injection-loss}
signature ---the discrete realization of \eqref{eq:fallo-apertura}, the
loss of converter $p$'s contribution, $F_f=-I_p$--- whose $\widehat x_f$
falls within a tolerance $\delta$ of a site $x_p$ is diagnosed as the
\emph{outage of plant $p$}, whereas a fault located farther than $\delta$
from every site points to the cable itself (e.g.\ a shunt defect). The
decision threshold $\delta$ must be chosen between the localization
uncertainty and the half-spacing $L/(2P)$; by construction the rule is
blind to cable faults within $\delta$ of a site (a fraction $2\delta P/L$
of the line). Three caveats bound the reach of this diagnosis. The
separation is only as sharp as the localization uncertainty of
\cref{subsec:fase,subsec:modal} relative to the inter-site spacing
$L/P$, and degrades for dense converter sets or noisy measurements. The
subtraction presumes the dispatch $I_p(\st)$ and the sites
$x_p$ to be known exactly, and a dispatch error $\Delta I_p$ leaves the
spurious residual $-\sum_p\Delta I_p(\st)\,e^{-\sx x_p}$ competing with
the fault signature; \cref{subsec:discreto-result} quantifies the
sensitivity to this error. And the additive model $F_f=-I_p$
captures the loss of plant $p$'s feed-forward dispatch only: if the plant
also participates in the droop feedback \eqref{eq:droop-discreto}, its
outage additionally removes the local feedback term $-\kappa_p v(t,x_p)$,
a controller-fault perturbation of admittance type (paragraph~(d),
\eqref{eq:controlador-Yf}, with $\Delta K_f=-\kappa_p$) that the additive
subtraction does not capture.
\end{observacion}

\begin{observacion}[Relation to the implementation]
\label{obs:codigo-fallos}
The results of this section are implemented in the module
\texttt{geolaplace.faults}. The routine \path{point_fault_signature}
represents \eqref{eq:firma}; the function \path{prewhiten_voltage} applies the
prewhitening \eqref{eq:prewhitening}
$\widetilde R=-(\Deterc/\Zimp)V_{\mathrm{res}}$, so that phase localization
(\path{phase_two_point_location}, \path{estimate_location_phase}) operates
on balance residuals, whether native or prewhitened; the robust phase estimator
solves \eqref{eq:fase-ls} on $[0,L]$. Modal localization
(\path{estimate_location_modal}, \path{estimate_fault_severity_modal})
uses the \emph{unnormalized} basis $\phi_n=\sin(k_nx)$, i.e.\ the direct
measurement $A_n=F_f\,\phi_n(x_f)$ without the factor $\sqrt{2/L}$ of
\eqref{eq:psi-n}. Since the same convention is used when synthesizing the
measurement (\path{inject_point_fault}) and when estimating the severity, the
normalization factor cancels and the severity is recovered without bias; by
contrast, the physical $L^2$ projection (\path{project_signal_to_modes})
does carry the $2/L$ factor, and to go from its coefficients to the measurement
$A_n$ one must divide by the modal transfer $G_n$, not rescale by $2/L$ (mixing
the two conventions would introduce a spurious $L/2$ factor in the severity;
localization, being scale-invariant, is unaffected). The generator
\path{ff_injection_loss}
admits the temporal shift of \eqref{eq:fallo-apertura}. Finally,
\path{classify_fault} is applied after the physical consistency test
\eqref{eq:consistencia} (\path{physical_consistency_test}) and distinguishes
between injection loss and the shunt subtypes (resistive, capacitive, and
inductive). The controller fault is not labeled autonomously: being formally
indistinguishable from a local shunt (\cref{eq:controlador-Yf}), it is subsumed
into that class, and separating it requires additional telemetry. The site
matching of \cref{obs:fallos-sitios} is implemented by
\path{nearest_converter} in the module \texttt{geolaplace.converters}.
\end{observacion}

\section{Numerical results}
\label{sec:resultados}

This section presents the results obtained by applying the framework
developed in the previous sections to a concrete case study of a
distributed-converter line. All numerical values come from the direct
execution of the simulation and analysis code; the scripts that generate
them are indicated in \cref{obs:reproducibilidad}. The goal is to
illustrate, in a verifiable way, the modal behavior, the fault
localization, and its classification under controlled conditions,
including scenarios with noise.

\subsection{Base case and distributed parameters}
\label{subsec:caso-base}

We consider a line of length $L=100\unit{m}$ with the distributed parameters
collected in \cref{tab:caso-base}. The nominal distributed conductance is
negative, $g=-0.05\unit{S/m}$, which models the active (energy-injecting)
character of the converters and, as we will see, introduces the possibility of
open-loop instability. The parameter $k_v$ is the virtual conductance provided
by the distributed admittance-shaping control described in \cref{sec:control}.

\begin{table}[ht]
  \centering
  \caption{Parameters of the base case of the distributed-converter line.}
  \label{tab:caso-base}
  \begin{tabular}{lll}
    \toprule
    Quantity & Symbol & Value \\
    \midrule
    Line length                     & $L$        & $100\unit{m}$ \\
    Distributed inductance          & $\ell$     & $50\times10^{-6}\unit{H/m}$ \\
    Distributed capacitance         & $c$        & $200\times10^{-6}\unit{F/m}$ \\
    Distributed resistance          & $r$        & $0.02\unit{\Omega/m}$ \\
    Distributed conductance         & $g$        & $-0.05\unit{S/m}$ \\
    Virtual conductance (control)   & $k_v$      & $0.04\unit{S/m}$ \\
    Inductance--capacitance product & $\ell c$ & $1\times10^{-8}$ \\
    \bottomrule
  \end{tabular}
\end{table}

The product $\ell c = 1\times10^{-8}$ sets the temporal scale of the
propagation: through the dispersion relation of \cref{sec:linea}, it is the
coefficient of the second-order term in $\st$ of the modal polynomial.

\subsection{Modal analysis: first mode}
\label{subsec:modo-1}

With homogeneous Dirichlet-type boundary conditions, the admissible wavenumbers
are $k_n = n\pi/L$. For the first mode,
\begin{equation}
  k_1 = \frac{\pi}{L} = 0.031416\unit{rad/m}, \qquad
  k_1^2 = 9.8696\times10^{-4}\unit{rad^2/m^2}.
  \label{eq:k1-resultados}
\end{equation}
Substituting into the dispersion relation specialized to the mode $n=1$ yields
the characteristic polynomial in $\st$
\begin{equation}
  D_1(\st) = 1.0\times10^{-8}\,\st^2 + 1.5\times10^{-6}\,\st
             - 1.304\times10^{-5},
  \label{eq:D1-resultados}
\end{equation}
whose quadratic coefficient coincides with $\ell c$ and whose constant term is
negative: the latter is a direct consequence of the active conductance $g<0$
and anticipates the presence of a positive real root, that is, an unstable
pole.

\subsection{Open-loop and closed-loop modal poles}
\label{subsec:polos}

\Cref{tab:polos} collects the poles of the first five modes in open loop
(without control) and in closed loop, once the virtual conductance has been
applied, which shifts the effective conductance from $g$ to
$g+k_v=-0.01\unit{S/m}$. In open loop, the first mode presents two \emph{real}
roots, $-158.240$ and $+8.240$; the second lies in the right half-plane
($\Real = +8.240 > 0$) and renders the line unstable. The higher modes appear
as complex-conjugate pairs with common real part $-75.000$, so they are stable,
but the instability of the fundamental mode compromises the whole system. After
control, all modes exhibit real part $-175.000$ and, in particular, the first
mode becomes a stable complex-conjugate pair $-175.000\pm219.251\,\Bt$,
eliminating the root in the right half-plane.

\begin{table}[ht]
  \centering
  \caption{Modal poles (in s$^{-1}$) of the first five modes in open loop
    and in closed loop with virtual conductance ($g\to g+k_v=-0.01$).
    The imaginary part is expressed in the $\Bt$ direction.}
  \label{tab:polos}
  \begin{tabular}{cll}
    \toprule
    Mode & Open loop & Closed loop \\
    \midrule
    $1$ & $-158.240,\ +8.240$ (real) & $-175.000 \pm 219.251\,\Bt$ \\
    $2$ & $-75.000 \pm 537.735\,\Bt$   & $-175.000 \pm 586.651\,\Bt$ \\
    $3$ & $-75.000 \pm 884.669\,\Bt$   & $-175.000 \pm 915.226\,\Bt$ \\
    $4$ & $-75.000 \pm 1213.883\,\Bt$  & $-175.000 \pm 1236.330\,\Bt$ \\
    $5$ & $-75.000 \pm 1536.807\,\Bt$  & $-175.000 \pm 1554.598\,\Bt$ \\
    \bottomrule
  \end{tabular}
\end{table}

\begin{observacion}
\label{obs:estabilizacion}
The stabilization mechanism is transparent from the point of view of the
dispersion relation. The active conductance $g<0$ contributes a positive real
part to the constant term of $D_1$ in \cref{eq:D1-resultados}, which places a
root of the fundamental mode at $\Real = +8.240 > 0$. The admittance-shaping
control (\cref{sec:control}) adds a virtual conductance $k_v=0.04\unit{S/m}$
that does not fully cancel the distributed conductance ($g+k_v=-0.01<0$ is
still negative), but does uniformly shift the real part of all poles down to
$-175.000$, so that the first mode ceases to be a pair of real roots and
becomes a stable complex-conjugate pair. The stabilization is therefore modal
and simultaneous for all the modes considered.
\end{observacion}

\Cref{fig:poles} depicts the open- and closed-loop pole map in the plane
generated by $1$ and $\Bt$, illustrating the crossing of the imaginary axis by
the fundamental mode before control. \Cref{fig:timeresp} shows the associated
modal time response, where the divergence of mode $1$ in open loop contrasts
with the damped response in closed loop. Finally, \cref{fig:vsurface} depicts
the voltage surface $v(t,x)$ reconstructed in the presence of a localized
fault, which serves as a link to the following subsection.

\begin{figure}[ht]
  \centering
  \includegraphics[width=0.78\textwidth]{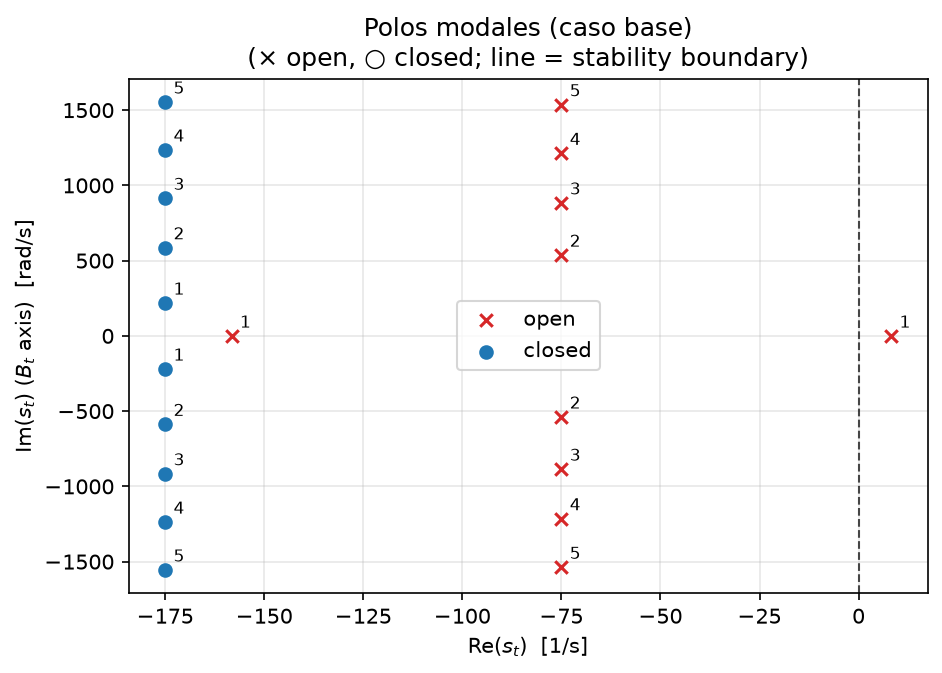}
  \caption{Modal poles of the first five modes in open loop
    (unstable: mode $1$ presents a root with $\Real=+8.240$) and in closed
    loop with virtual conductance (all with $\Real=-175.000$).}
  \label{fig:poles}
\end{figure}

\begin{figure}[ht]
  \centering
  \includegraphics[width=0.78\textwidth]{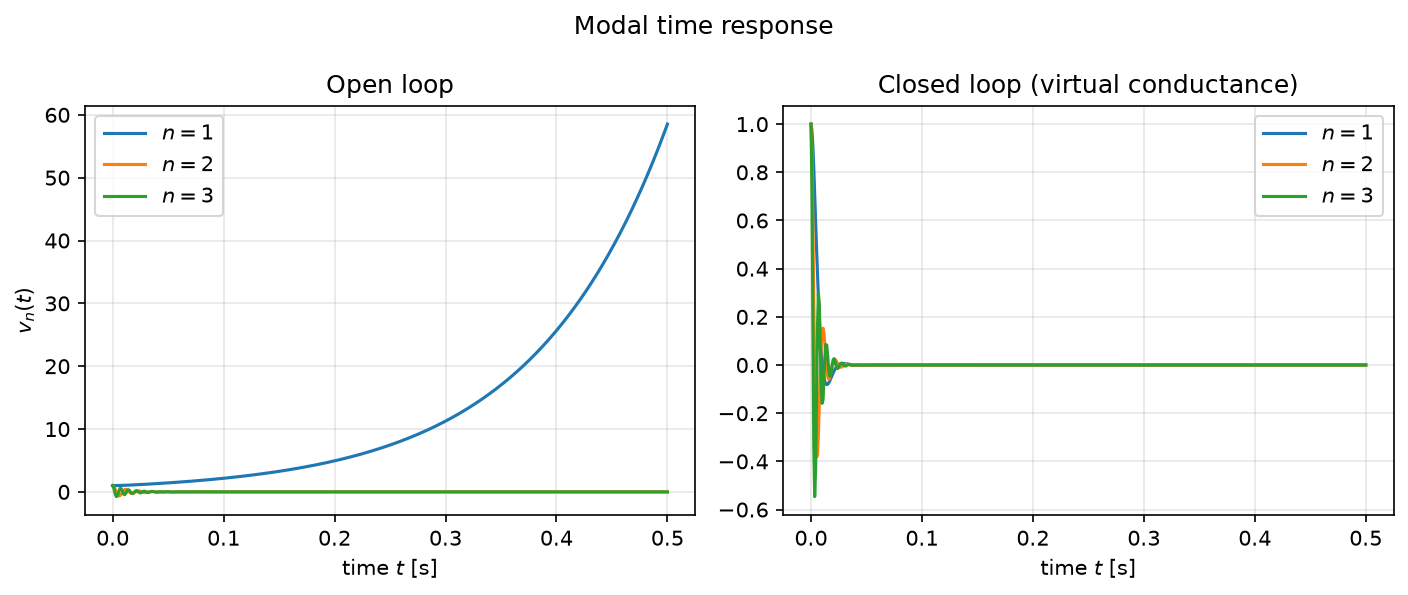}
  \caption{Modal time response in open loop and in closed loop. The
    divergence of the fundamental mode in open loop disappears after the
    action of the distributed control.}
  \label{fig:timeresp}
\end{figure}

\begin{figure}[ht]
  \centering
  \includegraphics[width=0.78\textwidth]{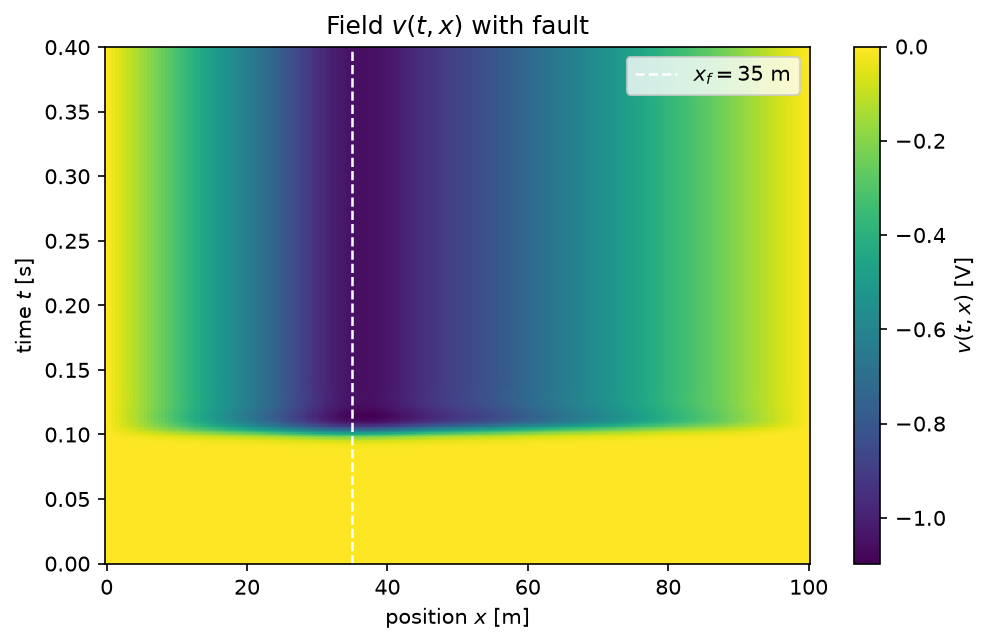}
  \caption{Voltage surface $v(t,x)$ reconstructed with the first $M=12$
    spatial modes $\sin(k_n x)$ under an \emph{injection loss}
    ($f_f(t)=-I_0\,H(t-t_0)$, $I_0=2$) localized at $x_f=35\unit{m}$ and
    activated at $t_0=0.1\unit{s}$, with the line in \emph{closed loop} under
    virtual-conductance control $k_v=0.04\unit{S/m}$. The modal forcing
    incorporates the numerator $\Zimp(\st)$ of the transfer function
    $G_n=\Zimp/D_n$ (\cref{eq:Gn}), that is,
    $\Zimp(\pt)\,u_n=r\,u_n+\ell\,\pt u_n$.}
  \label{fig:vsurface}
\end{figure}

\subsection{Localization of a fault}
\label{subsec:localizacion}

An injection-loss fault is introduced at position $x_f=35\unit{m}$, modeled by
the signature $F_f=-2$ evaluated at $\st=1$, in accordance with the
localized-fault formulation of \cref{sec:fallos}. The two estimators described
there are compared: the phase estimator, based on the phase ratio of the
residual in the $\Bx$ direction, and the modal estimator, based on least
squares over the modal decomposition.

In the absence of noise, both estimators recover the position exactly:
\begin{equation}
  \hat x_{\mathrm{modal}} = 35.00\unit{m}\ (\text{error }0.00\unit{m},\ \text{confidence }1.000),
  \qquad
  \hat x_{\mathrm{fase}} = 35.00\unit{m}\ (\text{error }0.00\unit{m}),
  \label{eq:loc-sinruido}
\end{equation}
with \emph{true} physical consistency. When measurement noise of standard
deviation $\sigma=0.05$ (in the units of the residual) is added, the estimators
degrade in a controlled manner:
\begin{align}
  \hat x_{\mathrm{modal}} &= 34.99\unit{m} && (\text{error }0.008\unit{m}),
  \label{eq:loc-modal-ruido}\\
  \hat x_{\mathrm{fase}}  &= 35.11\unit{m} && (\text{error }0.11\unit{m}).
  \label{eq:loc-fase-ruido}
\end{align}
In this particular realization (fixed seed) the modal estimator turns out to be
more accurate than the phase one, with a notably smaller error, which is
consistent with the fact that the former integrates the information of several
modes in a least-squares fit. This figure corresponds, however, to a single
noise realization and should not be read as a general advantage: over a sweep
of seeds, the mean error of the modal estimator ($\sim0.02\unit{m}$) and that
of the phase estimator ($\sim0.05\unit{m}$) differ only by a factor of order
$2$, and the modal one is more accurate in roughly two thirds of the cases, not
systematically. It is therefore inadvisable to draw from a single experiment a
general conclusion about the superiority of one estimator over the other; both
recover the position with error below $0.2\unit{m}$ over a $100\unit{m}$ line.

\Cref{fig:phasecost,fig:modalcost} show the one-dimensional cost functions of
both estimators as a function of the candidate position, whose minimum
localizes the fault.

\begin{figure}[ht]
  \centering
  \includegraphics[width=0.70\textwidth]{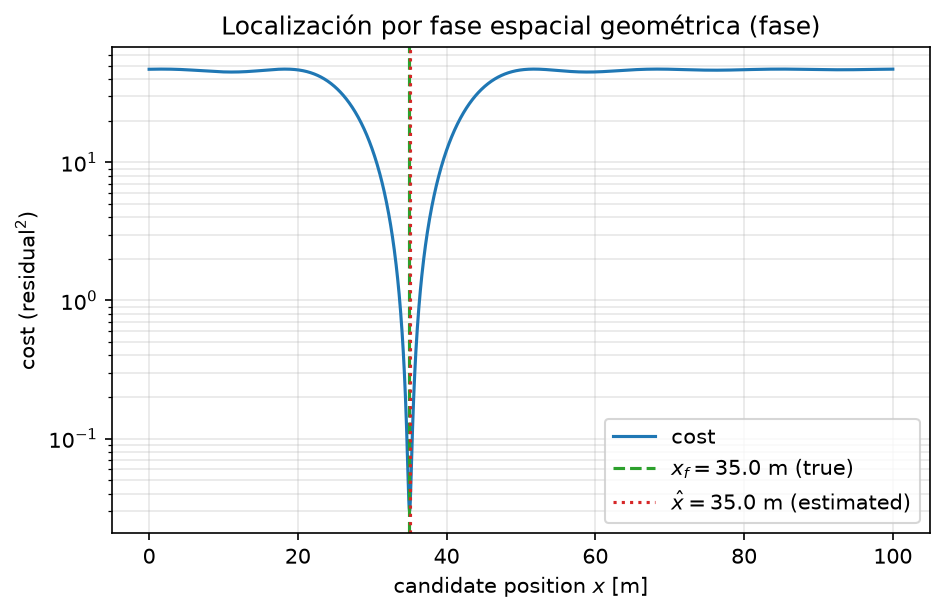}
  \caption{Cost function of the phase estimator versus the candidate
    position $x$. The minimum indicates the estimated position of the fault.}
  \label{fig:phasecost}
\end{figure}

\begin{figure}[ht]
  \centering
  \includegraphics[width=0.70\textwidth]{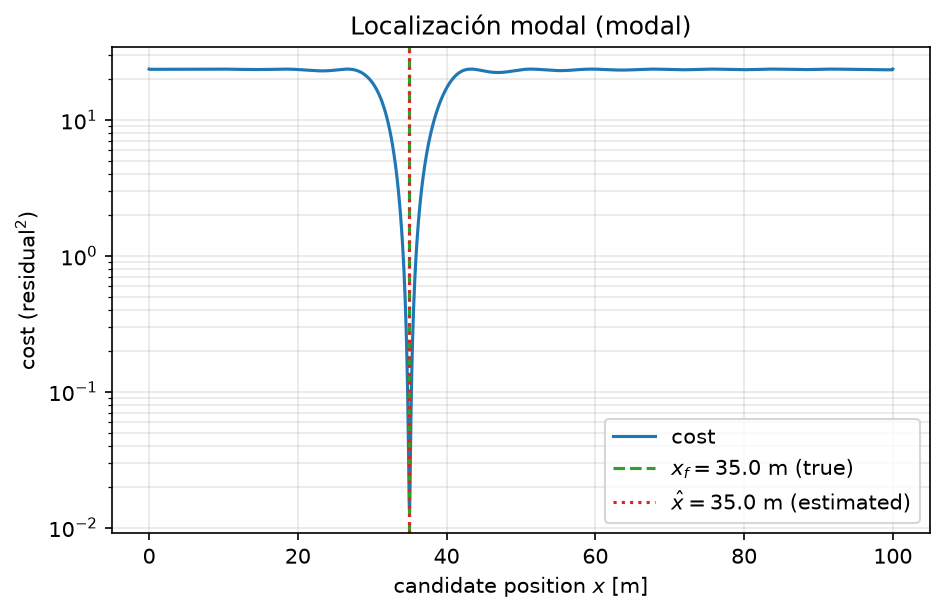}
  \caption{Cost function of the modal estimator versus the candidate
    position $x$. The minimum indicates the estimated position of the fault.}
  \label{fig:modalcost}
\end{figure}

\subsection{Fault classification}
\label{subsec:clasif-result}

Finally, we evaluate the fault classification based on the spectral signatures
of \cref{sec:fallos}, for three cases with a fault at $x_f=35\unit{m}$ and
noise of standard deviation $\sigma=0.02$. It is important to emphasize that
the classification is performed from \emph{signatures estimated from the
measurements}: for a set of temporal frequencies $s\in\{0.5,1,2,5,10\}$ the
modal amplitudes $A_n(s)$ are measured (with noise), $\widehat F_f(s)$ is
estimated by modal severity, and the local voltage $\widehat V_f(s)$ is
measured; the decision uses those estimates, not the true generating functions.
The local voltage is modeled with a non-constant spectrum $V_f(s)=V_0/(1+\tau
s)$, which breaks the degeneracy between injection loss and inductive shunt
(both $\propto 1/s$). The decision metric is the relative dispersion of the fit
to each signature model (the smaller, the better) together with the
cross-checked physical consistency test. The results are summarized in
\cref{tab:clasificacion}.

\begin{description}[leftmargin=1.5em,style=nextline]
  \item[A. Injection loss.] The classifier assigns the label
    \texttt{injection\_loss}, with physical consistency. The relative dispersion
    of the \texttt{injection\_loss} model ($0.025$) is the smallest and is
    clearly separated from the alternatives
    (\path{shunt_inductive} $=0.293$,
    \path{shunt_resistive} $=0.802$,
    \path{shunt_capacitive} $=1.384$); the nonzero residual reflects that the
    fit is performed on noisy estimates, not on the ideal signature.
  \item[B. Resistive shunt.] The label \texttt{shunt\_resistive} is assigned,
    again with physical consistency. The relative-dispersion scores of the four
    candidate models are
    \path{shunt_resistive} $=0.098$,
    \path{injection_loss} $=0.770$,
    \path{shunt_capacitive} $=0.918$ and
    \path{shunt_inductive} $=1.015$,
    which clearly separates the correct hypothesis from the alternatives. This
    separation between signatures is robust against noise; it is worth noting,
    however, that the \emph{physical consistency test} of this case operates with
    a narrow margin relative to the threshold ($\sim0.21$ against the tolerance
    $0.25$): unlike case~A, whose margin is comfortable ($\sim0.02$), in a
    fraction of the noise realizations (of the order of $12\%$ over a sweep of
    seeds) a legitimate resistive shunt may be rejected by the consistency gate
    and relabeled as a sensor anomaly. The signature label is robust; the safety
    margin of the \emph{Physical consistency} column of
    \cref{tab:clasificacion}, by contrast, is not the same in A and in B.
  \item[C. Sensor fault.] The label \texttt{sensor\_fault} is assigned.
    In this case the cross-checked physical consistency is \emph{false}: the
    relative error between the measured phase channel and that predicted by the
    modal fit is $1.006$, above the admissible threshold, the modal confidence
    is low ($0.184$), and the estimated position is spurious
    ($\hat x = 82.2\unit{m}$). The diagnosis itself thus reveals that the fit
    does not correspond to a localized physical fault, but rather to a sensor
    anomaly.
\end{description}

\begin{table}[ht]
  \centering
  \caption{Fault classification at $x_f=35\unit{m}$ with noise
    $\sigma=0.02$. The relative dispersion is smaller the better the fit
    to the corresponding signature model.}
  \label{tab:clasificacion}
  \begin{tabular}{lllc}
    \toprule
    Case & Assigned label & Physical consistency & Key indicator \\
    \midrule
    A. Injection loss & \texttt{injection\_loss} & True
      & dispersion $0.025$ \\
    B. Resistive shunt      & \texttt{shunt\_resistive} & True
      & dispersion $0.098$ \\
    C. Sensor fault      & \texttt{sensor\_fault} & False
      & rel. error $1.006$, conf. $0.184$ \\
    \bottomrule
  \end{tabular}
\end{table}

\Cref{fig:classsig} depicts the spectral signatures of the fault models
considered, which constitute the basis of the classifier's decision.

\begin{figure}[ht]
  \centering
  \includegraphics[width=0.80\textwidth]{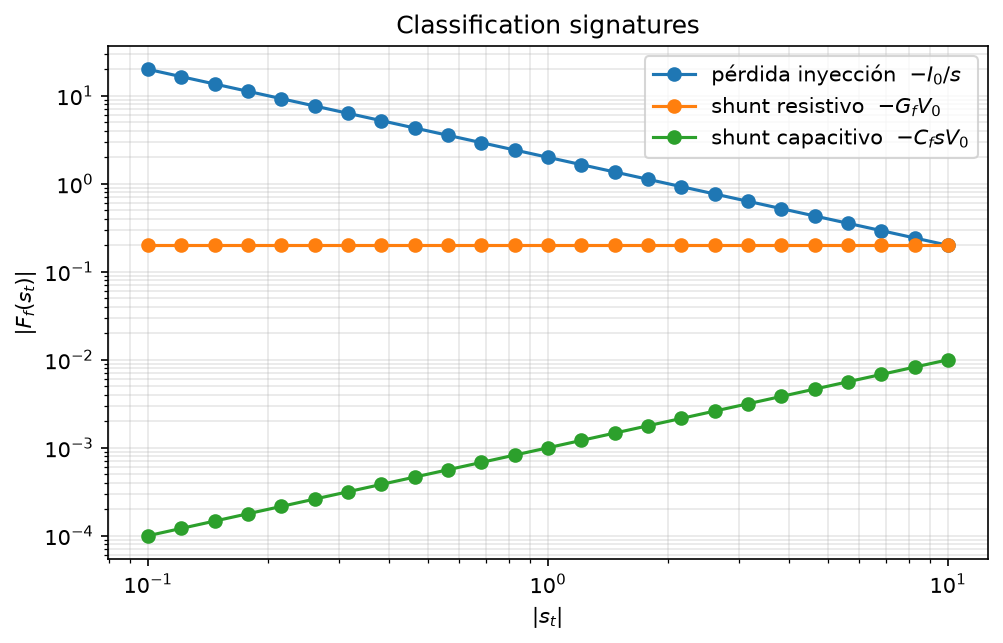}
  \caption{Spectral signatures of three of the fault models employed by the
    classifier (injection loss, resistive shunt, and capacitive shunt). The
    inductive shunt ($\propto 1/s$) is omitted because it overlaps with the
    injection loss, also $\propto 1/s$ (the separation between the two requires
    $V_f(s_t)$ non-constant, not the signature of $F_f$).}
  \label{fig:classsig}
\end{figure}

\begin{observacion}
\label{obs:caso-C}
Case C illustrates a desirable property of the scheme: the physical consistency
test acts as a safeguard mechanism. Although the fitting procedure always
returns a label and a position $\hat x$, the combination of a high relative
error ($1.006$), a low modal confidence ($0.184$), and a position incompatible
with that of the actual fault ($82.2\unit{m}$ against $35\unit{m}$) makes it
possible to reject the localized-physical-fault hypothesis. In this way, an
instrumentation anomaly is not confused with a fault of the line. In the
synthetic model, the sensor anomaly corrupts both the modal channel and the
phase channel (it does not respect the spatial propagation), so that no position
estimator is reliable; the diagnosis relies on the consistency test and on the
low modal confidence. The scope of this test should be clear: here only
the \emph{maximally incoherent} anomaly regime is evaluated (noise that does not
respect propagation), where rejection is essentially guaranteed. A
\emph{structured} sensor bias, capable of partially imitating the spatial
signature $e^{-\sx x_f}$, would constitute a borderline case not covered by this
experiment and remains as future work on the robustness of the scheme.
\end{observacion}

\subsection{Monte Carlo validation of the classification}
\label{subsec:montecarlo}

To avoid conclusions based on a single noise realization, the classification is
evaluated with a reproducible Monte Carlo experiment: $N_{\mathrm{MC}}=100$
realizations per class ($\sigma=0.02$, $x_f=35\unit{m}$, deterministic seeds),
with true classes \texttt{injection\_loss}, \texttt{shunt\_resistive},
\texttt{shunt\_capacitive} and \texttt{sensor\_fault}. The severities are chosen
so that the magnitude of the signature is comparable across classes
(homogeneous SNR): thus the confusion matrix reflects the \emph{discriminative}
capability of the classifier and not a mere signal-to-noise effect (with
$C_f=10^{-3}$ the capacitive signature at $s=1$ falls below the noise and would
be systematically rejected; that result would measure SNR, not discrimination).

The results (\cref{tab:montecarlo,fig:confusion}) show an overall accuracy of
$0.920$ over the $400$ realizations. When a physical fault passes the
consistency test, it is \emph{never} confused with another physical type
(precision $1.000$ for \texttt{injection\_loss}, \texttt{shunt\_resistive} and
\texttt{shunt\_capacitive}). The dominant error effect is the \emph{fragility of
the consistency gate} under noise, already noted in \cref{subsec:clasif-result}:
a fraction of the legitimate shunts ($15\%$ of the resistive, $17\%$ of the
capacitive) is rejected as \texttt{sensor\_fault}. By construction, the sensor
anomaly is detected with recall $1.000$, but its precision ($0.758$) is reduced
precisely by those rejections of genuine physical faults. Overall, the scheme is
\emph{conservative}: it prefers to reject a noisy physical fault rather than emit
an erroneous physical label.

\begin{table}[ht]
  \centering
  \small
  \caption{Confusion matrix of the fault classification over $N_{\mathrm{MC}}=100$ noise realizations per class ($\sigma=0.02$, $x_f=35$\,m). Overall accuracy $0.920$. Column labels abbreviate the predicted classes (IL = \texttt{injection\_loss}, SR = \texttt{shunt\_resistive}, SC = \texttt{shunt\_capacitive}, SI = \texttt{shunt\_inductive}, SF = \texttt{sensor\_fault}); column SF also collects genuine faults rejected by the physical-consistency test.}
  \label{tab:montecarlo}
  \begin{tabular}{lccccc}
    \toprule
    true $\backslash$ predicted & IL & SR & SC & SI & SF \\
    \midrule
    \texttt{injection\_loss} & 100 & 0 & 0 & 0 & 0 \\
    \texttt{shunt\_resistive} & 0 & 85 & 0 & 0 & 15 \\
    \texttt{shunt\_capacitive} & 0 & 0 & 83 & 0 & 17 \\
    \texttt{sensor\_fault} & 0 & 0 & 0 & 0 & 100 \\
    \bottomrule
  \end{tabular}

  \par\medskip
  \begin{tabular}{lccc}
    \toprule
    Class & Recall & Precision & Rejection rate \\
    \midrule
    \texttt{injection\_loss} & 1.000 & 1.000 & 0.000 \\
    \texttt{shunt\_resistive} & 0.850 & 1.000 & 0.150 \\
    \texttt{shunt\_capacitive} & 0.830 & 1.000 & 0.170 \\
    \texttt{sensor\_fault} & 1.000 & 0.758 & 1.000 \\
    \bottomrule
  \end{tabular}
\end{table}

\begin{figure}[ht]
  \centering
  \includegraphics[width=0.78\textwidth]{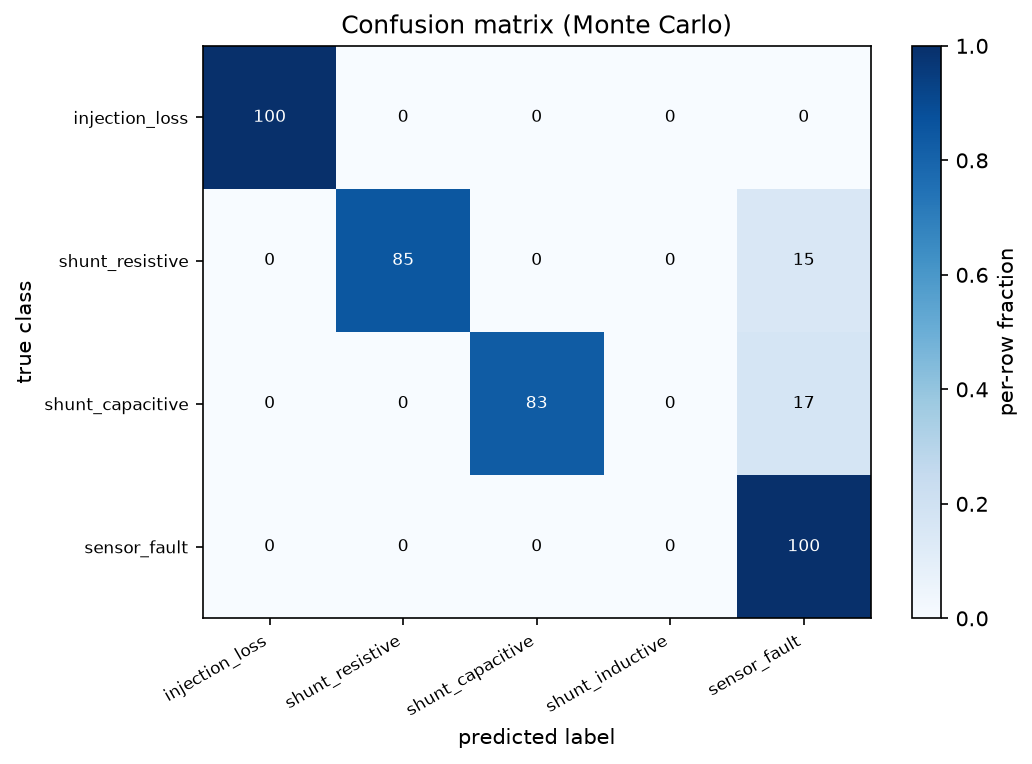}
  \caption{Confusion matrix of the fault classification
    ($N_{\mathrm{MC}}=100$ per class, $\sigma=0.02$, $x_f=35\unit{m}$). The color
    indicates the fraction per row (recall); the cells, the absolute count.
    The off-diagonal mass concentrates in the
    \texttt{sensor\_fault} column: legitimate shunts rejected by the
    consistency test under noise, not confusions between physical types.}
  \label{fig:confusion}
\end{figure}

\subsection{Recovery of multiple faults}
\label{subsec:multifault-result}

To illustrate the sparse recovery of \cref{subsec:fallos-multiples} a modal
measurement of \emph{two} simultaneous faults is synthesized at
$x_1=25\unit{m}$ and $x_2=70\unit{m}$, with amplitudes
$\|F_1\|=1.08$ and $\|F_2\|=0.67$, using $N=24$ modes and measurement noise
$\sigma=0.02$ (fixed seed). Simultaneous OMP, given the number of faults,
recovers the positions with error of $0.005\unit{m}$ and $0.033\unit{m}$
respectively, and amplitudes $\|\widehat F_1\|=1.08$, $\|\widehat F_2\|=0.67$
that are practically exact. The convex group-LASSO relaxation ($\lambda=0.6$)
recovers the peaks at $24.98\unit{m}$ and $70.04\unit{m}$, with amplitudes
slightly shrunk by the bias characteristic of the $\ell_1$ penalty.
\Cref{fig:multifault} shows both results over the group-LASSO support profile.

\begin{figure}[ht]
  \centering
  \includegraphics[width=0.78\textwidth]{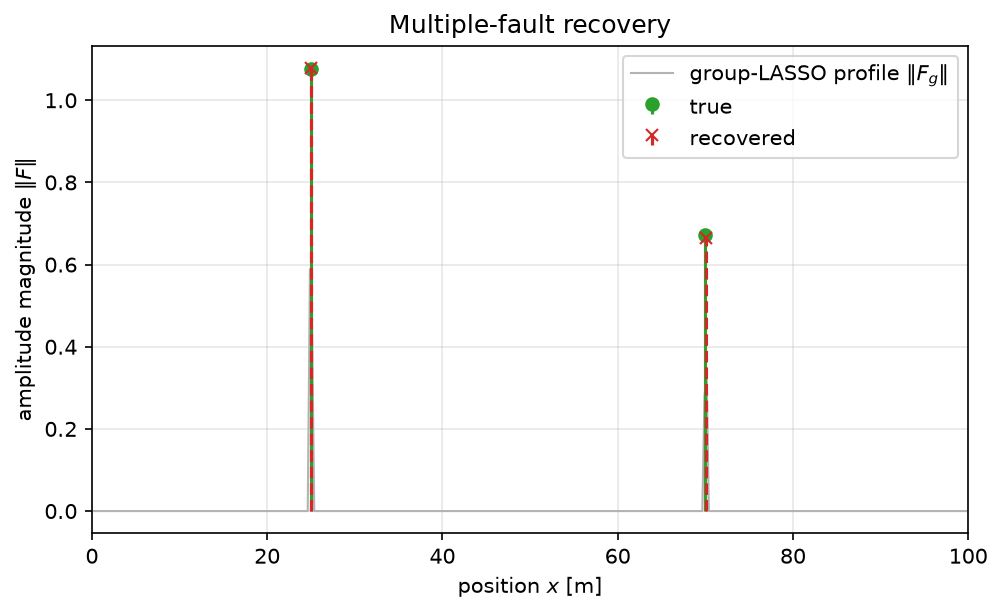}
  \caption{Sparse recovery of two simultaneous faults ($x_1=25\unit{m}$,
    $x_2=70\unit{m}$; $N=24$ modes, $\sigma=0.02$). Markers: true amplitude
    $\|F_j\|$; stems: positions and amplitudes recovered by simultaneous OMP;
    gray curve: support profile $\|F_g\|$ of the group-LASSO. The
    estimates coincide with the true faults within the noise. Separability
    requires sufficient modal observability
    (\cref{obs:observabilidad}): faults very close relative to the spatial
    resolution of the observed modes are not distinguishable.}
  \label{fig:multifault}
\end{figure}

\paragraph{Close faults: resolution study.}
To characterize the close-fault regime (\cref{subsec:fallos-multiples}) a
resolution study is carried out: two faults centered at $x=50\unit{m}$ with
increasing separation $\Delta x$, $N=24$ modes, noise $\sigma=0.05$ and $40$
realizations per point. The success rate is measured (both positions recovered
within $\Delta x/2$) for off-grid refinement (OMP+VARPRO) and for gridless
spectral estimation (ESPRIT). The approximate Rayleigh resolution is
$L/N\approx4.2\unit{m}$. Both methods attain $100\%$ success down to
$\Delta x=2\unit{m}$ ---half the Rayleigh limit, which confirms
\emph{super-resolution} capability---; below that, OMP+VARPRO remains somewhat
more robust ($100\%$ at $1.5\unit{m}$, $78\%$ at $1\unit{m}$) than ESPRIT ($82\%$
and $20\%$ respectively). \Cref{fig:superres} summarizes the study. Honestly,
both methods fail when the separation drops well below the modal resolution: the
limit is intrinsic to the number of observed modes and to the noise, not to the
algorithm.

\begin{figure}[ht]
  \centering
  \includegraphics[width=0.72\textwidth]{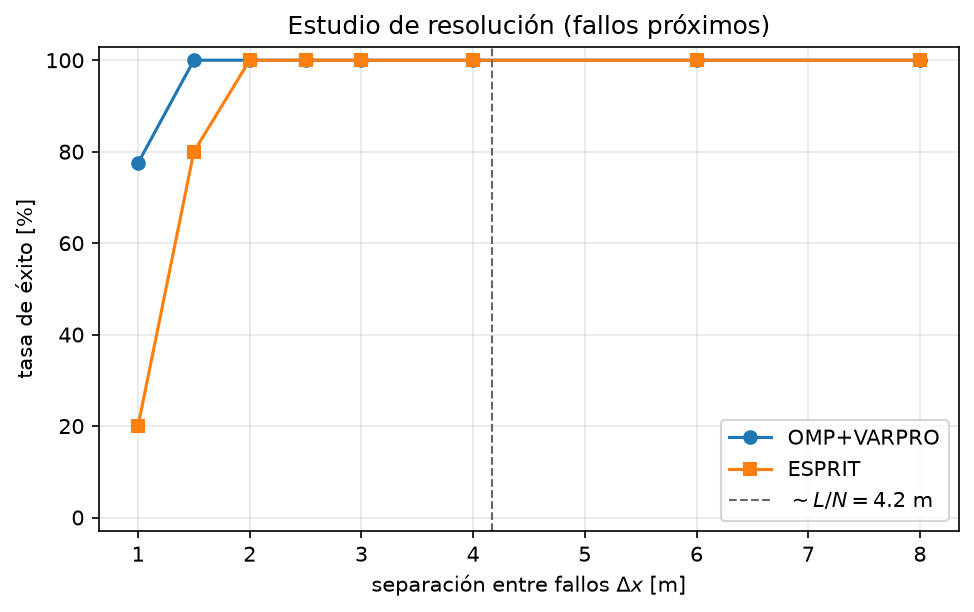}
  \caption{Resolution study of close faults: success rate of the
    recovery versus the separation $\Delta x$ between two faults ($N=24$
    modes, $\sigma=0.05$, $40$ realizations per point). Off-grid
    refinement (OMP+VARPRO) and gridless spectral estimation (ESPRIT) reliably
    resolve separations well below the Rayleigh limit
    $L/N\approx4.2\unit{m}$ (dashed line); OMP+VARPRO is more robust in the
    most demanding regime. The drop at very small separations reflects the
    intrinsic resolution limit, not a defect of the method.}
  \label{fig:superres}
\end{figure}

\paragraph{Automatic estimation of the number of faults.}
The MDL order selection (\cref{eq:mdl}) is evaluated over $J=0,1,2,3$ well-separated
faults ($\Delta x\ge 8\unit{m}$, positions in $[10,90]\unit{m}$), with
$N=24$ modes, pencil size $L_p=12$, $\texttt{max\_faults}=4$ and $100$
realizations per cell, for \emph{two} noise models: \texttt{scalar} (noise
only in the scalar component of $A_n$) and \texttt{four\_components} (independent
noise in the four components of the algebra). \Cref{tab:orden} (and
\cref{fig:orden-acc}) collect the hit rate $\widehat J=J$; the complete metadata
are stored in \texttt{results/data/order\_selection\_metadata.json}. Over both
noise models and $\sigma\in[0.01,0.05]$, the hit rate is $97$--$100\%$ for
$J=0$, $82$--$98\%$ for $J=1$ and $72$--$95\%$ for $J=2$; the case $J=3$ is
the most demanding ($62$--$90\%$) because the signal rank $2J=6$ consumes a
significant part of the observable subspace with $N=24$ modes. In the full
pipeline, passing \texttt{n\_faults=None} to the recoverers activates this
estimation transparently.

The table must \emph{not} be interpreted as a universal guarantee. In particular,
when the spatial separation is reduced or the deeply sub-Rayleigh regime is
entered, the order selection becomes ambiguous (\cref{obs:orden-alcance}); the
figures correspond to the well-separated-fault scenario described above.

\begin{table}[ht]
  \centering
  \caption{Exact-recovery rate of the automatic MDL selection of the number of faults ($\widehat J = J$) over 100 realizations per cell. $N=24$ modes, $\texttt{max\_faults}=4$, pencil size $L_p=12$, well-separated faults ($\Delta x \ge 8$\,m). Two noise models: \texttt{scalar} (noise only in the scalar component) and \texttt{four\_components} (independent noise in the four components of the algebra).}
  \label{tab:orden}
  \begin{tabular}{llcccc}
    \toprule
    Noise & $\sigma$ & $J=0$ & $J=1$ & $J=2$ & $J=3$ \\
    \midrule
    \texttt{scalar} & $0.01$ & 97\% & 92\% & 95\% & 75\% \\
     & $0.03$ & 97\% & 92\% & 89\% & 70\% \\
     & $0.05$ & 97\% & 89\% & 82\% & 65\% \\
    \midrule
    \texttt{four\_components} & $0.01$ & 100\% & 98\% & 94\% & 90\% \\
     & $0.03$ & 100\% & 92\% & 83\% & 78\% \\
     & $0.05$ & 100\% & 82\% & 72\% & 62\% \\
    \bottomrule
  \end{tabular}
\end{table}

\begin{figure}[ht]
  \centering
  \includegraphics[width=0.66\textwidth]{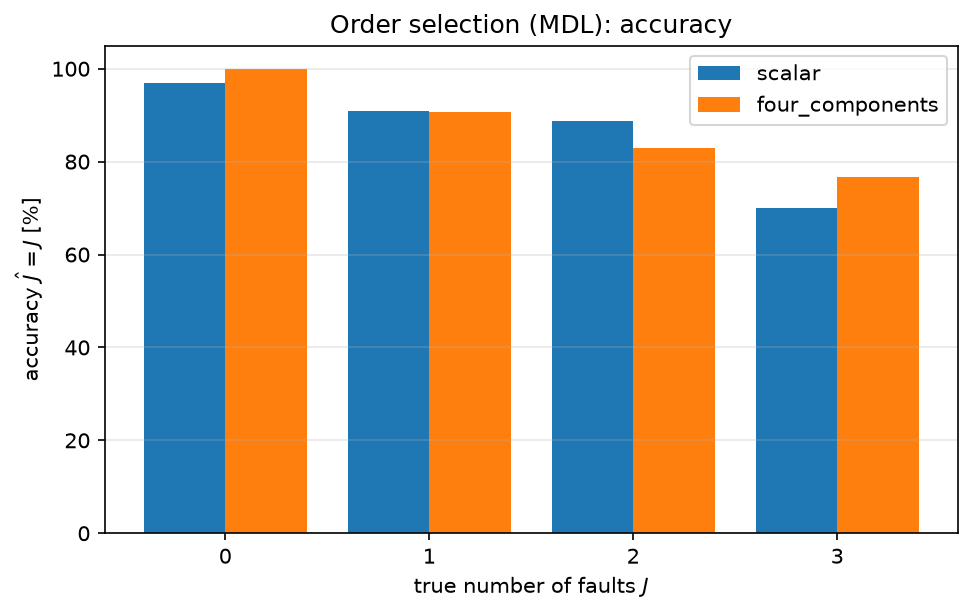}
  \caption{Hit rate of the automatic order selection ($\widehat J=J$) by
    number of faults $J$, averaged over $\sigma\in\{0.01,0.03,0.05\}$, for the
    two noise models. Reliability decreases with $J$ because the signal rank
    $2J$ consumes a growing fraction of the observable subspace with $N=24$
    modes.}
  \label{fig:orden-acc}
\end{figure}

\subsection{Discrete converters: exactness, aliasing, and localization}
\label{subsec:discreto-result}

The discrete-converter theory of \cref{subsec:inyeccion-discreta} and the
exact droop realization of \cref{prop:control-droop-discreto} are validated
on the base line of \cref{tab:caso-base} with $P=8$ plants on the uniform
midpoint grid, $x_p=(p-\tfrac12)L/P\in\{6.25,\allowbreak 18.75,\allowbreak
31.25,\allowbreak 43.75,\allowbreak 56.25,\allowbreak 68.75,\allowbreak
81.25,\allowbreak 93.75\}\unit{m}$ (inter-site spacing $L/P=12.5\unit{m}$),
virtual conductance $k_v=0.04\unit{S/m}$, and $N=12$ retained modes.
\Cref{fig:discreto} summarizes the two quantitative statements.

Panel~(a) examines the continuum limit of \cref{cor:limite-continuo}. For
the smooth density $\bar u(x)=1+0.5\sin(2\pi x/L)+0.3\,(x/L)^2$ and
cell-share currents $I_p=\bar u(x_p)\,L/P$, the coefficient error
$|U_n^{\mathrm{disc}}-U_n^{\mathrm{cont}}|$ decreases with the expected
midpoint-rule rate over $P\in\{4,8,16,32\}$: for mode $n=1$ it falls from
$3.82\times10^{-2}$ to $5.88\times10^{-4}$, with observed dyadic rates
$2.016$, $2.004$, $2.001$; modes $n=3$ and $n=5$ start further from the
asymptotic regime (rates $2.19$ and $2.63$ at the first step) but settle to
$2.011$ and $2.031$ at the last one, consistent with the $O(P^{-2})$ bound.

Panel~(b) examines the exactness of the discrete droop under two
protocols. With the per-converter gains $\kappa_p=k_vL/P$ of
\eqref{eq:droop-discreto}, the continuous closed-loop poles of
\cref{subsec:polos} are compared, for each mode $n$, with the spectrum of
the coupled discrete loop (\path{discrete_droop_poles}) computed with the
retained band $1,\dots,n$ ---the validity condition of
\cref{prop:control-droop-discreto}: modal content confined to the retained
band. Under this protocol the maximum deviation for the sub-Nyquist modes
$n\le7$ is $7.1\times10^{-13}\unit{s^{-1}}$ on a pole scale of order
$10^{3}\unit{s^{-1}}$, i.e.\ machine precision: the continuous closed-loop
analysis is \emph{exact} there, as \cref{prop:control-droop-discreto}(i)
asserts. Dropping the confinement condition has visible consequences. In
the \emph{fixed} fully coupled
loop that retains all $N=12$ modes at once (second marker set in the
figure), the alias rows $C_{n,2P-n}=k_v$ of
\cref{prop:control-droop-discreto}(ii) couple the sub-Nyquist modes
$4,\dots,7$ with their retained partners $12,\dots,9$ and shift their
poles by $2.0$ to $14.4\unit{s^{-1}}$, while modes $1,\dots,3$, whose
partners $15,\dots,13$ fall outside the retained band, remain at machine
precision. Mode $n=P=8$
sees the doubled gain: its discrete poles, $-275.000\pm2510.164\,\Bt$,
coincide with the continuous poles computed with $2k_v$ to
$4.7\times10^{-13}$, while they differ from the single-$k_v$ prediction by
$100.2\unit{s^{-1}}$ (the factor-2 jump visible at the Nyquist wall $n=P$
in the figure). Beyond the wall the aliased rows couple each mode with its
partner and the deviations are $O(1)$ on the pole scale (from
$18.4\unit{s^{-1}}$ at $n=9$ to $6.0\unit{s^{-1}}$ at $n=12$). Placement
matters as well: on the uniform \emph{node} grid $x_p=pL/(P+1)$ with $P=8$,
the controllability score \eqref{eq:score-controlabilidad} of mode $9$
evaluates to $\gamma_9=1.3\times10^{-15}$, numerically zero, i.e.\ mode
$P+1$ is exactly uncontrollable from that converter set, as predicted.
Non-uniform placements behave as anticipated in the observation following
\cref{prop:control-droop-discreto}: with $P$ sites drawn uniformly at
random and the Voronoi-weighted gains $\kappa_p=k_vw_p$, the coupled poles
of the low modes $n\le4$ deviate from the continuous prediction by
$32.8\unit{s^{-1}}$ at $P=8$ but only $3.7\times10^{-2}\unit{s^{-1}}$ at
$P=256$ (one random draw per $P$, fixed seed): the exactness is lost, yet
the convergence to the continuous loop in the quadrature sense of
\cref{cor:limite-continuo} is preserved.

\begin{figure}[ht]
  \centering
  \includegraphics[width=0.9\textwidth]{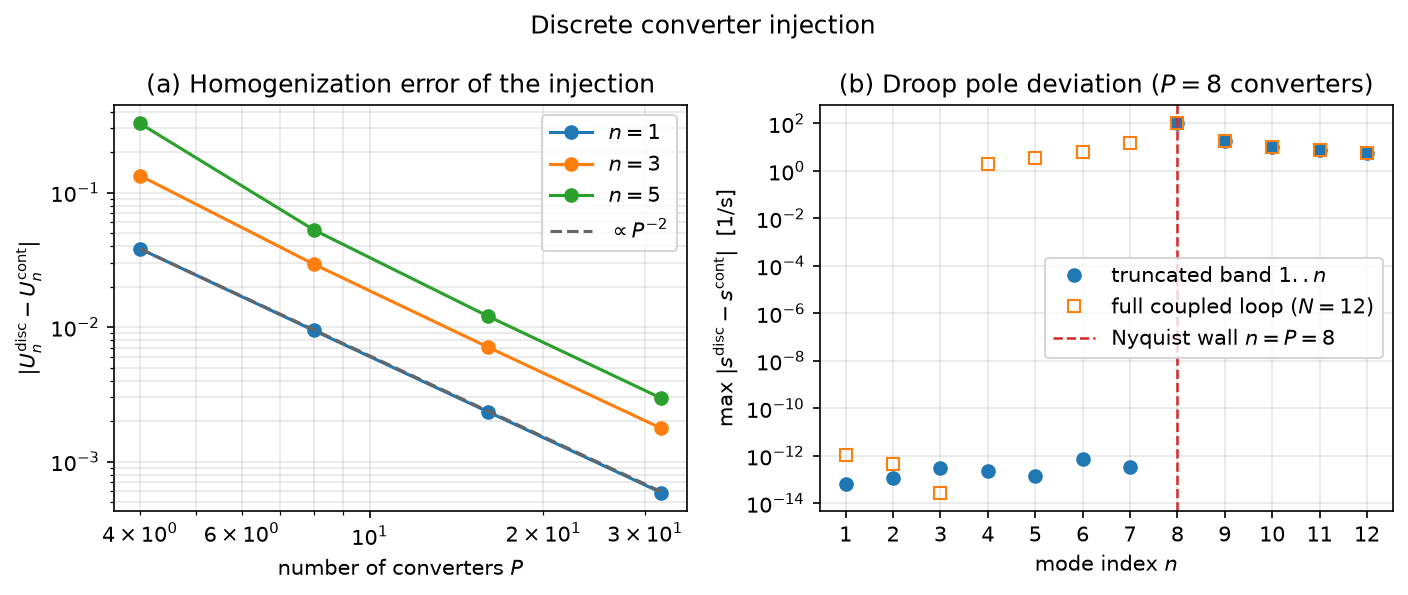}
  \caption{Discrete-converter study ($L=100\unit{m}$, midpoint grid).
    (a)~Modal-injection error $|U_n^{\mathrm{disc}}-U_n^{\mathrm{cont}}|$
    versus the number of plants $P$ for modes $n=1,3,5$ and the smooth
    density $\bar u(x)=1+0.5\sin(2\pi x/L)+0.3(x/L)^2$; the dashed guide
    line has slope $P^{-2}$ (\cref{cor:limite-continuo}).
    (b)~Per-mode maximum pole deviation between the discrete droop loop
    ($P=8$, $\kappa_p=k_vL/P$) and the continuous closed loop, under two
    protocols: retained band $1,\dots,n$ per mode (circles; the validity
    condition of \cref{prop:control-droop-discreto}) and fixed fully
    coupled loop with all $N=12$ modes (open squares). Machine precision
    for $n\le7$ under the first protocol (plot floor $10^{-16}$); under
    the second, the alias rows shift modes $4,\dots,7$ (partners
    $12,\dots,9$) by $2.0$--$14.4\unit{s^{-1}}$ while modes $1,\dots,3$
    remain exact. In both: factor-2 gain jump at the Nyquist wall $n=P=8$
    and $O(1)$ alias coupling for $n>8$.}
  \label{fig:discreto}
\end{figure}

Finally, we exercise the fault diagnosis of \cref{obs:fallos-sitios},
subtracting the known healthy injection from the residual and matching the
estimated position against the sites with the threshold
$\delta=1\unit{m}$, a value between the observed localization uncertainty
and the half-spacing $L/(2P)=6.25\unit{m}$. In the first scenario a cable
fault occurs at $x_f=35\unit{m}$, strictly between the sites
$x_3=31.25\unit{m}$ and $x_4=43.75\unit{m}$. With measurement noise
$\sigma=0.02$ and $N=12$ modes, the modal estimator places it at
$35.003\unit{m}$ and the phase estimator at $34.971\unit{m}$, errors of
$0.003$ and $0.029\unit{m}$. Since the estimated position lies
$3.72\unit{m}$ from the nearest site, well beyond $\delta$, the fault is
correctly attributed to the cable and not to a plant. In the second
scenario plant $3$ suffers an outage, so the residual is $-I_3$ at
$x_3=31.25\unit{m}$. The same subtraction pipeline recovers the site
position with an error of $8.5\times10^{-9}\unit{m}$ in the absence of
noise; with $\sigma=0.02$ it returns $\widehat x_f=31.220\unit{m}$, only
$0.030\unit{m}$ from the site, and \path{nearest_converter} flags the
outage of plant $3$. As an additional end-to-end check we ran the
single-fault experiment of \cref{subsec:clasificacion} with the pure
injection-loss signature at $x_3$ ---by \cref{obs:fallos-sitios}, the
distribution of the subtracted residual---, and the consistency gate and
the classifier label the anomaly \texttt{injection\_loss} with
$\widehat x_f=31.227\unit{m}$. The decision is also stable under noise.
Over $200$ noisy realizations of each scenario, the site rule with
$\delta=1\unit{m}$ attributed the outage to plant $3$ in every trial, and
the mid-span cable fault at $x=37.5\unit{m}$ to the cable in every trial.
The rule remains blind, by construction, to cable faults within $\delta$
of a site ---$16\%$ of the line in this configuration--- and the
separation would degrade for denser converter sets or noisier
measurements.

How much does this diagnosis depend on knowing the dispatch exactly? To
find out, we perturbed the currents assumed by the subtraction with
independent per-plant relative errors of standard deviation $\varepsilon$
and repeated the localization of the cable fault at $x_f=35\unit{m}$
($200$ trials per level, phase channel, $\sigma=0.02$). The median
localization error grows from $0.018\unit{m}$ with exact dispatch to
$0.049\unit{m}$ at $\varepsilon=1\%$ and $0.250\unit{m}$ at
$\varepsilon=5\%$; the corresponding $90$th percentiles are $0.041$,
$0.121$ and $0.593\unit{m}$. The degradation is graceful, and
percent-level dispatch errors leave the localization error two orders of
magnitude below the inter-site spacing. Still, this is the characteristic
failure mode of the residual approach: larger, or correlated, dispatch
errors would eventually erode the outage/cable separation (see the
limitations in \cref{subsec:limitaciones}).

\begin{observacion}
\label{obs:reproducibilidad}
All numerical results, tables and figures of this section are
reproducible; the implementation, the scripts and the data accompanying
this paper are available at
\url{https://github.com/fmarrabal/geo-laplace-distributed-converters}.
The modal analysis and the poles of \cref{tab:polos} are
obtained with \path{scripts/run_modal_analysis.py}; the fault
localization of \cref{subsec:localizacion} with
\path{scripts/run_fault_detection.py}; the classification of
\cref{subsec:clasif-result} with
\path{scripts/run_fault_classification.py}; the Monte Carlo validation of
\cref{subsec:montecarlo} (\cref{tab:montecarlo,fig:confusion}) with
\path{scripts/run_fault_classification_monte_carlo.py}, which regenerates
\cref{tab:montecarlo} in \path{results/data/classification_mc_table.tex}; the
recovery of multiple faults of \cref{subsec:multifault-result} with
\path{scripts/run_multiple_fault_recovery.py}; the resolution study of
close faults (\cref{fig:superres}) with
\path{scripts/run_superresolution_study.py}; the order selection of
\cref{tab:orden,fig:orden-acc} with
\path{scripts/run_order_selection_study.py}, which regenerates
\cref{tab:orden} in \path{results/data/order_selection_table.tex} and the
metadata in \path{order_selection_metadata.json}; the discrete-converter
study of \cref{subsec:discreto-result} (\cref{fig:discreto}) with
\path{scripts/run_discrete_converter_study.py}, which stores the raw
curves in \path{results/data/discrete_converter_study.npz} and the
metadata in \path{discrete_converter_metadata.json}; and all of the
figures with \path{scripts/generate_all_figures.py}.
\end{observacion}

\section{Conclusions, limitations and future work}
\label{sec:conclusiones}

We close by summarizing what the paper establishes, by stating the
assumptions under which those results hold, and by pointing out the
questions they leave open. One observation runs through the whole work and
is worth keeping in sight here: the operational novelty of the geometric
transform appears only when two or more independent variables are involved
---in our case, time $t$ and space $x$--- since the planes $\{1,\Bt\}$ and
$\{1,\Bx\}$ then keep their phases separate and the mixed bivector $\Btx$
encodes the space--time coupling.

\subsection{Conclusions}
\label{subsec:conclusiones}

The work builds a chain that links an algebraic object, an integral
transform and an engineering application. Its main results, read together
with the caveats of \cref{subsec:limitaciones}, are the following.

\begin{enumerate}[leftmargin=1.6em,label=(C\arabic*)]
  \item \textbf{Minimal bicomplex algebra.} We identify a commutative
        subalgebra generated by $\Bt$ and $\Bx$ with $\Bt^2=\Bx^2=-1$,
        $\Bt\Bx=\Bx\Bt=\Btx$ and $\Btx^2=+1$, which constitutes the minimal
        algebraic support for separating temporal and spatial phases on the
        elements $z=a+b\Bt+c\Bx+d\Btx$.

  \item \textbf{Two-dimensional transform with operational properties.} We
        define $\transf{f}=F(\st,\sx)=\int_0^\infty\!\int_0^\infty f(t,x)\,
        e^{-\st t}e^{-\sx x}\dd x\dd t$ and establish its
        operational rules (linearity, region of convergence, behavior
        with respect to $\pt$ and $\px$, and the transform of separable modes),
        so that the partial differential
        equations of the problem become algebraic relations
        in $(\st,\sx)$. As an integral object, the transform is
        known in bicomplex analysis
        \cite{kumar2011bicomplex,agarwal2014convolution} and in the classical
        theory of the double transform \cite{debnath2016double}; the
        contribution is its geometric reformulation
        (\cref{subsec:trabajo-relacionado}).

  \item \textbf{Line and dispersion model.} We formulate the DC line of
        distributed converters as a transmission-line-type system with
        parameters $\Zimp$ and $\Yadm$, obtaining the space--time
        dispersion relation $\Deter(\st,\sx)=0$ that organizes the
        response of the system.

  \item \textbf{Admittance shaping control.} We propose a distributed
        control design that shapes the effective admittance of the line,
        formulated in the transformed domain and therefore expressible as
        conditions on $(\st,\sx)$.

  \item \textbf{Fault modeling as singularities.} Localized faults
        are represented as \emph{spatial} singularities of the distributed
        residual $\Res(t,x)$. In the transformed domain they produce a
        factorization $R(\st,\sx)=F_f(\st)\,e^{-\sx x_f}$: the temporal nature
        of the fault remains in $F_f(\st)$, while its localization remains in the
        spatial rotor $e^{-\sx x_f}$, which is \emph{entire} in $\sx$ (it
        contributes no spatial poles). The poles of the \emph{measured response}
        appear when this residual is filtered by the line dynamics, that is, by
        $\Deterc^{-1}$, not in the pointwise residual itself.

  \item \textbf{Localization by geometric and modal phase.} We propose
        localization procedures based on the geometric spatial phase
        $\ArgBx$ and on modal weights, which estimate the fault position from
        the transformed response.

  \item \textbf{Fault classification.} We distinguish fault categories
        (open circuit or injection loss, shunt or short circuit, sensor
        fault and local controller fault) according to their signature in the
        geometric domain.

  \item \textbf{Reproducible numerical validation.} The whole development
        is accompanied by an implementation and reproducible numerical
        experiments that illustrate, within the assumptions of the model, the
        expected behavior.
\end{enumerate}

We claim no universal superiority over the established techniques of line
analysis or fault detection. The main contribution is the chain
(C5)--(C7): the integrated diagnosis ---detection, localization and
classification--- of faults in a distributed network of generators with
converters, supported by the geometric reading of the spatial phase. What
the two-dimensional geometric representation offers is a unified and
operationally convenient description of the space--time coupling, and that
description translates into interpretable fault signatures. Its
distinctive value requires the effective presence of two independent
variables and a reasonable fulfillment of the assumptions discussed next.

\subsection{Limitations}
\label{subsec:limitaciones}

The model, the control and the detection methods developed here rest on
simplifying assumptions whose violation degrades the guarantees obtained.
We list them explicitly.

\begin{enumerate}[leftmargin=1.6em,label=(L\arabic*)]
  \item \textbf{Nonuniform lines.} The dispersion relation
        $\Deter(\st,\sx)=0$ and the operator $\transf{\cdot}$ have been built
        assuming distributed parameters $\Zimp$ and $\Yadm$ that are homogeneous
        along $x$. In lines whose cross section, materials or loads vary with
        position, the coefficients cease to be constant in $x$ and the
        transform in $\sx$ no longer diagonalizes the problem: a piecewise
        treatment or one with variable coefficients is required, which is not
        addressed here.

  \item \textbf{Converter internal dynamics.} The spatial-continuity
        approximation of the injection is no longer an open limitation: the
        discrete converter sites are treated exactly in
        \cref{subsec:inyeccion-discreta,prop:control-droop-discreto}, where
        the homogenization is quantified ($O(P^{-2})$ for smooth densities)
        and shown to be \emph{exact} for the uniform droop feedback on the
        sub-Nyquist band $n<P$. What remains outside the model is the
        converters' \emph{internal} dynamics: each plant is reduced to its
        incremental injected current $I_p(t)$, abstracting away the current
        and voltage loops, the switching, and the output filters, whose
        bandwidths bound the validity of that reduction
        \cite{meng2017dynamics}. For the same reason, the diagnosis of
        \cref{obs:fallos-sitios} treats a plant outage as the additive
        loss of its feed-forward dispatch; the simultaneous loss of its
        droop feedback term ---a multiplicative admittance perturbation,
        \eqref{eq:controlador-Yf}--- and the closed-loop transient it
        triggers are not simulated in the time domain.

  \item \textbf{Known dispatch in the discrete residual.} The residual of
        \cref{obs:fallos-sitios} subtracts the healthy rotor sum assuming
        the dispatch $I_p(\st)$ and the sites $x_p$ to be exactly known. A
        dispatch error $\Delta I_p$ leaves the spurious residual
        $-\sum_p\Delta I_p\,e^{-\sx x_p}$ competing with the fault
        signature; the Monte Carlo sweep of \cref{subsec:discreto-result}
        shows a graceful degradation (percent-level errors leave the
        localization well below the inter-site spacing), but larger or
        correlated errors erode the separation between plant outages and
        cable faults.

  \item \textbf{Communication delays.} The admittance shaping control assumes
        immediate availability of the signals it requires.
        Communication delays between nodes introduce factors
        $e^{-\st\tau}$ not accounted for in the nominal design, which may reduce
        the space--time stability margins.

  \item \textbf{Noise and sensor scarcity.} Localization by geometric phase
        and by modal weights presupposes sufficient spatial observability.
        With sensors that are scarce relative to the number of relevant spatial
        modes, the reconstruction of the spatial phase $\ArgBx$ becomes
        under-determined and measurement noise limits the resolution attainable in
        the estimation of the fault position. In the discrete-converter
        setting this asymmetry is explicit: the actuation side has an
        exact sampling theory (\cref{prop:muestreo}), while the
        measurements are still taken as ideal modal/transform-domain
        quantities; the measurement dual ---the reconstruction of the
        modal amplitudes $A_n$ from voltage samples at discrete sensor
        sites--- is not developed here.

  \item \textbf{Converter saturation.} Admittance shaping is a linear
        design. Current or voltage saturation of the
        converters introduces nonlinearities that locally invalidate the
        superposition principle on which both the control and the
        interpretation of the distributed residual $\Res$ rely.

  \item \textbf{Multiple and simultaneous faults.} The model extends by
        linearity to $J$ faults, $A_n=\sum_j F_j\psi_n(x_j)$, and its sparse
        recovery is implemented (simultaneous OMP and group-LASSO,
        \cref{subsec:fallos-multiples,subsec:multifault-result}) for
        \emph{well-separated} faults. The case of faults that are very close
        relative to the spatial resolution of the observed modes produces nearly
        collinear columns in the dictionary and remains a hard identification
        problem, sensitive to noise and to observability
        \cite{isermann2006fault}.

  \item \textbf{Commutative treatment.} For operational simplicity we have
        imposed $\Bt\Bx=\Bx\Bt=\Btx$, that is, a \emph{commutative}
        bicomplex algebra. This choice algebraically decouples the physical
        temporal and spatial planes and allows closed manipulations, but it does
        not capture situations in which those planes are intrinsically
        coupled and would require generators that do \emph{not} commute. The
        recent noncommutative geometric Laplace transform of
        \cite{velasco2026geometric} shows concretely what is at stake:
        without commutativity the transform must be defined in left and
        right versions and even the exponential ceases to obey the
        classical rule, whereas the factorized signature
        $F_f(\st)\,e^{-\sx x_f}$ and the phase reading of this work rest
        precisely on the classical behavior of the rotors, which the
        commutative subalgebra preserves. The
        noncommutative case is outside the scope of the present work.

  \item \textbf{Experimental validation pending.} All the evidence provided is
        analytic and numerical. We do not yet have
        experimental validation on a physical plant, so the
        quantitative conclusions must be interpreted as predictions of the
        model and not as facts verified in hardware.
\end{enumerate}

\begin{observacion}[Reduction to the single-variable case]
\label{obs:reduccion-final}
One limitation of principle bears repeating. If $f$ depends
on a single independent variable, its geometric transform coincides, up to the
reinterpretation of the plane $\{1,\Bt\}\cong\CC$, with the classical Laplace
transform. All the distinctive usefulness of the proposed framework lies in the
simultaneous treatment of $\st=\rho_t+\Bt\omega$ and $\sx=\rho_x+\Bx k$, that
is, in the genuinely \emph{multivariable} nature of the
space--time problem.
\end{observacion}

\subsection{Future work}
\label{subsec:trabajo-futuro}

Each of the limitations above points to a line of continuation; the ones
we consider priorities are the following.

\begin{enumerate}[leftmargin=1.6em,label=(F\arabic*)]
  \item \textbf{Noncommutative algebras with coupled planes.} Extend the
        framework to generators that do not commute, so that the physical
        temporal and spatial planes may be intrinsically coupled. This requires
        reformulating the operational properties of $\transf{\cdot}$ and studying
        when the cost of losing commutativity is offset by greater
        physical fidelity. The left/right transform calculus recently
        developed for multivector-valued functions of a single real
        variable \cite{velasco2026geometric} provides a natural starting
        toolbox for that reformulation.

  \item \textbf{Nonuniform and piecewise lines.} Generalize the dispersion
        relation and the transform to variable parameters $\Zimp(x)$ and
        $\Yadm(x)$ and to sectioned structures, by means of
        variable-coefficient or piecewise-concatenation techniques.

  \item \textbf{Converters with internal dynamics.} Replace the ideal
        incremental injection $I_p(t)$ with converter models including
        their internal loops (current and voltage control, switching,
        filters). The continuous-versus-discrete approximation error is now
        quantified (\cref{cor:limite-continuo,prop:control-droop-discreto});
        the open question is how the internal converter bandwidths interact
        with the spatial modes and with the exactness of the discrete droop
        \cite{meng2017dynamics}.

  \item \textbf{Multiple faults and order selection.} The multifault
        recovery is implemented by on-grid and off-grid methods (OMP,
        group-LASSO, VARPRO and ESPRIT), and the automatic selection of the number
        of faults by MDL is validated in well-separated and moderately-separated
        scenarios (\cref{subsec:fallos-multiples,subsec:multifault-result}). Its
        reliability, however, depends on the number of modes, the SNR, the
        separation between faults, the pencil size and the upper bound
        $\texttt{max\_faults}$; at the exact limit it further requires a
        numerical-rank guard. The \emph{deeply} sub-Rayleigh regime, as well as
        order selection with scarce or very noisy data ---and the joint
        estimation of severity and position with guarantees against dictionary
        coherence--- remain as future work \cite{isermann2006fault}.

  \item \textbf{Robustness to noise and delays.} Quantify and improve the
        robustness of localization and classification against measurement noise,
        sensor scarcity and communication delays, incorporating the latter
        explicitly in the control design.

  \item \textbf{Experimental validation.} Contrast the predictions of the
        model on a laboratory DC microgrid, closing the gap between the
        current numerical evidence and hardware verification.
\end{enumerate}

In short, this work is a first articulation, internally consistent and
reproducible, of an idea whose full development still requires lifting the
commutative and uniformity assumptions, and confronting them with
experimental data.

\bibliographystyle{plain}
\bibliography{references}

\end{document}